\documentclass[a4paper,12pt]{article}

\usepackage[T1]{fontenc}      
\usepackage{cite}

\usepackage{color}
\usepackage{amssymb, amsmath}
\usepackage{hyperref}
\usepackage{tikz}
 \usepackage{enumitem}
\usepackage{pdflscape }
\newcommand\EFFACE[1]{}

\newcommand\ESOK[1]{#1}

\newtheorem{theorem}{Theorem}
\newtheorem{proposition}[theorem]{Proposition}

\newtheorem{lemma}[theorem]{Lemma}
\newtheorem{corollary}[theorem]{Corollary}

\newtheorem{question}{Question}

\newtheorem{claim}{Claim}

\newenvironment{proof}{
\par
\noindent {\bf Proof.}\rm}{\mbox{}\hfill$\square$\par\vskip 3mm}

\newcommand\pcn{\chi_{\rho}}

\newcommand\SOMMET[1]{\draw[fill=black] (#1) circle (2pt)}

\newcommand\BIGSOMMET[1]{\draw[fill=black] (#1) circle (3pt)}

\newcommand\ETIQUETTE[2]{\node[below] at (#1) {#2}}
\newcommand\ETIQUETTEA[2]{\node[above] at (#1) {#2}}
\newcommand\ETIQUETTEB[2]{\node[below] at (#1) {#2}}
\newcommand\ARETEH[2]{\draw[thick] (#1) -- ++(#2,0)}
\newcommand\ARETEV[1]{\draw[thick] (#1) -- ++(0,1)}
\newcommand\ARETEVV[2]{\draw[thick] (#1) -- ++(0,#2)}

\makeatletter
\let\@fnsymbol\@arabic
\makeatother

\newcommand\Xsp{\;}

\newcommand\Xud{12}
\newcommand\Xut{13}
\newcommand\Xuq{14}
\newcommand\Xuc{15}
\newcommand\Xus{16}
\newcommand\Xuse{17}

\newcommand\Xdu{21}
\newcommand\Xdt{23}

\newcommand\Xdc{25}

\newcommand\Xtu{31}
\newcommand\Xtd{32}

\newcommand\Xts{36}

\newcommand\Xqu{41}
\newcommand\Xqd{42}
\newcommand\Xqt{43}

\newcommand\Xcu{51}
\newcommand\Xcd{52}

\newcommand\Xcq{54}

\newcommand\Xsu{61}

\newcommand\Xst{63}

\newcommand\Xseu{71}

\begin{document}

\title{{\bf Packing colouring of some classes of cubic graphs}}

\author{Daouya LA\"{I}CHE~\thanks{Faculty of Mathematics, Laboratory L'IFORCE, University of Sciences and Technology
Houari Boumediene (USTHB), B.P.~32 El-Alia, Bab-Ezzouar, 16111 Algiers, Algeria.}
\and \'Eric SOPENA~\thanks{Univ. Bordeaux, Bordeaux INP, CNRS, LaBRI, UMR5800, F-33400 Talence, France.}~$^,$\footnote{Corresponding author. Eric.Sopena@labri.fr.}
}

\maketitle

\abstract{
The packing chromatic number $\pcn(G)$ of a graph $G$ is the smallest integer $k$ such that
its set of vertices $V(G)$ can be partitioned into $k$ disjoint subsets $V_1$, \ldots, $V_k$, in such a way that every two distinct vertices
in $V_i$ are at distance  greater than $i$ in $G$ for every $i$, $1\le i\le k$.

Recently, Balogh, Kostochka and Liu  proved that $\pcn$ is not bounded in the class of subcubic graphs
[Packing chromatic number of subcubic graphs,
{\it Discrete Math.} 341 (2018), 474--483],
thus answering a question previously addressed in several papers.
However, several subclasses of cubic or subcubic graphs have bounded packing chromatic number.
In this paper, we determine the exact value of, or upper and lower bounds on, the packing chromatic number of
some classes of cubic graphs, namely  circular ladders, and so-called H-graphs and generalised H-graphs.
}

\medskip


\medskip


\section{Introduction}

All the graphs we consider are \ESOK{simple}.
For a graph $G$, we denote by $V(G)$ its set of vertices and by $E(G)$ its set of edges.
The {\em distance} $d_G(u,v)$
between
vertices $u$ and $v$ in $G$ is the length \ESOK{(number of edges)} of a shortest path joining $u$ and $v$.
The {\em diameter} 
of $G$ is the maximum distance between two vertices of $G$.
We denote by $P_n$, $n\ge 1$, the path of order $n$ and by $C_n$, $n\ge 3$, the cycle of order~$n$.

A {\em packing $k$-colouring} of $G$ is a mapping $\pi:V(G)\rightarrow\{1,\ldots,k\}$
such that, for every two distinct vertices $u$ and $v$,  \ESOK{$\pi(u)=\pi(v)=i$ implies $d_G(u,v)>i$}.
The {\em packing chromatic number} $\pcn(G)$ of $G$ is then the smallest $k$ such that
$G$ admits a packing $k$-colouring.
In other words, $\pcn(G)$ is the smallest integer $k$ such that
$V(G)$ can be partitioned into $k$ disjoint subsets $V_i$, $1\le i\le k$, in such a way that every two vertices
in $V_i$ are at distance greater than $i$ in $G$ for every $i$, $1\le i\le k$.
A packing colouring of $G$ is {\em optimal} if it uses exactly $\pcn(G)$ colours.

The packing colouring of graphs was introduced by Goddard, Hedetniemi, Hedetniemi, Harris and Rall in~\cite{GHHHR03,GHHHR08}, under the name {\em broadcast colouring}.
%
%
In their seminal paper~\cite{GHHHR08},
the question of determining the maximum
packing chromatic number in the class of cubic graphs of a given order is posed.
In~\cite{S04}, Sloper proved that the packing chromatic number is unbounded in
the class of $k$-ary trees for every $k\ge 3$, from which it follows that
the packing chromatic number is unbounded in the class of graphs with maximum degree~4.

In~\cite{GT16}, Gastineau and Togni observed that each cubic graph of order at most~20
has packing chromatic number at most~10. They also observed that the largest cubic graph with diameter~4
(this graph has 38 vertices and is described in~\cite{AFY86})
has packing chromatic number~13,
and asked whether there exists a cubic graph with packing chromatic number larger than~13 or not.
This question was answered positively by Bre\v sar, Klav\v zar, Rall and Wash~\cite{BKRW17b}
who exhibited a cubic graph on 78~vertices with packing chromatic number at least~14.
Recently, Balogh, Kostochka and Liu finally proved in~\cite{BKL17} that the packing
chromatic number is unbounded in the class of cubic graphs, and Bre\v sar and Ferme
gave in~\cite{BF18} an explicit infinite family of subcubic graphs with unbounded packing chromatic
number.

On the other hand, the packing chromatic number is known to be upper bounded in several classes 
of graphs with maximum degree~3, as for instance complete binary trees~\cite{S04}, 
hexagonal lattices~\cite{BKR07,FKL09,KV14},
base-3 Sierpi\'nski graphs~\cite{BKR16} or particular Sierpi\'nski-type graphs~\cite{BF17},
subdivisions of subcubic graphs~\cite{GT16,BKRW17a}
and of cubic graphs~\cite{BKL18},
or several subclasses of outerplanar subcubic graphs~\cite{GHT17}.

\medskip

We prove in this paper that the packing chromatic number is bounded
in other classes of cubic graphs, extending in particular partial
results given in~\cite{WR13}.
More precisely, we determine the exact value of, or upper and lower bounds on, the packing chromatic number of
circular ladders (in Section~\ref{sec:ladder}), 
H-graphs (in Section~\ref{sec:Hgraphs}) and 
generalised H-graphs (in Section~\ref{sec:geneHgraphs}).

\section{Preliminary results}
\label{sec:preliminaryResults}

We give in this section a few results that will be useful in the sequel.

Let $G$ be a graph. 
A subset $S$ of $V(G)$ is an {\em $i$-packing}, for some integer $i\ge 1$,
if any two vertices in $S$ are at distance at least $i+1$ in $G$.
Note that such a set $S$ is a $1$-packing if and only if $S$ is an independent set.
A packing colouring of $G$ is thus a 
partition of $V(G)$ into $k$ disjoint subsets $V_1,\dots,V_k$, such that
$V_i$ is an $i$-packing for every $i$, $1\le i\le k$.

For every integer $i\ge 1$, we denote by $\rho_i(G)$ the maximum cardinality
of an $i$-packing in $G$.
Since at most $\rho_i(G)$ vertices can be assigned colour $i$ in any packing
colouring of $G$, we have the following result.

\begin{proposition}
If $G$ is a graph with $\pcn(G)=k$, then
$$\sum_{i=1}^{i=k}\rho_i(G) \ge |V(G)|.$$
\label{prop:rho-i}
\end{proposition}

Let $H$ be a subgraph of $G$. Since $d_G(u,v)\le d_H(u,v)$ for any two vertices $u,v\in V(H)$,
the restriction to $V(H)$ of any packing colouring of $G$ is a packing colouring
of $H$. 
Hence, having packing chromatic number at most $k$
is a hereditary property:

\begin{proposition}[Goddard, Hedetniemi, Hedetniemi, Harris and Rall~\cite{GHHHR08}] \mbox{}\\
Let $G$ and $H$ be two graphs.
If $H$ is a subgraph of $G$, then $\pcn(H)\le\pcn(G)$.
\label{prop:subgraph}
\end{proposition}

In particular, Proposition~\ref{prop:subgraph} gives a lower bound on the packing
chromatic number of a graph~$G$ whenever $G$ contains a subgraph $H$ whose packing chromatic
number is known. 
As we will see later, all the cubic graphs we consider in this paper contain a corona of a cycle
as a subgraph.
Recall that the {\em corona} $G\odot K_1$ of a graph $G$ is the graph obtained from $G$ by adding a degree-one neighbour to every vertex  of $G$.
In~\cite{LBS17}, we have determined with I.~Bouchemakh 
the packing chromatic number of the corona of 
cycles.

\begin{theorem}[La\"iche, Bouchemakh, Sopena~\cite{LBS17}] \mbox{}\\
The packing chromatic number of the corona graph $C_n\odot K_1$ is given by:
$$\chi_\rho(C_{n}\odot K_1)= \left\{
 \begin{array}{ll}
     $4$ & \hbox{if $n\in\{3,4\}$,} \\
     $5$ & \hbox{if $n\ge 5$.}
 \end{array}
\right.
$$
\label{th:CrCn}
\end{theorem}

This result will thus provide a lower bound on the packing chromatic number of each
cubic graph considered in this paper.

\section{Circular ladders}
\label{sec:ladder}

Recall that the Cartesian product $G\,\Box\,H$ of two graphs $G$ and $H$ is the graph with vertex set $V(G)\times V(H)$,
two vertices $(u,u')$ and $(v,v')$ being adjacent if and only if either $u=v$ and $u'v'\in E(H)$
or $u'=v'$ and $uv\in E(G)$.

\begin{figure}
\begin{center}
\begin{tikzpicture}[domain=0:12,x=1cm,y=1cm]
\SOMMET{2.5,-4}; \SOMMET{2.5,-3};
\SOMMET{3.5,-4}; \SOMMET{3.5,-3};
\SOMMET{4.5,-4}; \SOMMET{4.5,-3};
\SOMMET{5.5,-4}; \SOMMET{5.5,-3};
\SOMMET{6.5,-4}; \SOMMET{6.5,-3};
\SOMMET{7.5,-4}; \SOMMET{7.5,-3};
\SOMMET{8.5,-4}; \SOMMET{8.5,-3};
\ETIQUETTE{2.5,-2.4}{$u_0$}; \ETIQUETTE{2.5,-4.1}{$v_0$};
\ETIQUETTE{3.5,-2.4}{$u_1$}; \ETIQUETTE{3.5,-4.1}{$v_1$};
\ETIQUETTE{4.5,-2.4}{$u_2$}; \ETIQUETTE{4.5,-4.1}{$v_2$};
\ETIQUETTE{5.5,-2.4}{$u_3$}; \ETIQUETTE{5.5,-4.1}{$v_3$};
\ETIQUETTE{6.5,-2.4}{$u_4$}; \ETIQUETTE{6.5,-4.1}{$v_4$};
\ETIQUETTE{7.5,-2.4}{$u_5$}; \ETIQUETTE{7.5,-4.1}{$v_5$};
\ETIQUETTE{8.5,-2.4}{$u_6$}; \ETIQUETTE{8.5,-4.1}{$v_6$};
\ARETEH{2.5,-3}{6};
\ARETEH{2.5,-4}{6};
\ARETEV{2.5,-4};
\ARETEV{3.5,-4};
\ARETEV{4.5,-4};
\ARETEV{5.5,-4};
\ARETEV{6.5,-4};
\ARETEV{7.5,-4};
\ARETEV{8.5,-4};
\draw[thick] (2.5,-3) .. controls (3.5,-1.8) and (7.5,-1.8) .. (8.5,-3);
\draw[thick] (2.5,-4) .. controls (3.5,-5.2) and (7.5,-5.2) .. (8.5,-4);
\end{tikzpicture}
\caption{The circular ladder $CL_7$.}
\label{fig:CL7}
\end{center}
\end{figure}
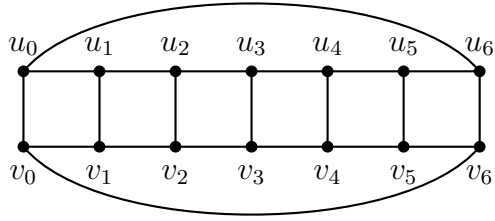

The \emph{circular ladder} $CL_n$ of length $n\ge 3$ is the Cartesian product $CL_n=C_n\,\Box\,K_2$.
Note that $CL_n$ is a bipartite graph if and only if $n$ is even.

For every circular ladder $CL_n$, 
we let 
$$V(CL_n)=\{u_0,\dots,u_{n-1}\}\ \cup\ \{v_0,\dots,v_{n-1}\},$$
and
$$E(CL_n)=\{u_iv_i\ |\ 0\le i\le n-1\}\ \cup\ \{u_iu_{i+1},v_iv_{i+1}\ |\ 0\le i\le n-1\}$$
(subscripts are taken modulo $n$).
Figure~\ref{fig:CL7} depicts the circular ladder $CL_7$.

Note that for every $n\ge 3$, the corona graph $C_n\odot K_1$ is a subgraph
of the circular ladder $CL_n$.
Therefore, by Proposition~\ref{prop:subgraph}, Theorem~\ref{th:CrCn} provides a lower
bound on the packing chromatic number of circular ladders.
More precisely, $\pcn(CL_n)\ge 4$ if $n\in\{3,4\}$, and $\pcn(CL_n)\ge 5$ if $n\ge 5$.

William and Roy~\cite{WR13} proved that the packing chromatic number of a circular ladder of length $n=6q$, $q\ge 1$,
is at most~5.
In Theorem~\ref{th:cl} below, we extend this result and determine the packing chromatic number of every circular ladder.

We first need the following technical lemma, which will also be useful in Section~\ref{sec:geneHgraphs}.

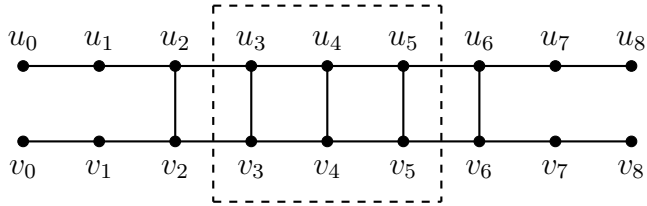
\begin{figure}
\begin{center}
\begin{tikzpicture}[domain=0:12,x=1cm,y=1cm]
\SOMMET{2.5,-4}; \SOMMET{2.5,-3};
\SOMMET{3.5,-4}; \SOMMET{3.5,-3};
\SOMMET{4.5,-4}; \SOMMET{4.5,-3};
\SOMMET{5.5,-4}; \SOMMET{5.5,-3};
\SOMMET{6.5,-4}; \SOMMET{6.5,-3};
\SOMMET{7.5,-4}; \SOMMET{7.5,-3};
\SOMMET{8.5,-4}; \SOMMET{8.5,-3};
\SOMMET{9.5,-4}; \SOMMET{9.5,-3};
\SOMMET{10.5,-4}; \SOMMET{10.5,-3};
\ETIQUETTE{2.5,-2.4}{$u_0$}; \ETIQUETTE{2.5,-4.1}{$v_0$};
\ETIQUETTE{3.5,-2.4}{$u_1$}; \ETIQUETTE{3.5,-4.1}{$v_1$};
\ETIQUETTE{4.5,-2.4}{$u_2$}; \ETIQUETTE{4.5,-4.1}{$v_2$};
\ETIQUETTE{5.5,-2.4}{$u_3$}; \ETIQUETTE{5.5,-4.1}{$v_3$};
\ETIQUETTE{6.5,-2.4}{$u_4$}; \ETIQUETTE{6.5,-4.1}{$v_4$};
\ETIQUETTE{7.5,-2.4}{$u_5$}; \ETIQUETTE{7.5,-4.1}{$v_5$};
\ETIQUETTE{8.5,-2.4}{$u_6$}; \ETIQUETTE{8.5,-4.1}{$v_6$};
\ETIQUETTE{9.5,-2.4}{$u_7$}; \ETIQUETTE{9.5,-4.1}{$v_7$};
\ETIQUETTE{10.5,-2.4}{$u_8$}; \ETIQUETTE{10.5,-4.1}{$v_8$};
\ARETEH{2.5,-3}{8};
\ARETEH{2.5,-4}{8};
\ARETEV{4.5,-4};
\ARETEV{5.5,-4};
\ARETEV{6.5,-4};
\ARETEV{7.5,-4};
\ARETEV{8.5,-4};
\draw[thick,dashed] (5,-2.2) -- (8,-2.2);
\draw[thick,dashed] (5,-4.8) -- (8,-4.8);
\draw[thick,dashed] (5,-2.2) -- (5,-4.8);
\draw[thick,dashed] (8,-2.2) -- (8,-4.8);
\end{tikzpicture}
\caption{The graph $X$.}
\label{fig:graphX}
\end{center}
\end{figure}

\begin{lemma}
Let $X$ be the graph depicted in Figure~\ref{fig:graphX}, 
and $\pi$ be a packing $5$-colouring of $X$.
If $\pi(u_i)\neq 1$ and $\pi(v_i)\neq 1$ for some integer $i$, $3\le i\le 5$,
then either $u_i$ or $v_i$ has colour $2$, and its three neighbours have colours
$3$, $4$ and $5$ (the three corresponding edges are the vertical edges surrounded by the dashed box).
\label{lem:graphX}
\end{lemma}

\begin{proof}
The proof is done by case analysis and is given in Appendix~\ref{ap:graphX}.
\end{proof}

Observe now that for every integer $n\ge 9$,
the subgraph of $CL_n$ induced by the set of vertices
$\{u_i,v_i\ |\ 0\le i\le 8\}$ contains
the graph $X$ of Figure~\ref{fig:graphX} as a subgraph.
Moreover, every packing $5$-colouring $\pi$ of $CL_n$, $6\le n\le 8$, can be 
``unfolded'' to produce a packing $5$-colouring $\pi'$ of $X$, 
by setting 
$\pi'(u_i)=\pi(u_i)$ and $\pi'(v_i)=\pi(v_i)$ for every $i$, $0\le i\le n-1$,
and
$\pi'(u_{n-1+j})=\pi(u_{j-1})$ and $\pi'(v_{n-1+j})=\pi(v_{j-1})$ for every $j$, $1\le j\le 9-n$.
This follows from the fact that vertices $u_j$ and~$u_{n+j}$, as well as vertices
$v_j$ and~$v_{n+j}$, are at distance $n\ge 6$ from each other, while the largest colour
used by $\pi'$ is~$5$.
Therefore, thanks to the symmetries of $CL_n$ for every $n\ge 6$,
Proposition~\ref{prop:subgraph} and Lemma~\ref{lem:graphX} give the following corollary.

\begin{corollary}
Let $CL_n$, $n\ge 6$, be a circular ladder with $\pcn(CL_n)\le 5$,
and $\pi$ be a packing $5$-colouring of $CL_n$.
For every integer $i$, $0\le i\le n-1$, if $\pi(u_i)\neq 1$ and $\pi(v_{i})\neq 1$,
then either $u_i$ or $v_i$ has colour $2$, and its three neighbours have colours
$3$, $4$ and $5$.
\label{cor:CL1alt}
\end{corollary}

Let $CL_n$ be a circular ladder satisfying the hypothesis of Corollary~\ref{cor:CL1alt},
and $\pi$ be a packing 5-colouring of $CL_n$.
From Corollary~\ref{cor:CL1alt}, it follows that if $\pi(u_i)\neq 1$ and $\pi(v_i)\neq 1$ for some edge $u_iv_i$ of $CL_n$,
then the colour~2 has to be used on the edge $u_iv_i$ and, since the neighbours of the 2-coloured vertex
are coloured with 3, 4 and~5, the colour~2 can be replaced by colour~1.
Therefore, we get the following corollary.


\begin{corollary}
If $CL_n$, $n\ge 6$, is a circular ladder with $\pcn(CL_n)\le 5$,
then there exists a packing $5$-colouring
of $CL_n$ such that the colour~$1$ is used on each edge of $CL_n$.
\label{coro:bipa}
\end{corollary}

Note that from Corollary~\ref{coro:bipa}, it follows that
for every integer $n\ge 6$, $\pcn(CL_n)\le 5$ implies that $CL_n$ is a bipartite graph.
Hence, $\pcn(CL_n)\ge 6$ for every odd $n\ge 6$.


\begin{figure} 
\begin{center}
\begin{tikzpicture}[domain=0:2,x=1cm,y=1cm]
\SOMMET{0,0}; \SOMMET{1,0}; \SOMMET{2,0}; 
\SOMMET{0,1}; \SOMMET{1,1}; \SOMMET{2,1}; 
\ARETEH{0,0}{2}; \ARETEH{0,1}{2}; 
\ARETEV{0,0};
\ARETEV{1,0};
\ARETEV{2,0};
\ETIQUETTEA{0,1}{1};
\ETIQUETTEA{1,1}{2};
\ETIQUETTEA{2,1}{3};
\ETIQUETTEB{0,0}{4};
\ETIQUETTEB{1,0}{1};
\ETIQUETTEB{2,0}{5};
\draw[thick] (0,1) .. controls (0.5,1.8) and (1.5,1.8) .. (2,1);
\draw[thick] (0,0) .. controls (0.5,-0.8) and (1.5,-0.8) .. (2,0);
\end{tikzpicture}
\hskip 1.5cm
\begin{tikzpicture}[domain=0:3,x=1cm,y=1cm]
\SOMMET{0,0}; \SOMMET{1,0}; \SOMMET{2,0}; \SOMMET{3,0}; 
\SOMMET{0,1}; \SOMMET{1,1}; \SOMMET{2,1}; \SOMMET{3,1}; 
\ARETEH{0,0}{3}; \ARETEH{0,1}{3}; 
\ARETEV{0,0};
\ARETEV{1,0};
\ARETEV{2,0};
\ARETEV{3,0};
\ETIQUETTEA{0,1}{1};
\ETIQUETTEA{1,1}{2};
\ETIQUETTEA{2,1}{1};
\ETIQUETTEA{3,1}{3};
\ETIQUETTEB{0,0}{4};
\ETIQUETTEB{1,0}{1};
\ETIQUETTEB{2,0}{5};
\ETIQUETTEB{3,0}{1};
\draw[thick] (0,1) .. controls (0.5,1.8) and (2.5,1.8) .. (3,1);
\draw[thick] (0,0) .. controls (0.5,-0.8) and (2.5,-0.8) .. (3,0);
\end{tikzpicture}
\hskip 1.5cm
\begin{tikzpicture}[domain=0:4,x=1cm,y=1cm]
\SOMMET{0,0}; \SOMMET{1,0}; \SOMMET{2,0}; \SOMMET{3,0}; \SOMMET{4,0}; 
\SOMMET{0,1}; \SOMMET{1,1}; \SOMMET{2,1}; \SOMMET{3,1}; \SOMMET{4,1}; 
\ARETEH{0,0}{4}; \ARETEH{0,1}{4}; 
\ARETEV{0,0};
\ARETEV{1,0};
\ARETEV{2,0};
\ARETEV{3,0};
\ARETEV{4,0};
\ETIQUETTEA{0,1}{1};
\ETIQUETTEA{1,1}{3};
\ETIQUETTEA{2,1}{1};
\ETIQUETTEA{3,1}{2};
\ETIQUETTEA{4,1}{6};
\ETIQUETTEB{0,0}{2};
\ETIQUETTEB{1,0}{1};
\ETIQUETTEB{2,0}{4};
\ETIQUETTEB{3,0}{1};
\ETIQUETTEB{4,0}{5};
\draw[thick] (0,1) .. controls (0.5,1.8) and (3.5,1.8) .. (4,1);
\draw[thick] (0,0) .. controls (0.5,-0.8) and (3.5,-0.8) .. (4,0);
\end{tikzpicture}
\caption{Optimal packing colouring of $CL_3$, $CL_4$ and $CL_5$.}
\label{fig:CL345}
\end{center}
\end{figure}
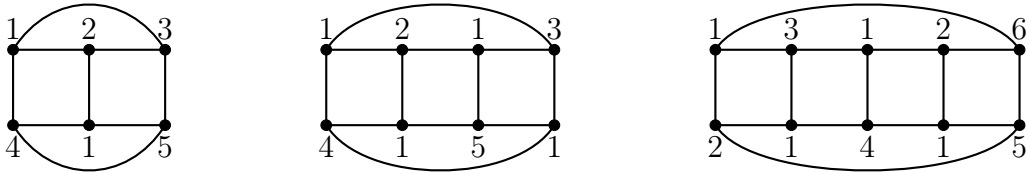


We are now able to prove the main result of this section.

\begin{theorem}
For every integer $n\ge 3$,
$$\pcn(CL_n)=\left\{
     \begin{array}{ll}
        5 & \hbox{if $n=3$, or $n$ is even and $n\not\in\{8,14\}$,} \\
        7 & \hbox{if $n\in\{7,8,9\}$,} \\
        6 & \hbox{otherwise.}
     \end{array}\right.$$
\label{th:cl}
\end{theorem}

\begin{proof}
We first consider the case $n\le 5$.
Figure~\ref{fig:CL345} describes a packing 5-colouring of $CL_3$ and $CL_4$,
and a packing 6-colouring of $CL_5$.
We claim that these three packing colourings are optimal.
To see that, observe that
$\rho_1(CL_3)=2$, $\rho_i(CL_3)=1$ for every $i\ge 2$,
$\rho_1(CL_4)=\rho_1(CL_5)=4$, $\rho_2(CL_4)=\rho_2(CL_5)=2$, and 
$\rho_i(CL_4)=\rho_i(CL_5)=1$ for every $i\ge 3$.
The optimality for $CL_3$ and $CL_5$ then follows from Proposition~\ref{prop:rho-i}.
The optimality for $CL_4$ also follows, with the additional 
observation that colour~$2$ can be used at most once if colour~$1$ is used four times.

Assume now $n\ge 6$.
Since $n\ge 6$ and every circular ladder $CL_n$ contains the corona graph $C_n\odot K_1$ as a subgraph,
we get $\pcn(CL_n)\ge \pcn(C_n\odot K_1) \ge 5$ by Theorem~\ref{th:CrCn} and Proposition~\ref{prop:subgraph}.
Moreover, by Corollary~\ref{coro:bipa}, we have $\pcn(CL_n)\ge 6$ if $n$ is odd.

We now consider two general cases.

\begin{enumerate}
\item {\it $n$ is even} and $n\notin\{8,14\}$.\\
As observed above, in that case, it is enough to exhibit a packing $5$-colouring of $CL_n$
to prove $\pcn(CL_n)=5$.

If $n\equiv 0\pmod 6$, a packing 5-colouring of $CL_n$ is obtained by repeating the following 
circular pattern (the first row gives the colours of vertices $u_i$, $0\le i\le n-1$,
the second row gives the colours of vertices $v_i$, $0\le i\le n-1$, according to the value of $(i\mod 6)$):
\begin{center}
$\begin{array}{||c||}
1\ 3\ 1\ 2\ 1\ 5\\
2\ 1\ 4\ 1\ 3\ 1
\end{array}$
\end{center}

If $n\equiv 2\pmod 6$, which implies $n\ge 20$, a packing 5-colouring of $CL_n$ is obtained by repeating
the previous circular pattern $\frac{n-20}{6}$ times and adding a pattern of length 20, as illustrated below:
\begin{center}
$\begin{array}{||c||c}
1\ 3\ 1\ 2\ 1\ 5 & 1\ 3\ 1\ 2\ 1\ 3\ 1\ 4\ 1\ 5\ 1\ 3\ 1\ 2\ 1\ 3\ 1\ 4\ 1\ 5\\
2\ 1\ 4\ 1\ 3\ 1 & 2\ 1\ 4\ 1\ 5\ 1\ 2\ 1\ 3\ 1\ 2\ 1\ 4\ 1\ 5\ 1\ 2\ 1\ 3\ 1
\end{array}$
\end{center}

Finally, if $n\equiv 4\pmod 6$, which implies $n\ge 10$, a packing 5-colouring of $CL_n$ is obtained by repeating
the same circular pattern $\frac{n-10}{6}$ times and adding a pattern of length 10:
\begin{center}
$\begin{array}{||c||c}
1\ 3\ 1\ 2\ 1\ 5 & 1\ 3\ 1\ 2\ 1\ 3\ 1\ 4\ 1\ 5\\
2\ 1\ 4\ 1\ 3\ 1 & 2\ 1\ 4\ 1\ 5\ 1\ 2\ 1\ 3\ 1
\end{array}$
\end{center}

\item {\it $n$ is odd} and $n\ge 11$.\\
As observed above, in that case, it is enough to exhibit a packing $6$-colouring of $CL_n$
to prove $\pcn(CL_n)=6$.

Similarly as in the previous case, 
if $n\equiv 1,3$ or $5\pmod 6$,
a packing 6-colouring of $CL_n$ is obtained by repeating
the previous circular pattern $\frac{n-7}{6}$, $\frac{n-9}{6}$ or $\frac{n-5}{6}$ times, respectively, 
and adding a pattern of length 7, 9 or 5, respectively, as illustrated below:
\begin{center}
$\begin{array}{||c||c}
1\ 3\ 1\ 2\ 1\ 5 & 1\ 3\ 1\ 4\ 1\ 2\ 6\\
2\ 1\ 4\ 1\ 3\ 1 & 6\ 1\ 2\ 1\ 3\ 1\ 5
\end{array}$
\end{center}
\begin{center}
$\begin{array}{||c||c}
1\ 3\ 1\ 2\ 1\ 5 & 1\ 4\ 1\ 2\ 3\ 1\ 4\ 1\ 6\\
2\ 1\ 4\ 1\ 3\ 1 & 2\ 1\ 6\ 1\ 5\ 2\ 1\ 3\ 1
\end{array}$
\end{center}
\begin{center}
$\begin{array}{||c||c}
1\ 3\ 1\ 2\ 1\ 5 & 1\ 3\ 1\ 2\ 6\\
2\ 1\ 4\ 1\ 3\ 1 & 2\ 1\ 4\ 1\ 5
\end{array}$
\end{center}

\end{enumerate}

It remains to consider four cases, namely $n=7,8,9,14$,
which we consider separately.

\begin{enumerate}
\item $n=7$.\\
We first claim that $\pcn(CL_7)\ge 7$.
Note that $\rho_1(CL_7)=6$, $\rho_2(CL_7)=3$, $\rho_3(CL_7)=2$, 
and $\rho_i(CL_7)=1$ for every $i\ge 4$.
However, if we use six times colour~$1$, colour~$2$ can be used at most twice.
Hence, at most 13 vertices of $CL_7$ can be coloured with a colour in $\{1,\dots,6\}$
and the claim follows.

A packing 7-colouring of $CL_7$ is given by the following pattern:
\begin{center}
$\begin{array}{c}
1\ 3\ 1\ 2\ 1\ 4\ 5\\
2\ 1\ 6\ 1\ 3\ 1\ 7
\end{array}$
\end{center}

\item $n=8$.\\
We first claim that $\pcn(CL_8)\ge 7$.
Note that $\rho_1(CL_8)=8$, $\rho_2(CL_8)=4$, $\rho_3(CL_8)=\rho_4(CL_8)=2$, 
and $\rho_i(CL_8)=1$ for every $i\ge 5$.
However, if we use eight times colour~$1$, colour~$2$ can be used at most twice, and then colour~$4$ at most once.
On the other hand, if we use seven times colour~$1$, then, either colour~$2$ is used thrice, and then colour~$4$ can
be used at most once,
or colour~$2$ is used at most twice, and then colour~$4$ can be used at most twice.
Hence, at most 15 vertices of $CL_8$ can be coloured with a colour in $\{1,\dots,6\}$
and the claim follows.

A packing 7-colouring of $CL_8$ is given by the following pattern:
\begin{center}
$\begin{array}{c}
1\ 3\ 1\ 2\ 1\ 5\ 1\ 7\\
2\ 1\ 4\ 1\ 3\ 1\ 6\ 1
\end{array}$
\end{center}

\item $n=9$.\\
We first claim that $\pcn(CL_9)\ge 7$.
Note that $\rho_1(CL_9)=8$, $\rho_2(CL_9)=4$, $\rho_3(CL_9)=\rho_4(CL_9)=2$, 
and $\rho_i(CL_9)=1$ for every $i\ge 5$.
However, if we use eight times colour~$1$, colour~$2$ can be used at most thrice.
Hence, at most 17 vertices of $CL_9$ can be coloured with a colour in $\{1,\dots,6\}$
and the claim follows.

A packing 7-colouring of $CL_9$ is given by the following pattern:
\begin{center}
$\begin{array}{c}
1\ 3\ 1\ 2\ 1\ 5\ 1\ 4\ 6\\
2\ 1\ 4\ 1\ 3\ 1\ 2\ 1\ 7
\end{array}$
\end{center}

\item $n=14$.\\
We first claim that $\pcn(CL_{14})\ge 6$.
Note that $\rho_1(CL_{14})=14$, $\rho_2(CL_{14})=6$, $\rho_3(CL_{14})=4$, $\rho_4(CL_{14})=3$
and $\rho_5(CL_{14})=2$.
However, if we use 14 times colour~$1$, colour~$2$ can be used at most four times.
On the other hand, if we use 13 times colour~$1$, colour~$2$ can be used at most five times.
Hence, at most 27 vertices of $CL_{14}$ can be coloured with a colour in $\{1,\dots,5\}$
and the claim follows.

A packing 6-colouring of $CL_{14}$ is given by the following pattern:
\begin{center}
$\begin{array}{c}
1\ 3\ 1\ 2\ 1\ 5\ 1\ 2\ 1\ 4\ 1\ 3\ 1\ 6\\
2\ 1\ 4\ 1\ 3\ 1\ 6\ 1\ 3\ 1\ 2\ 1\ 5\ 1
\end{array}$
\end{center}

\end{enumerate}

This completes the proof of Theorem~\ref{th:cl}.
\end{proof}

\section{H-graphs}
\label{sec:Hgraphs}


\begin{figure} 
\begin{center}
\begin{tikzpicture}[domain=0:12,x=1cm,y=1cm]


\SOMMET{2.5,-5};\SOMMET{2.5,-4}; \SOMMET{2.5,-3};
\SOMMET{3.5,-5};\SOMMET{3.5,-4}; \SOMMET{3.5,-3};
\SOMMET{4.5,-5};\SOMMET{4.5,-4}; \SOMMET{4.5,-3};
\SOMMET{5.5,-5};\SOMMET{5.5,-4}; \SOMMET{5.5,-3};
\SOMMET{6.5,-5};\SOMMET{6.5,-4}; \SOMMET{6.5,-3};
\SOMMET{7.5,-5};\SOMMET{7.5,-4}; \SOMMET{7.5,-3};
\SOMMET{8.5,-5};\SOMMET{8.5,-4}; \SOMMET{8.5,-3};
\SOMMET{9.5,-5};\SOMMET{9.5,-4}; \SOMMET{9.5,-3};


\ETIQUETTE{2.5,-2.4}{$u_0$}; \ETIQUETTE{2.5,-5.1}{$w_0$}; \ETIQUETTE{2.25,-3.51}{$v_0$};
\ETIQUETTE{3.5,-2.4}{$u_1$}; \ETIQUETTE{3.5,-5.1}{$w_1$}; \ETIQUETTE{3.75,-3.51}{$v_1$};
\ETIQUETTE{4.5,-2.4}{$u_2$}; \ETIQUETTE{4.5,-5.1}{$w_2$}; \ETIQUETTE{4.25,-3.51}{$v_2$};
\ETIQUETTE{5.5,-2.4}{$u_3$}; \ETIQUETTE{5.5,-5.1}{$w_3$}; \ETIQUETTE{5.75,-3.51}{$v_3$};
\ETIQUETTE{6.5,-2.4}{$u_4$}; \ETIQUETTE{6.5,-5.1}{$w_4$}; \ETIQUETTE{6.25,-3.51}{$v_4$};
\ETIQUETTE{7.5,-2.4}{$u_5$}; \ETIQUETTE{7.5,-5.1}{$w_5$}; \ETIQUETTE{7.75,-3.51}{$v_5$};
\ETIQUETTE{8.5,-2.4}{$u_6$}; \ETIQUETTE{8.5,-5.1}{$w_6$}; \ETIQUETTE{8.25,-3.51}{$v_6$};
\ETIQUETTE{9.5,-2.4}{$u_7$}; \ETIQUETTE{9.5,-5.1}{$w_7$}; \ETIQUETTE{9.75,-3.51}{$v_7$};

\ARETEH{2.5,-3}{7};
\ARETEH{2.5,-5}{7};

\ARETEV{2.5,-5};
\ARETEV{3.5,-5};
\ARETEV{4.5,-5};
\ARETEV{5.5,-5};
\ARETEV{6.5,-5};
\ARETEV{7.5,-5};
\ARETEV{8.5,-5};
\ARETEV{9.5,-5};

\ARETEV{2.5,-4};
\ARETEV{3.5,-4};
\ARETEV{4.5,-4};
\ARETEV{5.5,-4};
\ARETEV{6.5,-4};
\ARETEV{7.5,-4};
\ARETEV{8.5,-4};
\ARETEV{9.5,-4};

\ARETEH{2.5,-4}{1};
\ARETEH{4.5,-4}{1};
\ARETEH{6.5,-4}{1};
\ARETEH{8.5,-4}{1};

\draw[thick] (2.5,-3) .. controls (4,-1.6) and (8,-1.6) .. (9.5,-3);
\draw[thick] (2.5,-5) .. controls (4,-6.4) and (8,-6.4) .. (9.5,-5);

\end{tikzpicture}
\caption{The H-graph $H(4)$.}
\label{fig:H4}
\end{center}
\end{figure}
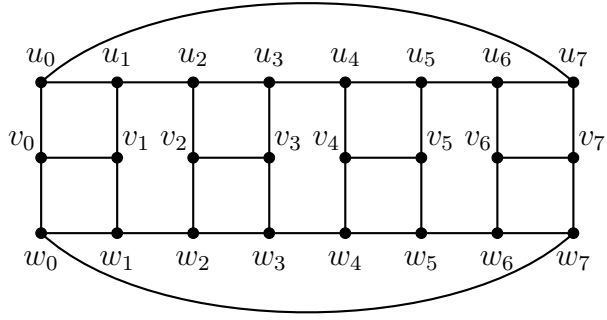


The \emph{H-graph} $H(r)$, $r\ge 2$, is the 3-regular graph of order $6r$, 
with vertex set
$$V(H(r))=\{u_i,v_i,w_i:0\le i\le 2r-1\},$$ 
 and edge set (subscripts are taken modulo $2r$)
 $$\begin{array}{rcl}
 E(H(r)) & = & \{(u_i,u_{i+1}),\ (w_i,w_{i+1}),\ (u_i,v_i),\ (v_i,w_i):0\le i\le 2r-1\}\\
 & & \cup \ \{(v_{2i},v_{2i+1}):0\le i\le r-1\}. 
 \end{array}$$
Figure~\ref{fig:H4} depicts the H-graph $H(4)$.
These graphs have been introduced and studied by William and Roy in~\cite{WR13},
where it is proved that 
$\pcn(H(r))\le 5$ for every H-graph $H(r)$ with even $r\ge 4$.
We complete their result in Theorem~\ref{th:H(r)} below.

We first prove a technical lemma.
For every pair of integers $r\ge 2$ and $0\le i\le r-1$,
we denote by $G_i(r)$ the subgraph of $H(r)$ 
induced by the set of vertices 
$\{u_{2i},u_{2i+1},v_{2i},v_{2i+1},w_{2i},w_{2i+1}\}.$
Observe that for every $r\ge 2$, all the subgraphs $G_i(r)$ are isomorphic
to the graph depicted in Figure~\ref{fig:G_i(r)}(a),
and thus $\pcn(G_i(r))=\pcn(P_2\Box P_3)=4$~\cite{GHHHR08}.

For a given packing 5-colouring $\pi$ of $H(r)$,
we denote by $\pi(G_i(r))$ the set of colours assigned to the vertices of $G_i(r)$.
We then have the following result.


\begin{lemma}\label{lem:G_i(r)}
For every integer $r\ge 3$ and every packing $5$-colouring $\pi$ of $H(r)$, 
$\pi(G_i(r)) \cap \pi(G_{i+1}(r)) = \{1,2,3\}$ for every $i$, $0\le i\le r-1$.
\end{lemma}

\begin{proof}
Since $\pcn(P_2\Box P_3)=4$, 
every packing 5-colouring of $H(r)$ must use colour 4 or colour~5 on every $G_i(r)$, $0\le i\le r-1$.
%
We now prove that if colour~4 (resp. colour~5) is used on $G_i(r)$, then colour~4 (resp. colour~5) cannot
be used on $G_{i+1}(r)$.
Observe first that every vertex of $G_i(r)$ is at distance at most~5 from every vertex of $G_{i+1}(r)$.
Therefore, colour~5 cannot be used on both $G_i(r)$ and $G_{i+1}(r)$.
Suppose now that colour~4 is used on both $G_i(r)$ and $G_{i+1}(r)$.
Up to symmetries, we necessarily have one of the two following cases.

\begin{enumerate}
\item $\pi(u_{2i})=\pi(w_{2i+3})=4$ (see Figure~\ref{fig:G_i(r)}(b)).\\
Since every vertex of $G_{i-1}(r)$ is at distance at most~4 from $u_{2i}$,
it follows that $G_{i-1}(r)$ does not contain the colour $4$.
This implies that $G_{i-1}(r)$ contains the colour $5$ since $\pcn(G_{i-1}(r)) > 3$.
By symmetry, $G_{i+2}(r)$ must also contain the colour $5$.
Furthermore, since two consecutive $G_i(r)$s cannot both use colour $5$, neither $G_{i}(r)$ nor
$G_{i+1}(r)$ contains the colour $5$.

Now, on the remaining uncoloured vertices of $G_i(r)$, colour~1 can be used at
most thrice, colour~2 at most twice and colour~3 at most once.
If colour~1 is used thrice, then we necessarily have $\pi(u_{2i+1})=\pi(v_{2i})=\pi(w_{2i+1})=1$,
so that $\{\pi(v_{2i+1}),\pi(w_{2i})\}=\{2,3\}$, 
and no colour is available for $w_{2i+2}$ (recall that colour~5 is not used on $G_{i+1}(r)$).
If colour~1 is used twice, then we necessarily have, up to symmetry,
$\pi(v_{2i})=\pi(w_{2i+1})=1$, $\pi(u_{2i+1})=\pi(w_{2i})=2$, and $\pi(v_{2i+1})=3$,
and no colour is available for $w_{2i+2}$.

\item $\pi(v_{2i})=\pi(v_{2i+3})=4$ (see Figure~\ref{fig:G_i(r)}(c)).\\
Similarly as before, since every vertex of $G_{i-1}(r)$ is at distance at most~4 from $v_{2i}$
and two consecutive $G_i(r)$'s cannot both use colour~5, it follows from
the first item of Lemma~\ref{lem:G_i(r)} that colour~5 is used neither on $G_i(r)$, nor,
by symmetry, on $G_{i+1}(r)$.

Again, on the remaining uncoloured vertices of $G_i(r)$, colour~1 can be used at
most thrice, colour~2 at most twice and colour~3 at most once.
If colour~1 is used thrice, then we necessarily have $\pi(u_{2i})=\pi(v_{2i+1})=\pi(w_{2i})=1$,
so that $\{\pi(u_{2i+1}),\pi(w_{2i+1})\}=\{2,3\}$.
Up to symmetry, we may assume $\pi(u_{2i+1})=2$ and $\pi(w_{2i+1})=3$,
which implies $\pi(u_{2i+2})=1$, 
and no colour is available for $v_{2i+2}$ (recall that colour~5 is not used on $G_{i+1}(r)$).
If colour~1 is used twice, then we necessarily have, up to symmetry,
$\pi(u_{2i+1})=\pi(w_{2i})=1$, $\pi(u_{2i})=\pi(w_{2i+1})=2$, and $\pi(v_{2i+1})=3$,
and no colour is available for $u_{2i+2}$.
\end{enumerate}
This completes the proof.
\end{proof}


\begin{figure}
\begin{center}
\begin{tikzpicture}[domain=0:12,x=1cm,y=1cm]
\ETIQUETTEA{0,2.83}{{\scriptsize $2i$}}; \ETIQUETTEA{1,2.8}{{\scriptsize $2i+1$}};
\ETIQUETTEA{-0.5,1.8}{{\footnotesize $u$}}; \ETIQUETTEA{-0.5,0.8}{{\footnotesize $v$}}; \ETIQUETTEA{-0.5,-0.2}{{\footnotesize $w$}}; 
\SOMMET{0,0}; \SOMMET{1,0}; 
\SOMMET{0,1}; \SOMMET{1,1}; 
\SOMMET{0,2}; \SOMMET{1,2}; 
\ARETEH{0,0}{1}; \ARETEH{0,1}{1}; \ARETEH{0,2}{1};
\ARETEVV{0,0}{2}; \ARETEVV{1,0}{2};
\ETIQUETTEB{0.5,-0.6}{(a)};
\end{tikzpicture}
\hskip 1cm
\begin{tikzpicture}[domain=0:12,x=1cm,y=1cm]
\ETIQUETTEA{0,2.83}{{\scriptsize $2i$}}; \ETIQUETTEA{1,2.8}{{\scriptsize $2i+1$}};
\ETIQUETTEA{2,2.8}{{\scriptsize $2i+2$}}; \ETIQUETTEA{3,2.8}{{\scriptsize $2i+3$}};
\SOMMET{0,0}; \SOMMET{1,0}; \SOMMET{2,0}; \SOMMET{3,0}; 
\SOMMET{0,1}; \SOMMET{1,1}; \SOMMET{2,1}; \SOMMET{3,1}; 
\SOMMET{0,2}; \SOMMET{1,2}; \SOMMET{2,2}; \SOMMET{3,2}; 
\ARETEH{0,0}{3}; \ARETEH{0,2}{3}; 
\ARETEH{0,1}{1}; \ARETEH{2,1}{1}; 
\ARETEVV{0,0}{2}; \ARETEVV{1,0}{2}; \ARETEVV{2,0}{2}; \ARETEVV{3,0}{2};  
\ETIQUETTEA{0,2.1}{4};
\ETIQUETTEA{1,2}{1/2};
\ETIQUETTEA{2,2}{};
\ETIQUETTEA{3,2}{};
\ETIQUETTEB{0,0}{3,2/2};
\ETIQUETTEB{1,0}{1/1};
\ETIQUETTEB{2,0}{?/?};
\ETIQUETTEB{3,0}{4};
\ETIQUETTEA{-0.4,1}{1/1};
\ETIQUETTEA{1.5,1}{2,3/3};
\ETIQUETTEA{1.8,1}{};
\ETIQUETTEA{3.2,1}{};
\node[below] at (1.5,-0.6) {(b)};
\end{tikzpicture}
\hskip 1cm
\begin{tikzpicture}[domain=0:12,x=1cm,y=1cm]
\ETIQUETTEA{0,2.83}{{\scriptsize $2i$}}; \ETIQUETTEA{1,2.8}{{\scriptsize $2i+1$}};
\ETIQUETTEA{2,2.8}{{\scriptsize $2i+2$}}; \ETIQUETTEA{3,2.8}{{\scriptsize $2i+3$}};
\SOMMET{0,0}; \SOMMET{1,0}; \SOMMET{2,0}; \SOMMET{3,0}; 
\SOMMET{0,1}; \SOMMET{1,1}; \SOMMET{2,1}; \SOMMET{3,1}; 
\SOMMET{0,2}; \SOMMET{1,2}; \SOMMET{2,2}; \SOMMET{3,2}; 
\ARETEH{0,0}{3}; \ARETEH{0,2}{3}; 
\ARETEH{0,1}{1}; \ARETEH{2,1}{1}; 
\ARETEVV{0,0}{2}; \ARETEVV{1,0}{2}; \ARETEVV{2,0}{2}; \ARETEVV{3,0}{2};  
\ETIQUETTEA{0,2}{1/2};
\ETIQUETTEA{1,2}{2/1};
\ETIQUETTEA{2,2}{1/?};
\ETIQUETTEA{3,2}{};
\ETIQUETTEB{0,0}{1/1};
\ETIQUETTEB{1,0}{3/2};
\ETIQUETTEB{2,0}{};
\ETIQUETTEB{3,0}{};
\ETIQUETTEA{-0.2,1}{4};
\ETIQUETTEA{1.35,1}{1/3};
\ETIQUETTEA{2.3,1}{?/};
\ETIQUETTEA{3.2,1}{4};
\node[below] at (1.5,-0.6) {(c)};
\end{tikzpicture}
\caption{The subgraph $G_i(r)$ and two configurations for the proof of Lemma~\ref{lem:G_i(r)}.}
\label{fig:G_i(r)}
\end{center}
\end{figure}
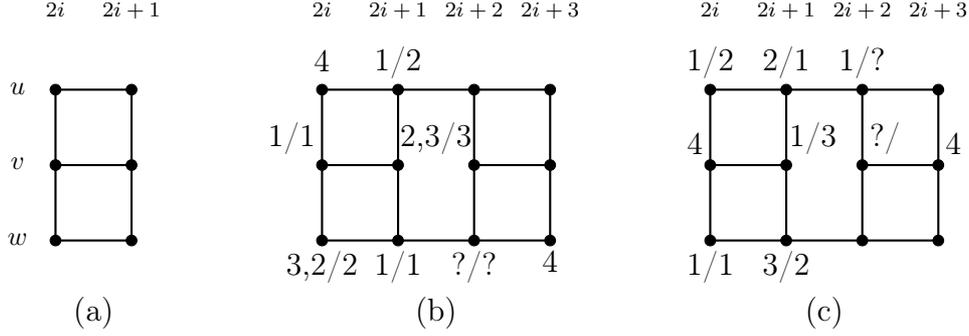


From Lemma~\ref{lem:G_i(r)}, it follows that every $G_i(r)$ must use colour 4 or~5, 
and that no two consecutive $G_i(r)$'s can use
the same colour from $\{4,5\}$. Therefore, $H(r)$ does not admit any packing 5-colouring
when $r$ is odd.

\begin{corollary}
For every odd integer $r$, $r\ge 3$, $\pcn(H(r))>5$.
\label{cor:H(r)-odd}
\end{corollary}


\begin{figure}
\begin{center}
\begin{tikzpicture}[domain=0:12,x=1cm,y=1cm]

\SOMMET{0,0}; \SOMMET{1,0}; \SOMMET{2,0}; \SOMMET{3,0}; 
\SOMMET{0,1}; \SOMMET{1,1}; \SOMMET{2,1}; \SOMMET{3,1}; 
\SOMMET{0,2}; \SOMMET{1,2}; \SOMMET{2,2}; \SOMMET{3,2}; 

\ARETEH{0,0}{3}; \ARETEH{0,2}{3}; 
\ARETEH{0,1}{1}; \ARETEH{2,1}{1}; 
\ARETEVV{0,0}{2}; \ARETEVV{1,0}{2}; \ARETEVV{2,0}{2}; \ARETEVV{3,0}{2}; 
\draw[thick] (0,2) .. controls (0.5,2.8) and (2.5,2.8) .. (3,2);
\draw[thick] (0,0) .. controls (0.5,-0.8) and (2.5,-0.8) .. (3,0);

\ETIQUETTEA{0,2}{1};
\ETIQUETTEA{1,2}{2};
\ETIQUETTEA{2,2}{1};
\ETIQUETTEA{3,2}{3};
\ETIQUETTEB{0,0}{1};
\ETIQUETTEB{1,0}{3};
\ETIQUETTEB{2,0}{1};
\ETIQUETTEB{3,0}{2};
\ETIQUETTEA{-0.2,1}{4};
\ETIQUETTEA{1.2,1}{1};
\ETIQUETTEA{1.8,1}{5};
\ETIQUETTEA{3.2,1}{1};

\node[below] at (1.5,-1) {(a) A packing 5-colouring pattern for $H(r)$, $r$ even, $r\ge 2$};
\end{tikzpicture}

\begin{tikzpicture}[domain=0:12,x=1cm,y=1cm]
\SOMMET{0,0}; \SOMMET{1,0}; \SOMMET{2,0}; \SOMMET{3,0}; \SOMMET{4,0}; 
\SOMMET{5,0}; \SOMMET{6,0}; \SOMMET{7,0}; \SOMMET{10,0}; \SOMMET{11,0}; 
\SOMMET{0,1}; \SOMMET{1,1}; \SOMMET{2,1}; \SOMMET{3,1}; \SOMMET{4,1}; 
\SOMMET{5,1}; \SOMMET{6,1}; \SOMMET{7,1}; \SOMMET{10,1}; \SOMMET{11,1}; 
\SOMMET{0,2}; \SOMMET{1,2}; \SOMMET{2,2}; \SOMMET{3,2}; \SOMMET{4,2}; 
\SOMMET{5,2}; \SOMMET{6,2}; \SOMMET{7,2}; \SOMMET{10,2}; \SOMMET{11,2}; 

\ARETEH{0,0}{7}; \ARETEH{0,2}{7}; \ARETEH{10,0}{1}; \ARETEH{10,2}{1}; 
\ARETEH{0,1}{1}; \ARETEH{2,1}{1}; \ARETEH{4,1}{1}; \ARETEH{6,1}{1}; \ARETEH{10,1}{1}; 
\ARETEVV{0,0}{2}; \ARETEVV{1,0}{2}; \ARETEVV{2,0}{2}; \ARETEVV{3,0}{2}; \ARETEVV{4,0}{2}; 
\ARETEVV{5,0}{2}; \ARETEVV{6,0}{2}; \ARETEVV{7,0}{2}; \ARETEVV{10,0}{2}; \ARETEVV{11,0}{2}; 
\draw[thick] (0,2) .. controls (1,3.3) and (10,3.3) .. (11,2);
\draw[thick] (0,0) .. controls (1,-1.3) and (10,-1.3) .. (11,0);

\draw[thick,dotted] (8,0) -- (9.5,0);
\draw[thick,dotted] (8,1) -- (9.5,1);
\draw[thick,dotted] (8,2) -- (9.5,2);

\ETIQUETTEA{0,2}{1};
\ETIQUETTEA{1,2}{2};
\ETIQUETTEA{2,2}{1};
\ETIQUETTEA{3,2}{3};
\ETIQUETTEA{4,2}{1};
\ETIQUETTEA{5,2}{2};
\ETIQUETTEA{6,2}{1};
\ETIQUETTEA{7,2}{3};
\ETIQUETTEA{10,2}{2};
\ETIQUETTEA{11,2}{6};
\ETIQUETTEB{0,0}{2};
\ETIQUETTEB{1,0}{3};
\ETIQUETTEB{2,0}{1};
\ETIQUETTEB{3,0}{2};
\ETIQUETTEB{4,0}{1};
\ETIQUETTEB{5,0}{3};
\ETIQUETTEB{6,0}{1};
\ETIQUETTEB{7,0}{2};
\ETIQUETTEB{10,0}{7};
\ETIQUETTEB{11,0}{1};
\ETIQUETTEA{-0.2,1}{5};
\ETIQUETTEA{1.2,1}{1};
\ETIQUETTEA{1.8,1}{4};
\ETIQUETTEA{3.2,1}{1};
\ETIQUETTEA{3.8,1}{5};
\ETIQUETTEA{5.2,1}{1};
\ETIQUETTEA{5.8,1}{4};
\ETIQUETTEA{7.2,1}{1};
\ETIQUETTEA{9.8,1}{1};
\ETIQUETTEA{11.2,1}{4};

\draw[thick,dashed] (3.5,-0.6) -- (7.5,-0.6);
\draw[thick,dashed] (3.5,2.6) -- (7.5,2.6);
\draw[thick,dashed] (3.5,-0.6) -- (3.5,2.6);
\draw[thick,dashed] (7.5,-0.6) -- (7.5,2.6);

\node[below] at (5.5,-1.5) {(b) A packing 7-colouring pattern for $H(r)$, $r$ odd, $r\ge 3$};
\end{tikzpicture}

\caption{Packing colouring patterns for H-graphs.}
\label{fig:H-graphs}
\end{center}
\end{figure}


We are now able to prove the main result of this section.

\begin{theorem}
For every integer $r\ge 2$, $\pcn(H(r))=5$ if $r$ is even, and $6\le\pcn(H(r))\le 7$ if $r$ is odd.
\label{th:H(r)}
\end{theorem}

\begin{proof}
We consider two cases, according to the parity of $r$.
\begin{enumerate}
\item {\em $r$ is even}.\\
Since $H(r)$ contains the corona graph $C_{6}\odot K_1$ as a subgraph
(consider for instance the 6-cycle $u_1v_1w_1w_2v_2u_2$),
we get $\pcn(H(r))\ge 5$ by Theorem~\ref{th:CrCn} and Proposition~\ref{prop:subgraph}.
A packing 5-colouring of $H(r)$ is then obtained by repeating the pattern 
depicted in Figure~\ref{fig:H-graphs}(a), and thus $\pcn(H(r))=5$.

\item {\em $r$ is odd}.\\
From Corollary~\ref{cor:H(r)-odd}, we get $\pcn(H(r))\ge 6$.
A packing 7-colouring of $H(r)$ is described in Figure~\ref{fig:H-graphs}(b), where
the circular pattern (surrounded by the dashed box) is repeated $\frac{r-3}{2}$ times.
This gives $\pcn(H(r))\le 7$.
\end{enumerate}
This concludes the proof.
\end{proof}


\section{Generalised H-graphs}
\label{sec:geneHgraphs}

We now consider a natural extension of H-graphs.
For every integer $r\ge 2$,
the \emph{generalised H-graph} $H^\ell(r)$ with $\ell$ levels, $\ell\ge 1$,  
is the 3-regular graph of order $2r(\ell+2)$, with vertex set
$$V(H^\ell(r))=\{u^i_j:0\le i\le \ell+1, 0\le j\le 2r-1\}$$
and edge set (subscripts are taken modulo $2r$)
 $$\begin{array}{rcl}
 E(H^\ell(r)) & = & \{(u^0_j,u^0_{j+1}),\ (u^{\ell+1}_j,u^{\ell+1}_{j+1}):0\le j\le 2r-1\}\\
 & & \cup \ \{(u^i_{2j},u^i_{2j+1}):1\le i\le \ell,\ 0\le j\le r-1\}\\
 & & \cup\ \{(u^i_j,u^{i+1}_j): 0\le i\le \ell,\ 0\le j\le 2r-1\}.
 \end{array}$$

Figure~\ref{fig:H34} depicts the generalised H-graph with three levels $H^3(4)$.
Note that generalised H-graphs with one level are precisely H-graphs.


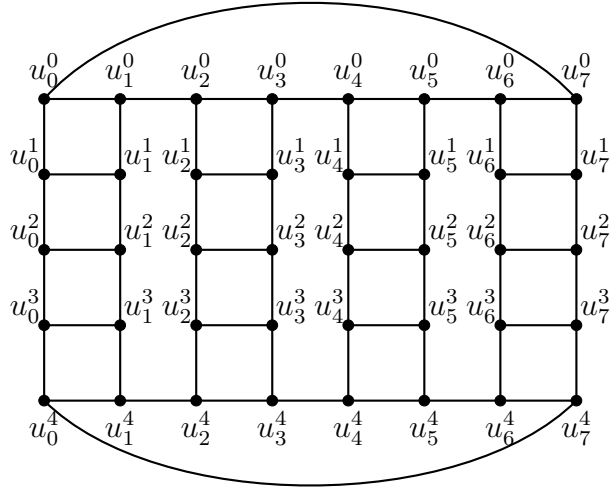
\begin{figure} 
\begin{center}
\begin{tikzpicture}[domain=0:12,x=1cm,y=1cm]


\SOMMET{2.5,-5};\SOMMET{2.5,-4}; \SOMMET{2.5,-3};
\SOMMET{3.5,-5};\SOMMET{3.5,-4}; \SOMMET{3.5,-3};
\SOMMET{4.5,-5};\SOMMET{4.5,-4}; \SOMMET{4.5,-3};
\SOMMET{5.5,-5};\SOMMET{5.5,-4}; \SOMMET{5.5,-3};
\SOMMET{6.5,-5};\SOMMET{6.5,-4}; \SOMMET{6.5,-3};
\SOMMET{7.5,-5};\SOMMET{7.5,-4}; \SOMMET{7.5,-3};
\SOMMET{8.5,-5};\SOMMET{8.5,-4}; \SOMMET{8.5,-3};
\SOMMET{9.5,-5};\SOMMET{9.5,-4}; \SOMMET{9.5,-3};

\SOMMET{2.5,-6};\SOMMET{2.5,-7}; 
\SOMMET{3.5,-6};\SOMMET{3.5,-7}; 
\SOMMET{4.5,-6};\SOMMET{4.5,-7}; 
\SOMMET{5.5,-6};\SOMMET{5.5,-7}; 
\SOMMET{6.5,-6};\SOMMET{6.5,-7}; 
\SOMMET{7.5,-6};\SOMMET{7.5,-7}; 
\SOMMET{8.5,-6};\SOMMET{8.5,-7}; 
\SOMMET{9.5,-6};\SOMMET{9.5,-7}; 


\ETIQUETTE{2.5,-2.25}{$u^0_0$}; 
\ETIQUETTE{3.5,-2.25}{$u^0_1$}; 
\ETIQUETTE{4.5,-2.25}{$u^0_2$}; 
\ETIQUETTE{5.5,-2.25}{$u^0_3$}; 
\ETIQUETTE{6.5,-2.25}{$u^0_4$}; 
\ETIQUETTE{7.5,-2.25}{$u^0_5$}; 
\ETIQUETTE{8.5,-2.25}{$u^0_6$}; 
\ETIQUETTE{9.5,-2.25}{$u^0_7$}; 

\ETIQUETTE{2.5,-7}{$u^4_0$}; 
\ETIQUETTE{3.5,-7}{$u^4_1$}; 
\ETIQUETTE{4.5,-7}{$u^4_2$}; 
\ETIQUETTE{5.5,-7}{$u^4_3$}; 
\ETIQUETTE{6.5,-7}{$u^4_4$}; 
\ETIQUETTE{7.5,-7}{$u^4_5$}; 
\ETIQUETTE{8.5,-7}{$u^4_6$}; 
\ETIQUETTE{9.5,-7}{$u^4_7$}; 

\ETIQUETTE{2.25,-3.38}{$u^1_0$};
\ETIQUETTE{3.75,-3.38}{$u^1_1$};
\ETIQUETTE{4.25,-3.38}{$u^1_2$};
\ETIQUETTE{5.75,-3.38}{$u^1_3$};
\ETIQUETTE{6.25,-3.38}{$u^1_4$};
\ETIQUETTE{7.75,-3.38}{$u^1_5$};
\ETIQUETTE{8.25,-3.38}{$u^1_6$};
\ETIQUETTE{9.75,-3.38}{$u^1_7$};

\ETIQUETTE{2.25,-4.38}{$u^2_0$};
\ETIQUETTE{3.75,-4.38}{$u^2_1$};
\ETIQUETTE{4.25,-4.38}{$u^2_2$};
\ETIQUETTE{5.75,-4.38}{$u^2_3$};
\ETIQUETTE{6.25,-4.38}{$u^2_4$};
\ETIQUETTE{7.75,-4.38}{$u^2_5$};
\ETIQUETTE{8.25,-4.38}{$u^2_6$};
\ETIQUETTE{9.75,-4.38}{$u^2_7$};

\ETIQUETTE{2.25,-5.38}{$u^3_0$};
\ETIQUETTE{3.75,-5.38}{$u^3_1$};
\ETIQUETTE{4.25,-5.38}{$u^3_2$};
\ETIQUETTE{5.75,-5.38}{$u^3_3$};
\ETIQUETTE{6.25,-5.38}{$u^3_4$};
\ETIQUETTE{7.75,-5.38}{$u^3_5$};
\ETIQUETTE{8.25,-5.38}{$u^3_6$};
\ETIQUETTE{9.75,-5.38}{$u^3_7$};

\ARETEH{2.5,-3}{7};
\ARETEH{2.5,-7}{7};

\ARETEVV{2.5,-7}{4};
\ARETEVV{3.5,-7}{4};
\ARETEVV{4.5,-7}{4};
\ARETEVV{5.5,-7}{4};
\ARETEVV{6.5,-7}{4};
\ARETEVV{7.5,-7}{4};
\ARETEVV{8.5,-7}{4};
\ARETEVV{9.5,-7}{4};

\ARETEH{2.5,-4}{1};
\ARETEH{4.5,-4}{1};
\ARETEH{6.5,-4}{1};
\ARETEH{8.5,-4}{1};

\ARETEH{2.5,-5}{1};
\ARETEH{4.5,-5}{1};
\ARETEH{6.5,-5}{1};
\ARETEH{8.5,-5}{1};

\ARETEH{2.5,-6}{1};
\ARETEH{4.5,-6}{1};
\ARETEH{6.5,-6}{1};
\ARETEH{8.5,-6}{1};

\draw[thick] (2.5,-3) .. controls (4,-1.3) and (8,-1.3) .. (9.5,-3);
\draw[thick] (2.5,-7) .. controls (4,-8.5) and (8,-8.5) .. (9.5,-7);

\end{tikzpicture}
\caption{The generalised H-graph $H^3(4)$.}
\label{fig:H34}
\end{center}
\end{figure}


The three following lemmas will be useful for determining the packing chromatic number of generalised H-graphs.

\begin{lemma}
For every pair of integers $\ell\ge 3$ and $r\ge 3$,
let $H^\ell(r)$ be a generalised H-graph with $\pcn(H^\ell(r))\le 5$ and let $\pi$ be a
packing $5$-colouring of $H^\ell(r)$.
For every edge
$u^i_{2j}u^i_{2j+1}$, $1\le i\le \ell$, $0\le j\le r-1$,
with $\pi(u^i_{2j})\neq 1$ and $\pi(u^i_{2j+1})\neq 1$,
either $u^i_{2j}$ or $u^i_{2j+1}$ has colour~$2$ and its three neighbours have colours $3$, $4$ and~$5$.
\label{lem:1altV}
\end{lemma}

\begin{proof}
%
We first claim that every such edge $u^i_{2j}u^i_{2j+1}$
belongs to a subgraph of $H^\ell(r)$ isomorphic to the graph $X$
depicted in Figure~\ref{fig:graphX}, in such a way that $u^i_{2j}u^i_{2j+1}$ corresponds to one of the
edges $u_3v_3$, $u_4v_4$ or $u_5v_5$ of $X$.
Indeed, consider first the ``extremal'' case of $H^3(3)$, and observe that $X$
is a subgraph of the subgraph of $H^3(3)$ induced by the set of vertices
$$\{u^0_0,\dots,u^0_5\}\ \cup\ \{u^4_0,\dots,u^4_5\}\ \cup\ \{u^1_2,u^1_3,u^2_2,u^2_3,u^3_2,u^3_3\}.$$
Our claim then follows for $H^3(3)$ thanks to its symmetries.

It is now easy to see that our claim holds for every generalised H-graph $H^\ell(r)$
with $\ell,r\ge 3$.
The result then follows by Lemma~\ref{lem:graphX}.
\end{proof}

From Lemma~\ref{lem:1altV}, it follows that if $\pi(u^i_{2j})\neq 1$ and $\pi(u^i_{2j+1})\neq 1$ 
for some edge $u^i_{2j}u^i_{2j+1}$ of $H^\ell(r)$, $1\le i\le \ell$, $0\le j\le r-1$,
then the colour~2 has to be used on this edge and, since the neighbours of the 2-coloured vertex
are coloured with 3, 4 and~5, the colour~2 can be replaced by colour~1.
Therefore, we get the following corollary. 

\begin{corollary}
For every pair of integers $\ell\ge 3$ and $r\ge 3$,
if $H^\ell(r)$ is a generalised H-graph with $\pcn(H^\ell(r))\le 5$, then there exists a
packing $5$-colouring of $H^\ell(r)$ such that, for every pair of integers $i$ and $j$, $1\le i\le \ell$, $0\le j\le r-1$,
the colour~$1$ is used on the edge $u^i_{2j}u^i_{2j+1}$ of $H^\ell(r)$.
\label{cor:1altV}
\end{corollary}

\begin{lemma}
For every pair of integers $\ell\ge 3$ and $r\ge 3$,
let $H^\ell(r)$ be a generalised H-graph with $\pcn(H^\ell(r))\le 5$ and $\pi$ be a
packing $5$-colouring of $H^\ell(r)$.
For every $j$, $0\le j\le 2r-1$, $\pi$ must assign colour~$1$ to one vertex of 
each of the edges $u^0_ju^0_{j+1}$ and $u^{\ell+1}_ju^{\ell+1}_{j+1}$ (subscripts are taken modulo $2r$).
\label{lem:1altH}
\end{lemma}

\begin{proof}
The proof is done by case analysis and is given in Appendix~\ref{ap:1altH}.
\end{proof}

Let $H^\ell(r)$ be a generalised H-graph with $\pcn(H^\ell(r))\le 5$.
From Corollary~\ref{cor:1altV} and Lemma~\ref{lem:1altH},
it follows that one can always produce a packing 5-colouring of $H^\ell(r)$
that uses colour~1 on each edge $u^i_{2j}u^i_{2j+1}$ of $H^\ell(r)$,
$0\le i\le \ell+1$, $0\le j\le r-1$.
Since adjacent vertices cannot be assigned the same colour and $H^\ell(r)$ is a bipartite graph,
we get the following corollary. 

\begin{corollary}
For every pair of integers $\ell\ge 3$ and $r\ge 3$,
if $H^\ell(r)$ is a generalised H-graph with $\pcn(H^\ell(r))\le 5$, then there exists a
packing $5$-colouring of $H^\ell(r)$ such that the colour~$1$ is used on each edge of $H^\ell(r)$.
\label{coro:bipaH}
\end{corollary}

\begin{lemma}
For every pair of integers $\ell\ge 3$ and $r\ge 3$,
if $H^\ell(r)$ is a generalised H-graph with $\pcn(H^\ell(r))\le 5$, then there exists a
packing $5$-colouring $\pi$ of $H^\ell(r)$ such that 
$\pi(u^0_j)\notin\{4,5\}$ and $\pi(u^{\ell+1}_j)\notin\{4,5\}$ for every $j$, $0\le j\le 2r-1$.
\label{lem:level0-no45}
\end{lemma}


\begin{figure} 
\begin{center}
\begin{tikzpicture}[domain=0:12,x=1cm,y=1cm]
\SOMMET{0,0}; \BIGSOMMET{1,0}; \SOMMET{2,0}; \BIGSOMMET{3,0}; \SOMMET{4,0}; 
\SOMMET{1,-1}; \BIGSOMMET{2,-1}; \SOMMET{3,-1}; 
\SOMMET{2,-2}; \BIGSOMMET{3,-2}; 
\SOMMET{3,-3}; 
\ETIQUETTEA{0,0}{2,3};
\ETIQUETTEA{1,0.1}{$x$};
\ETIQUETTEA{2,0}{5};
\ETIQUETTEA{4,0}{4};
\ETIQUETTEB{1,-1}{2,3};
\ETIQUETTEB{2,-2}{2,3};
\ETIQUETTEA{1.7,-1.2}{$y$};
\ETIQUETTEA{3.4,-1.2}{2,3};
\ETIQUETTEA{3.3,-3.2}{$z$};

\ARETEH{0,0}{4};
\ARETEH{2,-1}{1};  
\ARETEH{2,-2}{1}; 
\ARETEVV{1,-1}{1};
\ARETEVV{2,-2}{2};
\ARETEVV{3,-3}{3};
\end{tikzpicture}
\caption{The subgraph $Y$ of $H^\ell(r)$.}
\label{fig:graphY}
\end{center}
\end{figure}
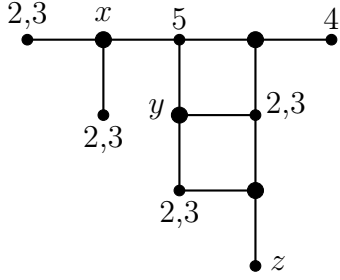


\begin{proof}
Let $\pi$ be a packing 5-colouring of $H^\ell(r)$ such that colour~1 is used on each edge of $H^\ell(r)$
(the existence of such a colouring is ensured by Corollary~\ref{coro:bipaH}).
Thanks to the symmetries of $H^\ell(r)$, it suffices to prove the result for any vertex $u^0_{2j}$, $0\le j\le r-1$.
Suppose to the contrary that $\pi(u^0_{2j})\in\{4,5\}$ for some $j$, $0\le j\le r-1$.
We have two cases to consider.

\begin{enumerate}
\item $\pi(u^0_{2j})=4$.\\
Let $Y$ be the subgraph of $H^\ell(r)$ depicted in Figure~\ref{fig:graphY},
where the vertex $u^0_{2j}$ is the unique vertex with colour~4, and
vertices with colour~1 are drawn as ``big vertices''.
Observe that the three neighbours of $x$, as well as the three neighbours of $y$,
must use colours 2, 3 and~5.
Therefore, the common neighbour of $x$ and $y$ must be assigned colour~5.
It then follows that no colour is available for $z$.

\item $\pi(u^0_{2j})=5$.\\
The proof is similar to the proof of the previous case, by switching colours 4 and~5.

\end{enumerate}
This completes the proof.
\end{proof}

Let $H^\ell(r)$ be a generalised H-graph satisfying the hypothesis of Lemma~\ref{lem:level0-no45},
and $\pi$ be a packing 5-colouring of $H^\ell(r)$.
From Lemma~\ref{lem:level0-no45}, it follows that the restriction of $\pi$ to the $2r$-cycle induced 
by the set of vertices $\{u^0_j\ |\ 0\le j\le 2r-1\}$ is a packing 3-colouring.
It is not difficult to check (see~\cite{LBS17}) that a $2r$-cycle admits a packing 3-colouring
if and only if $r$ is even.
Therefore, we get the following corollary. 

\begin{corollary}
For every pair of integers $\ell\ge 3$ and $r\ge 3$, $r$ odd, $\pcn(H^\ell(r))\ge 6$.
\label{cor:r-odd-6}
\end{corollary}

We are now able to prove the main results of this section.
We first consider the case of generalised H-graphs $H^\ell(r)$ with $\ell\notin\{2,5\}$.

\begin{theorem}
For every pair of integers $\ell\ge 3$, $\ell\neq 5$, and $r\ge 2$, 
$$\pcn(H^\ell(r))=
\left\{
    \begin{array}{ll}
     5 & \hbox{if $r$ is even,} \\
     6 & \hbox{otherwise.}
    \end{array}
\right.$$
\label{th:Hell(r)}
\end{theorem}

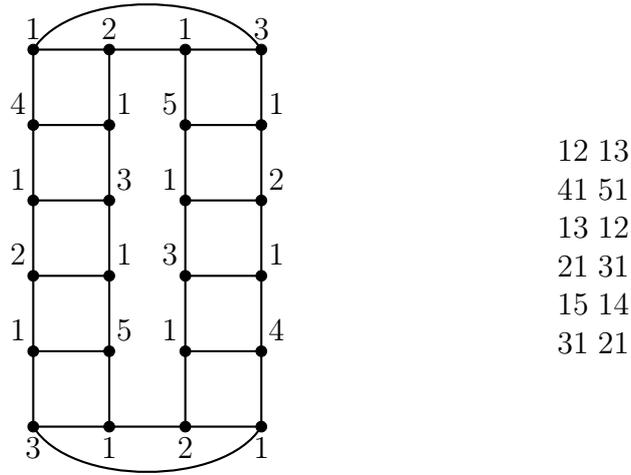
\begin{figure}
\begin{center}
\begin{tikzpicture}[domain=0:12,x=1cm,y=1cm]

\SOMMET{0,0}; \SOMMET{1,0}; \SOMMET{2,0}; \SOMMET{3,0}; 
\SOMMET{0,1}; \SOMMET{1,1}; \SOMMET{2,1}; \SOMMET{3,1}; 
\SOMMET{0,2}; \SOMMET{1,2}; \SOMMET{2,2}; \SOMMET{3,2}; 
\SOMMET{0,3}; \SOMMET{1,3}; \SOMMET{2,3}; \SOMMET{3,3}; 
\SOMMET{0,4}; \SOMMET{1,4}; \SOMMET{2,4}; \SOMMET{3,4}; 
\SOMMET{0,5}; \SOMMET{1,5}; \SOMMET{2,5}; \SOMMET{3,5}; 

\ARETEH{0,0}{3}; \ARETEH{0,5}{3}; 
\ARETEH{0,1}{1}; \ARETEH{2,1}{1}; 
\ARETEH{0,2}{1}; \ARETEH{2,2}{1}; 
\ARETEH{0,3}{1}; \ARETEH{2,3}{1}; 
\ARETEH{0,4}{1}; \ARETEH{2,4}{1}; 
\ARETEVV{0,0}{5}; \ARETEVV{1,0}{5}; \ARETEVV{2,0}{5}; \ARETEVV{3,0}{5}; 
\draw[thick] (0,5) .. controls (0.5,5.8) and (2.5,5.8) .. (3,5);
\draw[thick] (0,0) .. controls (0.5,-0.8) and (2.5,-0.8) .. (3,0);

\ETIQUETTEA{0,5}{1};
\ETIQUETTEA{1,5}{2};
\ETIQUETTEA{2,5}{1};
\ETIQUETTEA{3,5}{3};
\ETIQUETTEB{0,0}{3};
\ETIQUETTEB{1,0}{1};
\ETIQUETTEB{2,0}{2};
\ETIQUETTEB{3,0}{1};

\ETIQUETTEA{-0.2,4}{4};
\ETIQUETTEA{1.2,4}{1};
\ETIQUETTEA{1.8,4}{5};
\ETIQUETTEA{3.2,4}{1};

\ETIQUETTEA{-0.2,3}{1};
\ETIQUETTEA{1.2,3}{3};
\ETIQUETTEA{1.8,3}{1};
\ETIQUETTEA{3.2,3}{2};

\ETIQUETTEA{-0.2,2}{2};
\ETIQUETTEA{1.2,2}{1};
\ETIQUETTEA{1.8,2}{3};
\ETIQUETTEA{3.2,2}{1};

\ETIQUETTEA{-0.2,1}{1};
\ETIQUETTEA{1.2,1}{5};
\ETIQUETTEA{1.8,1}{1};
\ETIQUETTEA{3.2,1}{4};

\end{tikzpicture}
\hskip 3cm
\begin{tikzpicture}[domain=0:12,x=1cm,y=1cm]
\ETIQUETTEA{0,0}{$\begin{array}{c}
\Xud \Xsp \Xut\\
\Xqu \Xsp \Xcu\\
\Xut \Xsp \Xud\\
\Xdu \Xsp \Xtu\\
\Xuc \Xsp \Xuq\\
\Xtu \Xsp \Xdu
\end{array}$};
\ETIQUETTEA{0,-1.5}{$\ $};
\end{tikzpicture}

\caption{A packing 5-colouring of $H^4(2)$ and its corresponding colouring pattern.}
\label{fig:H4(2)}
\end{center}
\end{figure}

\begin{figure}
\begin{center}
$\begin{array}{ccccc}
\begin{array}{c}
\hline\hline
\Xud \Xsp \Xut\\
\Xqu \Xsp \Xcu\\
\Xut \Xsp \Xud\\
\Xdu \Xsp \Xtu\\
\Xuc \Xsp \Xuq\\
\Xtu \Xsp \Xdu\\
\hline\hline
\Xuq \Xsp \Xuc\\
\Xdu \Xsp \Xtu\\ \\ \\ \\ \\ \\ \\ \\ \\
\end{array} 
& \begin{array}{c}
\hline\hline
\Xud \Xsp \Xut\\
\Xqu \Xsp \Xcu\\
\Xut \Xsp \Xud\\
\Xdu \Xsp \Xtu\\
\Xuc \Xsp \Xuq\\
\Xtu \Xsp \Xdu\\
\hline\hline
\Xud \Xsp \Xut\\
\Xqu \Xsp \Xcu\\
\Xut \Xsp \Xud\\ \\ \\ \\ \\ \\ \\ \\
\end{array} 
& \begin{array}{c}
\hline\hline
\Xud \Xsp \Xut\\
\Xqu \Xsp \Xcu\\
\Xut \Xsp \Xud\\
\Xdu \Xsp \Xtu\\
\Xuc \Xsp \Xuq\\
\Xtu \Xsp \Xdu\\
\hline\hline
\Xud \Xsp \Xut\\
\Xqu \Xsp \Xcu\\
\Xut \Xsp \Xud\\
\Xcu \Xsp \Xqu\\
\Xud \Xsp \Xut\\
\Xtu \Xsp \Xdu\\
\Xuq \Xsp \Xuc\\
\Xdu \Xsp \Xtu\\
\Xuc \Xsp \Xuq\\
\Xtu \Xsp \Xdu
\end{array} 
& \begin{array}{c}
\hline\hline
\Xud \Xsp \Xut\\
\Xqu \Xsp \Xcu\\
\Xut \Xsp \Xud\\
\Xdu \Xsp \Xtu\\
\Xuc \Xsp \Xuq\\
\Xtu \Xsp \Xdu\\
\hline\hline
\Xud \Xsp \Xut\\
\Xqu \Xsp \Xcu\\
\Xut \Xsp \Xud\\
\Xcu \Xsp \Xqu\\
\Xud \Xsp \Xut\\ \\ \\ \\ \\ \\
\end{array}
& \begin{array}{c}
\hline\hline
\Xud \Xsp \Xut\\
\Xqu \Xsp \Xcu\\
\Xut \Xsp \Xud\\
\Xdu \Xsp \Xtu\\
\Xuc \Xsp \Xuq\\
\Xtu \Xsp \Xdu\\
\hline\hline
\Xuq \Xsp \Xuc\\
\Xdu \Xsp \Xtu\\
\Xut \Xsp \Xud\\
\Xcu \Xsp \Xqu\\
\Xud \Xsp \Xut\\
\Xqu \Xsp \Xcu\\
\Xut \Xsp \Xud\\ \\ \\ \\
\end{array}
 \\
\ \\
\mbox{\small $\ell\equiv 0\pmod 6$} & \mbox{\small $\ell\equiv 1\pmod 6$} & \mbox{\small $\ell\equiv 2\pmod 6$} & \mbox{\small $\ell\equiv 3\pmod 6$} & \mbox{\small $\ell\equiv 5\pmod 6$}
\end{array}$
\caption{Colouring patterns for $H^\ell(r)$, $r$ even.}
\label{fig:patternsHeven}
\end{center}
\end{figure}

\begin{proof}
We consider two cases, according to the parity of $r$.
\begin{enumerate}
\item {\it $r$ is even}.\\
Since the corona graph $C_{2\ell+4}\odot K_1$ is a subgraph of $H^\ell(r)$ for every $r\ge 2$
(consider the cycle of length $2\ell+4$ induced by 
the set of vertices $\{u_1^i|0\le i\le\ell+1\} \cup \{u_2^i|0\le i\le\ell+1\}$), 
we get $\pcn(H^\ell(r))\ge \pcn(C_{2\ell+4}\odot K_1)=5$ by Theorem~\ref{th:CrCn} and Proposition~\ref{prop:subgraph}.

We now prove $\pcn(H^\ell(r))\le 5$. 
Figure~\ref{fig:H4(2)} depicts a packing 5-colouring of $H^4(2)$, together with its
corresponding colouring pattern.
It can easily be checked that this $(6\times 4)$-pattern is periodic, that is,
can be repeated, both vertically and horizontally, to produce a packing 5-colouring
of any generalised H-graph of the form $H^{6i+4}(2j)$, with $i\ge 0$ and $j\ge 1$.

If $\ell\not\equiv 4\pmod 6$, we use the colouring patterns depicted in Figure~\ref{fig:patternsHeven},
depending on the value of $\ell$ modulo 6.
The upper six rows of each colouring pattern, surrounded by double lines, 
can be repeated as many times as required, or even deleted when $\ell\equiv 1,2,3\pmod 6$.
Therefore, these colouring patterns give us a packing 5-colouring of any generalised H-graph
of the form $H^\ell(2)$, for every $\ell\ge 3$, $\ell\neq 5$.
It is again easy to check that each of these colouring patterns is ``horizontally periodic'', that is,
can be horizontally repeated in order to get a packing 5-colouring of any generalised H-graph
of the form $H^\ell(r)$, for every $\ell\ge 3$, $\ell\neq 5$, 
$\ell\not\equiv 4\pmod 6$, and even $r$.


\begin{figure}
\begin{center}
$\begin{array}{ccc}
\begin{array}{c}
\Xut \Xsp \Xud \Xsp \Xus\\
\Xdu \Xsp \Xcu \Xsp \Xqu\\
\Xuq \Xsp \Xdt \Xsp \Xud\\
\Xtu \Xsp \Xuq \Xsp \Xtu\\
\Xud \Xsp \Xsu \Xsp \Xdc\\
\ 
\end{array}
& & \begin{array}{||c||c}
\Xud \Xsp \Xut & \Xud \Xsp \Xuc \Xsp \Xud \Xsp \Xsu \Xsp \Xdt\\
\Xqu \Xsp \Xcu & \Xqu \Xsp \Xtu \Xsp \Xtu \Xsp \Xuq \Xsp \Xcu\\
\Xut \Xsp \Xud & \Xut \Xsp \Xuq \Xsp \Xuq \Xsp \Xdt \Xsp \Xud\\
\Xcu \Xsp \Xqu & \Xdu \Xsp \Xdu \Xsp \Xdu \Xsp \Xcu \Xsp \Xqu\\
\Xud \Xsp \Xut & \Xuc \Xsp \Xus \Xsp \Xut \Xsp \Xud \Xsp \Xut\\
\multicolumn{2}{c}{\ }
\end{array} \\
\ 
\mbox{(a)} & \ \ \ \ & \mbox{(b)}
\end{array}$
\caption{Colouring patterns for $H^3(3)$ and for $H^3(r)$,  $r\ge 5$, $r$ odd.}
\label{fig:patternsHodd-particular}
\end{center}
\end{figure}


\begin{figure}
\begin{center}
$\begin{array}{ccccc}
%
%
\begin{array}{c}
\Xut \Xsp \Xud \Xsp \Xus \\
\Xdu \Xsp \Xcu \Xsp \Xtu \\
\Xuq \Xsp \Xuq \Xsp \Xuc \\
\Xtu \Xsp \Xtu \Xsp \Xdu \\
\Xuc \Xsp \Xud \Xsp \Xuq \\
\Xdu \Xsp \Xsu \Xsp \Xtu \\
\multicolumn{1}{c}{} \\
\multicolumn{1}{c}{} \\
\multicolumn{1}{c}{}
\end{array} 
%
%
& & \begin{array}{||c||c}
\Xud \Xsp \Xut & \Xud \Xsp \Xuq\Xsp\Xut \Xsp \Xus \Xsp \Xut \\
\Xcu \Xsp \Xqu & \Xcu \Xsp \Xtu\Xsp\Xdu \Xsp \Xdu \Xsp \Xdu \\
\Xut \Xsp \Xud & \Xut \Xsp \Xud\Xsp\Xuc \Xsp \Xut \Xsp \Xuq \\
\Xdu \Xsp \Xtu & \Xdu \Xsp \Xcu\Xsp\Xtu \Xsp \Xcu \Xsp \Xtu \\
\Xuq \Xsp \Xuc & \Xuq \Xsp \Xut\Xsp\Xud \Xsp \Xuq \Xsp \Xuc \\
\Xtu \Xsp \Xdu & \Xtu \Xsp \Xdu\Xsp\Xsu \Xsp \Xtu \Xsp \Xdu \\
\multicolumn{2}{c}{} \\
\multicolumn{2}{c}{} \\
\multicolumn{2}{c}{}
\end{array} 
%
%
& & \begin{array}{c||c||c}
\Xud \Xsp \Xut & \Xud \Xsp \Xut & \Xus \\
\Xqu \Xsp \Xcu & \Xqu \Xsp \Xcu & \Xdu \\
\Xut \Xsp \Xud & \Xut \Xsp \Xud & \Xut \\
\Xdu \Xsp \Xtu & \Xdu \Xsp \Xtu & \Xcu \\
\Xuc \Xsp \Xuq & \Xuc \Xsp \Xuq & \Xud \\
\Xtu \Xsp \Xdu & \Xtu \Xsp \Xdu & \Xtu \\
\Xud \Xsp \Xuc & \Xuq \Xsp \Xuc & \Xuq \\
\Xsu \Xsp \Xtu & \Xdu \Xsp \Xtu & \Xdu \\
\multicolumn{3}{c}{}
\end{array} 
%
%
\ \\
\mbox{(a) $H^4(3)$} & & \mbox{(b) $H^4(r)$, $r\ge 5$, $r$ odd} & & \mbox{(c) $H^6(r)$, $r\ge 3$, $r$ odd}
\end{array}$

\vskip 1cm

$\begin{array}{ccc}
%
%
\begin{array}{c}
\Xud \Xsp \Xut \Xsp \Xuq \\
\Xtu \Xsp \Xcu \Xsp \Xdu \\
\Xus \Xsp \Xud \Xsp \Xuc \\
\Xdu \Xsp \Xtu \Xsp \Xtu \\
\Xuq \Xsp \Xuq \Xsp \Xus \\
\Xtu \Xsp \Xdu \Xsp \Xdu \\
\Xuc \Xsp \Xut \Xsp \Xut \\
\Xdu \Xsp \Xsu \Xsp \Xcu \\
\Xut \Xsp \Xud \Xsp \Xuq \\
\end{array} 
%
%
& \mbox{\hskip 2cm} & \begin{array}{||c||c}
\Xud \Xsp \Xut & \Xuc \Xsp \Xud \Xsp \Xus \\
\Xcu \Xsp \Xqu & \Xdu \Xsp \Xtu \Xsp \Xtu \\
\Xut \Xsp \Xud & \Xut \Xsp \Xuq \Xsp \Xud \\
\Xdu \Xsp \Xtu & \Xsu \Xsp \Xdu \Xsp \Xcu \\
\Xuq \Xsp \Xuc & \Xuq \Xsp \Xuc \Xsp \Xut \\
\Xtu \Xsp \Xdu & \Xtu \Xsp \Xtu \Xsp \Xdu \\
\Xud \Xsp \Xut & \Xud \Xsp \Xud \Xsp \Xuq \\
\Xcu \Xsp \Xqu & \Xcu \Xsp \Xqu \Xsp \Xtu \\
\Xut \Xsp \Xud & \Xut \Xsp \Xus \Xsp \Xud \\
\end{array} \\
\ \\
\mbox{(d) $H^7(3)$} & & \mbox{(e) $H^7(r)$, $r\ge 5$, $r$ odd}
\end{array}$
\caption{Colouring patterns for $H^4(r)$, $H^6(r)$ and $H^7(r)$,  $r\ge 3$, $r$ odd.}
\label{fig:patternsHodd-particular-bis}
\end{center}
\end{figure}

\item {\it $r$ is odd}.\\
The inequality $\pcn(H^\ell(r))\ge 6$ directly follows from Corollary~\ref{cor:r-odd-6}.
Therefore, we only need to prove the inequality $\pcn(H^\ell(r))\le 6$ (recall that $\ell\ge 3$ and $\ell\neq 5$).

We first consider a few particular cases.
A packing 6-colouring of $H^3(3)$ is depicted in Figure~\ref{fig:patternsHodd-particular}(a),
and a packing 6-colouring of $H^3(r)$, for every odd $r\ge 5$, 
is depicted in Figure~\ref{fig:patternsHodd-particular}(b) (the first four columns, surrounded by a double line,
are repeated $\frac{r-5}{2}$ times, and thus do not appear if $r=5$).
A packing 6-colouring of $H^4(3)$
is depicted in Figure~\ref{fig:patternsHodd-particular-bis}(a), 
and a packing 6-colouring of $H^4(r)$, for every odd $r\ge 5$,
is depicted in Figure~\ref{fig:patternsHodd-particular-bis}(b)  
(the first four columns are repeated $\frac{r-5}{2}$ times).
A packing 6-colouring of $H^6(r)$, for every odd $r\ge 3$,
is depicted in Figure~\ref{fig:patternsHodd-particular-bis}(c) 
(the four columns surrounded by a double line
are repeated $\frac{r-3}{2}$ times, and thus do not appear if $r=3$).
A packing 6-colouring of $H^7(3)$
is depicted in Figure~\ref{fig:patternsHodd-particular-bis}(d), 
and a packing 6-colouring of $H^7(r)$, for every odd $r\ge 5$,
is depicted in Figure~\ref{fig:patternsHodd-particular-bis}(e)
(the four columns surrounded by a double line,
are repeated $\frac{r-3}{2}$ times).


\begin{figure}
\begin{center}
$\begin{array}{ccccc}
%
%
\begin{array}{c}
\Xut \Xsp \Xud \Xsp \Xus \\
\Xdu \Xsp \Xcu \Xsp \Xqu \\
\Xuq \Xsp \Xut \Xsp \Xud \\
\Xtu \Xsp \Xdu \Xsp \Xtu \\
\Xud \Xsp \Xuq \Xsp \Xuc \\
\Xcu \Xsp \Xtu \Xsp \Xdu \\
\Xut \Xsp \Xud \Xsp \Xut \\
\Xdu \Xsp \Xcu \Xsp \Xqu \\
\Xuq \Xsp \Xdt \Xsp \Xud \\
\Xtu \Xsp \Xuq \Xsp \Xtu \\
\Xud \Xsp \Xsu \Xsp \Xdc \\
\\
\end{array} 
%
%
& & \begin{array}{||c||c}
\Xud \Xsp \Xut & \Xud \Xsp \Xuc \Xsp \Xud \Xsp \Xsu \Xsp \Xdt \\
\Xqu \Xsp \Xcu & \Xqu \Xsp \Xtu \Xsp \Xtu \Xsp \Xuq \Xsp \Xcu \\
\Xut \Xsp \Xud & \Xut \Xsp \Xuq \Xsp \Xuq \Xsp \Xdt \Xsp \Xud \\
\Xdu \Xsp \Xtu & \Xdu \Xsp \Xdu \Xsp \Xdu \Xsp \Xcu \Xsp \Xtu \\
\Xuc \Xsp \Xuq & \Xuc \Xsp \Xut \Xsp \Xut \Xsp \Xud \Xsp \Xuq \\
\Xtu \Xsp \Xdu & \Xtu \Xsp \Xcu \Xsp \Xcu \Xsp \Xtu \Xsp \Xdu \\
\Xud \Xsp \Xut & \Xud \Xsp \Xud \Xsp \Xud \Xsp \Xuq \Xsp \Xut \\
\Xqu \Xsp \Xcu & \Xqu \Xsp \Xtu \Xsp \Xtu \Xsp \Xdu \Xsp \Xcu \\
\Xut \Xsp \Xud & \Xut \Xsp \Xuq \Xsp \Xuq \Xsp \Xut \Xsp \Xud \\
\Xcu \Xsp \Xqu & \Xdu \Xsp \Xdu \Xsp \Xdu \Xsp \Xcu \Xsp \Xqu \\
\Xud \Xsp \Xut & \Xuc \Xsp \Xus \Xsp \Xut \Xsp \Xud \Xsp \Xut \\
\multicolumn{2}{c}{}
\end{array} 
%
%
& & \begin{array}{c||c||c}
\Xud \Xsp \Xut & \Xud \Xsp \Xut & \Xus \\
\Xqu \Xsp \Xcu & \Xqu \Xsp \Xcu & \Xdu \\
\Xut \Xsp \Xud & \Xut \Xsp \Xud & \Xut \\
\Xdu \Xsp \Xtu & \Xdu \Xsp \Xtu & \Xqu \\
\Xuc \Xsp \Xuq & \Xuc \Xsp \Xuq & \Xud \\
\Xtu \Xsp \Xdu & \Xtu \Xsp \Xdu & \Xtu \\
\Xud \Xsp \Xut & \Xud \Xsp \Xut & \Xuc \\
\Xqu \Xsp \Xcu & \Xqu \Xsp \Xcu & \Xdu \\
\Xut \Xsp \Xud & \Xut \Xsp \Xud & \Xuq \\
\Xdu \Xsp \Xtu & \Xdu \Xsp \Xtu & \Xtu \\
\Xuc \Xsp \Xuq & \Xuc \Xsp \Xuq & \Xud \\
\Xtu \Xsp \Xdu & \Xtu \Xsp \Xdu & \Xsu \\
\end{array} 
\\
\ \\
\ell=9,\ r=3 &  \Xsp \Xsp \Xsp & \ell=9,\ r\ge 5 &  \Xsp \Xsp \Xsp & \ell=10,\ r\ge 3
\end{array}$

\vskip 1cm

$\begin{array}{ccc}
%
%
\begin{array}{c}
\Xut \Xsp \Xud \Xsp \Xus \\
\Xqu \Xsp \Xcu \Xsp \Xtu \\
\Xud \Xsp \Xut \Xsp \Xud \\
\Xcu \Xsp \Xqu \Xsp \Xcu \\
\Xut \Xsp \Xud \Xsp \Xut \\
\Xdu \Xsp \Xtu \Xsp \Xdu \\
\Xuq \Xsp \Xuc \Xsp \Xuq \\
\Xtu \Xsp \Xdu \Xsp \Xsu \\
\Xuc \Xsp \Xuq \Xsp \Xud \\
\Xdu \Xsp \Xtu \Xsp \Xcu \\
\Xut \Xsp \Xud \Xsp \Xut \\
\Xsu \Xsp \Xcu \Xsp \Xdu \\
\Xud \Xsp \Xut \Xsp \Xuq \\
\end{array}
%
%
& & \begin{array}{c||c||c}
\Xut \Xsp \Xud & \Xut \Xsp \Xud & \Xus \\
\Xqu \Xsp \Xcu & \Xqu \Xsp \Xcu & \Xtu \\
\Xud \Xsp \Xut & \Xud \Xsp \Xut & \Xud \\
\Xcu \Xsp \Xqu & \Xcu \Xsp \Xqu & \Xcu \\
\Xut \Xsp \Xud & \Xut \Xsp \Xud & \Xut \\
\Xdu \Xsp \Xtu & \Xdu \Xsp \Xtu & \Xdu \\
\Xuq \Xsp \Xuc & \Xuq \Xsp \Xuc & \Xuq \\
\Xtu \Xsp \Xdu & \Xtu \Xsp \Xdu & \Xsu \\
\Xuc \Xsp \Xuq & \Xuc \Xsp \Xuq & \Xud \\
\Xdu \Xsp \Xtu & \Xdu \Xsp \Xtu & \Xcu \\
\Xut \Xsp \Xud & \Xut \Xsp \Xud & \Xut \\
\Xsu \Xsp \Xcu & \Xqu \Xsp \Xcu & \Xdu \\
\Xud \Xsp \Xut & \Xud \Xsp \Xut & \Xuq \\
\end{array} 
\\
\ \\
\ell=11,\ r=3 &  \Xsp \Xsp \Xsp & \ell=11,\ r\ge 5
\end{array}$
\caption{Colouring patterns for  $H^\ell(r)$, $9\le \ell\le 11$,  $r\ge 3$, $r$ odd.}
\label{fig:patternsHodd-mod6}
\end{center}
\end{figure}

\begin{figure}
\begin{center}
$\begin{array}{ccccc}
%
%
\begin{array}{c||c||c}
\Xud \Xsp \Xut & \Xud \Xsp \Xut & \Xus \\
\hline\hline
\Xqu \Xsp \Xcu & \Xqu \Xsp \Xcu & \Xdu \\
\Xut \Xsp \Xud & \Xut \Xsp \Xud & \Xut \\
\Xdu \Xsp \Xtu & \Xdu \Xsp \Xtu & \Xqu \\
\Xuc \Xsp \Xuq & \Xuc \Xsp \Xuq & \Xud \\
\Xtu \Xsp \Xdu & \Xtu \Xsp \Xdu & \Xtu \\
\Xud \Xsp \Xut & \Xud \Xsp \Xut & \Xuc \\
\hline\hline
\Xqu \Xsp \Xcu & \Xqu \Xsp \Xcu & \Xdu \\
\Xut \Xsp \Xud & \Xut \Xsp \Xud & \Xut \\
\Xdu \Xsp \Xqu & \Xcu \Xsp \Xqu & \Xqu \\
\Xuc \Xsp \Xut & \Xud \Xsp \Xut & \Xud \\
\Xtu \Xsp \Xdu & \Xtu \Xsp \Xdu & \Xsu \\
\Xud \Xsp \Xuc & \Xuq \Xsp \Xuc & \Xut \\
\Xqu \Xsp \Xtu & \Xdu \Xsp \Xtu & \Xdu \\
\multicolumn{2}{c}{}\\
\multicolumn{2}{c}{}
\end{array} 
%
%
& & \begin{array}{c||c||c}
\Xud \Xsp \Xut & \Xud \Xsp \Xut & \Xus \\
\hline\hline
\Xqu \Xsp \Xcu & \Xqu \Xsp \Xcu & \Xdu \\
\Xut \Xsp \Xud & \Xut \Xsp \Xud & \Xut \\
\Xdu \Xsp \Xtu & \Xdu \Xsp \Xtu & \Xqu \\
\Xuc \Xsp \Xuq & \Xuc \Xsp \Xuq & \Xud \\
\Xtu \Xsp \Xdu & \Xtu \Xsp \Xdu & \Xtu \\
\Xud \Xsp \Xut & \Xud \Xsp \Xut & \Xuc \\
\hline\hline
\Xqu \Xsp \Xcu & \Xqu \Xsp \Xcu & \Xdu \\
\Xut \Xsp \Xud & \Xut \Xsp \Xud & \Xus \\
\Xdu \Xsp \Xtu & \Xdu \Xsp \Xtu & \Xqu \\
\Xuc \Xsp \Xuq & \Xuc \Xsp \Xuq & \Xut \\
\Xtu \Xsp \Xdu & \Xtu \Xsp \Xdu & \Xdu \\
\Xud \Xsp \Xut & \Xud \Xsp \Xut & \Xuc \\
\Xsu \Xsp \Xcu & \Xqu \Xsp \Xcu & \Xtu \\
\Xut \Xsp \Xud & \Xut \Xsp \Xud & \Xuq \\
\multicolumn{2}{c}{}\\
\end{array} 
%
%
& & \begin{array}{c||c||c}
\Xud \Xsp \Xut & \Xud \Xsp \Xut & \Xus \\
\hline\hline
\Xqu \Xsp \Xcu & \Xqu \Xsp \Xcu & \Xdu \\
\Xut \Xsp \Xud & \Xut \Xsp \Xud & \Xut \\
\Xdu \Xsp \Xtu & \Xdu \Xsp \Xtu & \Xqu \\
\Xuc \Xsp \Xuq & \Xuc \Xsp \Xuq & \Xud \\
\Xtu \Xsp \Xdu & \Xtu \Xsp \Xdu & \Xtu \\
\Xud \Xsp \Xut & \Xud \Xsp \Xut & \Xuc \\
\hline\hline
\Xqu \Xsp \Xcu & \Xqu \Xsp \Xcu & \Xdu \\
\Xut \Xsp \Xud & \Xut \Xsp \Xud & \Xus \\
\Xcu \Xsp \Xqu & \Xcu \Xsp \Xqu & \Xtu \\
\Xud \Xsp \Xut & \Xud \Xsp \Xut & \Xuq \\
\Xtu \Xsp \Xdu & \Xtu \Xsp \Xdu & \Xdu \\
\Xuq \Xsp \Xuc & \Xuq \Xsp \Xuc & \Xut \\
\Xdu \Xsp \Xtu & \Xdu \Xsp \Xtu & \Xcu \\
\Xuc \Xsp \Xuq & \Xuc \Xsp \Xuq & \Xud \\
\Xtu \Xsp \Xdu & \Xtu \Xsp \Xdu & \Xsu \\
\end{array} 
\\
\ \\
\ell\equiv 0\pmod 6 &  \Xsp \Xsp \Xsp & \ell\equiv 1\pmod 6 &  \Xsp \Xsp \Xsp & \ell\equiv 2\pmod 6\\
\ell \ge 12 &  \Xsp \Xsp \Xsp & \ell \ge 13  &  \Xsp \Xsp \Xsp & \ell \ge 8 
\end{array}$

\vskip 1cm

$\begin{array}{ccccc}
%
%
\begin{array}{c||c||c}
\Xud \Xsp \Xut & \Xud \Xsp \Xut & \Xus \\
\hline\hline
\Xqu \Xsp \Xcu & \Xqu \Xsp \Xcu & \Xdu \\
\Xut \Xsp \Xud & \Xut \Xsp \Xud & \Xut \\
\Xdu \Xsp \Xtu & \Xdu \Xsp \Xtu & \Xqu \\
\Xuc \Xsp \Xuq & \Xuc \Xsp \Xuq & \Xud \\
\Xtu \Xsp \Xdu & \Xtu \Xsp \Xdu & \Xtu \\
\Xud \Xsp \Xut & \Xud \Xsp \Xut & \Xuc \\
\hline\hline
\Xqu \Xsp \Xcu & \Xqu \Xsp \Xcu & \Xsu \\
\Xut \Xsp \Xud & \Xut \Xsp \Xud & \Xud \\
\Xcu \Xsp \Xqu & \Xcu \Xsp \Xqu & \Xqu \\
\Xud \Xsp \Xut & \Xud \Xsp \Xut & \Xut \\
\Xtu \Xsp \Xdu & \Xtu \Xsp \Xdu & \Xdu \\
\Xuq \Xsp \Xuc & \Xuq \Xsp \Xuc & \Xuc \\
\Xdu \Xsp \Xtu & \Xdu \Xsp \Xtu & \Xtu \\
\Xut \Xsp \Xud & \Xut \Xsp \Xud & \Xuq \\
\Xcu \Xsp \Xqu & \Xcu \Xsp \Xqu & \Xdu \\
\Xud \Xsp \Xut & \Xud \Xsp \Xut & \Xus \\
\multicolumn{2}{c}{}\\
\multicolumn{2}{c}{}
\end{array}
%
%
& & \begin{array}{c||c||c}
\Xud \Xsp \Xut & \Xud \Xsp \Xut & \Xus \\
\hline\hline
\Xqu \Xsp \Xcu & \Xqu \Xsp \Xcu & \Xdu \\
\Xut \Xsp \Xud & \Xut \Xsp \Xud & \Xuc \\
\Xdu \Xsp \Xtu & \Xdu \Xsp \Xtu & \Xtu \\
\Xuc \Xsp \Xuq & \Xuc \Xsp \Xuq & \Xud \\
\Xtu \Xsp \Xdu & \Xtu \Xsp \Xdu & \Xqu \\
\Xud \Xsp \Xut & \Xud \Xsp \Xut & \Xut \\
\hline\hline
\Xqu \Xsp \Xcu & \Xqu \Xsp \Xcu & \Xdu \\
\Xut \Xsp \Xud & \Xut \Xsp \Xud & \Xuc \\
\Xdu \Xsp \Xtu & \Xdu \Xsp \Xtu & \Xsu \\
\Xuc \Xsp \Xuq & \Xuc \Xsp \Xuq & \Xud \\
\Xtu \Xsp \Xdu & \Xtu \Xsp \Xdu & \Xtu \\
\Xud \Xsp \Xut & \Xud \Xsp \Xut & \Xuq \\
\Xqu \Xsp \Xcu & \Xqu \Xsp \Xcu & \Xdu \\
\Xut \Xsp \Xud & \Xut \Xsp \Xud & \Xuc \\
\Xdu \Xsp \Xtu & \Xdu \Xsp \Xtu & \Xtu \\
\Xuc \Xsp \Xuq & \Xuc \Xsp \Xuq & \Xud \\
\Xtu \Xsp \Xdu & \Xtu \Xsp \Xdu & \Xsu \\
\multicolumn{2}{c}{}
\end{array} 
%
%
& & \begin{array}{c||c||c}
\Xut \Xsp \Xud & \Xut \Xsp \Xud & \Xus \\
\hline\hline
\Xqu \Xsp \Xcu & \Xqu \Xsp \Xcu & \Xtu \\
\Xud \Xsp \Xut & \Xud \Xsp \Xut & \Xud \\
\Xtu \Xsp \Xdu & \Xtu \Xsp \Xdu & \Xqu \\
\Xuc \Xsp \Xuq & \Xuc \Xsp \Xuq & \Xut \\
\Xdu \Xsp \Xtu & \Xdu \Xsp \Xtu & \Xdu \\
\Xut \Xsp \Xud & \Xut \Xsp \Xud & \Xuc \\
\hline\hline
\Xqu \Xsp \Xcu & \Xqu \Xsp \Xcu & \Xtu \\
\Xud \Xsp \Xut & \Xud \Xsp \Xut & \Xud \\
\Xcu \Xsp \Xqu & \Xcu \Xsp \Xqu & \Xsu \\
\Xut \Xsp \Xud & \Xut \Xsp \Xud & \Xuq \\
\Xdu \Xsp \Xtu & \Xdu \Xsp \Xtu & \Xdu \\
\Xuq \Xsp \Xuc & \Xuq \Xsp \Xuc & \Xut \\
\Xtu \Xsp \Xdu & \Xtu \Xsp \Xdu & \Xcu \\
\Xuc \Xsp \Xuq & \Xuc \Xsp \Xuq & \Xud \\
\Xdu \Xsp \Xtu & \Xdu \Xsp \Xtu & \Xqu \\
\Xut \Xsp \Xud & \Xut \Xsp \Xud & \Xut \\
\Xqu \Xsp \Xcu & \Xqu \Xsp \Xcu & \Xdu \\
\Xud \Xsp \Xut & \Xud \Xsp \Xut & \Xus \\
\end{array} 
\\
\ \\
\ell\equiv 3\pmod 6 &  \Xsp \Xsp \Xsp & \ell\equiv 4\pmod 6 &  \Xsp \Xsp \Xsp & \ell\equiv 5\pmod 6\\
\ell \ge 15 &  \Xsp \Xsp \Xsp & \ell \ge 16 &  \Xsp \Xsp \Xsp & \ell \ge 17
\end{array}$
\caption{Colouring patterns for  $H^\ell(r)$, $\ell=8\ \mbox{or}\ \ell\ge 12$,  $r\ge 3$, $r$ odd.}
\label{fig:patternsHodd-mod6-bis}
\end{center}
\end{figure}


In order to produce a packing 6-colouring of $H^\ell(r)$, with $\ell\ge 8$, $r\ge 3$, and $r$ odd,
we use the colouring patterns depicted in Figures \ref{fig:patternsHodd-mod6} and~\ref{fig:patternsHodd-mod6-bis}.
In both these figures, the four columns surrounded by double lines must be repeated 
$\frac{r-3}{2}$ times (and thus do not appear if $r=3$)
or $\frac{r-5}{2}$ times when $\ell=9$ and $r\ge 5$ (and thus do not appear if $r=5$).
In Figure~\ref{fig:patternsHodd-mod6-bis}, the six rows surrounded by double lines must be repeated 
$\frac{\ell-6-(\ell \mod 6)}{6}$ times 
(and thus do not appear if $\ell=8$).

\end{enumerate}

This completes the proof.
\end{proof}

\begin{figure}
\begin{center}
%
%
$\begin{array}{c}
\Xut \Xsp \Xus\\
\Xdu \Xsp \Xdu\\
\Xuq \Xsp \Xuse\\
\Xtu \Xsp \Xcu\\
\\
r=2
\end{array}$
\hskip 1cm
%
%
%
$\begin{array}{c}
\Xut \Xsp \Xus \Xsp \Xud \Xsp \Xuc\\
\Xdu \Xsp \Xdu \Xsp \Xtu \Xsp \Xtu\\
\Xuq \Xsp \Xut \Xsp \Xuq \Xsp \Xud\\
\Xtu \Xsp \Xcu \Xsp \Xdu \Xsp \Xseu\\
\\
r=4
\end{array}$
\hskip 1cm
%
%
%
$\begin{array}{c}
\Xut \Xsp \Xus \Xsp \Xuq \Xsp \Xud \Xsp \Xuse \Xsp \Xuq \Xsp \Xuc \\
\Xdu \Xsp \Xdu \Xsp \Xdu \Xsp \Xcu \Xsp \Xtu \Xsp \Xdu \Xsp \Xtu \\
\Xuq \Xsp \Xut \Xsp \Xut \Xsp \Xut \Xsp \Xud \Xsp \Xut \Xsp \Xud \\
\Xtu \Xsp \Xcu \Xsp \Xseu \Xsp \Xdu \Xsp \Xqu \Xsp \Xcu \Xsp \Xsu \\
\\
r=7
\end{array}$

\vskip 1cm

%
%
$\begin{array}{c}
\Xut \Xsp \Xus \Xsp \Xut \Xsp \Xud \Xsp \Xuse \Xsp \Xut \Xsp \Xud \Xsp \Xuc \\
\Xdu \Xsp \Xdu \Xsp \Xdu \Xsp \Xcq \Xsp \Xtu \Xsp \Xdu \Xsp \Xqu \Xsp \Xtu \\
\Xuq \Xsp \Xut \Xsp \Xuq \Xsp \Xdu \Xsp \Xuq \Xsp \Xuc \Xsp \Xut \Xsp \Xud \\
\Xtu \Xsp \Xcu \Xsp \Xseu \Xsp \Xts \Xsp \Xdu \Xsp \Xtu \Xsp \Xdu \Xsp \Xsu \\
\\
r=8
\end{array}$
\hskip 1cm
%
%
%
$\begin{array}{c}
\Xut \Xsp \Xus \Xsp \Xut \Xsp \Xud \Xsp \Xuq \Xsp \Xut \Xsp \Xus \Xsp \Xuq \Xsp \Xud \Xsp \Xuse \Xsp \Xuc \\
\Xdu \Xsp \Xdu \Xsp \Xdu \Xsp \Xcu \Xsp \Xtu \Xsp \Xdu \Xsp \Xdu \Xsp \Xdu \Xsp \Xtu \Xsp \Xtu \Xsp \Xtu \\
\Xuq \Xsp \Xut \Xsp \Xuq \Xsp \Xut \Xsp \Xud \Xsp \Xuc \Xsp \Xut \Xsp \Xut \Xsp \Xuc \Xsp \Xud \Xsp \Xud \\
\Xtu \Xsp \Xcu \Xsp \Xseu \Xsp \Xdu \Xsp \Xsu \Xsp \Xtu \Xsp \Xqu \Xsp \Xseu \Xsp \Xdu \Xsp \Xqu \Xsp \Xsu \\
\\
r=11
\end{array}$
\caption{Packing 7-colourings of $H^2(r)$, $r\in\{2,4,7,8,11\}$.}
\label{fig:H2-247811}
\end{center}
\end{figure}

\begin{figure}
\begin{center}
$\begin{array}{ccccc}
\begin{array}{c}
\Xut \Xsp \Xud \Xsp \Xut \Xsp \Xud \Xsp \Xus \\
\Xdu \Xsp \Xcu \Xsp \Xqu \Xsp \Xcu \Xsp \Xtu \\
\Xuq \Xsp \Xut \Xsp \Xud \Xsp \Xuq \Xsp \Xuc \\
\Xtu \Xsp \Xdu \Xsp \Xsu \Xsp \Xtu \Xsp \Xdu \\
\end{array}
& & \begin{array}{||c||}
\Xut \Xsp \Xud \Xsp \Xuc \\
\Xdu \Xsp \Xqu \Xsp \Xtu \\
\Xus \Xsp \Xut \Xsp \Xuq \\
\Xtu \Xsp \Xcu \Xsp \Xdu \\
\end{array}
& & \begin{array}{||c||c}
\Xut \Xsp \Xud \Xsp \Xuc & \Xut \Xsp \Xuq \Xsp \Xud \Xsp \Xut \Xsp \Xuq \Xsp \Xud \Xsp \Xuc \\
\Xdu \Xsp \Xqu \Xsp \Xtu & \Xsu \Xsp \Xdu \Xsp \Xtu \Xsp \Xcu \Xsp \Xdu \Xsp \Xsu \Xsp \Xtu \\
\Xus \Xsp \Xut \Xsp \Xuq & \Xud \Xsp \Xut \Xsp \Xus \Xsp \Xud \Xsp \Xuc \Xsp \Xut \Xsp \Xud \\
\Xtu \Xsp \Xcu \Xsp \Xdu & \Xtu \Xsp \Xcu \Xsp \Xdu \Xsp \Xqu \Xsp \Xtu \Xsp \Xdu \Xsp \Xqu \\
\end{array}
\\ \\
r=5 & & r\equiv 0\pmod 3 & & r\equiv 1\pmod 3,\ r\ge 10
\end{array}$
\vskip 0.5cm
$\begin{array}{c}
\begin{array}{c}
\Xut \Xsp \Xud \Xsp \Xuq \Xsp \Xut \Xsp \Xus \Xsp \Xud \Xsp \Xut \Xsp \Xud \Xsp \Xut \Xsp \Xus \Xsp \Xud \Xsp \Xut \Xsp \Xud \Xsp \Xus \\
\Xdu \Xsp \Xcu \Xsp \Xtu \Xsp \Xdu \Xsp \Xdu \Xsp \Xtu \Xsp \Xqu \Xsp \Xcu \Xsp \Xqu \Xsp \Xdu \Xsp \Xtu \Xsp \Xqu \Xsp \Xcu \Xsp \Xtu \\
\Xuq \Xsp \Xut \Xsp \Xud \Xsp \Xuc \Xsp \Xut \Xsp \Xuc \Xsp \Xud \Xsp \Xut \Xsp \Xuc \Xsp \Xuq \Xsp \Xuc \Xsp \Xud \Xsp \Xdt \Xsp \Xuc \\
\Xtu \Xsp \Xdu \Xsp \Xsu \Xsp \Xtu \Xsp \Xqu \Xsp \Xdu \Xsp \Xtu \Xsp \Xsu \Xsp \Xdu \Xsp \Xtu \Xsp \Xdu \Xsp \Xts \Xsp \Xuq \Xsp \Xdu \\
\end{array}
\\ \\
r=14
\end{array}$
\vskip 0.5cm
$\begin{array}{c}
\begin{array}{c}
\Xut \Xsp \Xus \Xsp \Xud \Xsp \Xut \Xsp \Xuq \Xsp \Xud \Xsp \Xus \Xsp \Xut \Xsp \Xud \Xsp \Xuc \Xsp \Xut \Xsp \Xus \Xsp \Xud \Xsp \Xut \Xsp \Xuq \Xsp \Xud \Xsp \Xus \Xsp \Xut \Xsp \Xud \Xsp \Xuc \\
\Xdu \Xsp \Xdu \Xsp \Xtu \Xsp \Xcu \Xsp \Xdu \Xsp \Xtu \Xsp \Xtu \Xsp \Xdu \Xsp \Xqu \Xsp \Xtu \Xsp \Xdu \Xsp \Xdu \Xsp \Xtu \Xsp \Xcu \Xsp \Xdu \Xsp \Xtu \Xsp \Xtu \Xsp \Xdu \Xsp \Xqu \Xsp \Xtu \\
\Xuq \Xsp \Xut \Xsp \Xuq \Xsp \Xud \Xsp \Xut \Xsp \Xuc \Xsp \Xud \Xsp \Xuc \Xsp \Xut \Xsp \Xud \Xsp \Xuq \Xsp \Xut \Xsp \Xuq \Xsp \Xud \Xsp \Xut \Xsp \Xuc \Xsp \Xud \Xsp \Xuc \Xsp \Xut \Xsp \Xud \\
\Xtu \Xsp \Xcu \Xsp \Xdu \Xsp \Xtu \Xsp \Xsu \Xsp \Xdu \Xsp \Xqu \Xsp \Xtu \Xsp \Xdu \Xsp \Xsu \Xsp \Xtu \Xsp \Xcu \Xsp \Xdu \Xsp \Xtu \Xsp \Xsu \Xsp \Xdu \Xsp \Xqu \Xsp \Xtu \Xsp \Xdu \Xsp \Xsu \\
\end{array}
\\ \\
r=20
\end{array}$
\vskip 0.5cm
$\begin{array}{c}
\begin{array}{c||c||}
\Xut \Xsp \Xud \Xsp \Xuq \Xsp \Xut \Xsp \Xud \Xsp \Xqu \Xsp \Xst \Xsp \Xud \Xsp \Xuq \Xsp \Xut \Xsp \Xus \Xsp \Xud \Xsp \Xut \Xsp \Xud \Xsp \Xut \Xsp \Xus \Xsp \Xud \Xsp \Xut \Xsp \Xud \Xsp \Xuq & \Xut \Xsp \Xud \Xsp \Xuq \Xsp \Xut \Xsp \Xud \Xsp \Xuq \\
\Xdu \Xsp \Xcu \Xsp \Xtu \Xsp \Xdu \Xsp \Xcu \Xsp \Xtd \Xsp \Xdu \Xsp \Xcu \Xsp \Xtu \Xsp \Xdu \Xsp \Xdu \Xsp \Xtu \Xsp \Xqu \Xsp \Xcu \Xsp \Xqu \Xsp \Xdu \Xsp \Xtu \Xsp \Xqu \Xsp \Xcu \Xsp \Xtu & \Xdu \Xsp \Xcu \Xsp \Xtu \Xsp \Xsu \Xsp \Xcu \Xsp \Xtu \\
\Xus \Xsp \Xut \Xsp \Xuc \Xsp \Xus \Xsp \Xut \Xsp \Xuc \Xsp \Xuq \Xsp \Xut \Xsp \Xud \Xsp \Xuc \Xsp \Xut \Xsp \Xuc \Xsp \Xud \Xsp \Xut \Xsp \Xuc \Xsp \Xuq \Xsp \Xuc \Xsp \Xud \Xsp \Xdt \Xsp \Xuc & \Xus \Xsp \Xut \Xsp \Xuc \Xsp \Xud \Xsp \Xut \Xsp \Xuc \\
\Xtu \Xsp \Xqu \Xsp \Xdu \Xsp \Xtu \Xsp \Xqu \Xsp \Xdu \Xsp \Xtu \Xsp \Xdu \Xsp \Xsu \Xsp \Xtu \Xsp \Xqu \Xsp \Xdu \Xsp \Xtu \Xsp \Xsu \Xsp \Xdu \Xsp \Xtu \Xsp \Xdu \Xsp \Xts \Xsp \Xuq \Xsp \Xdu & \Xtu \Xsp \Xqu \Xsp \Xdu \Xsp \Xtu \Xsp \Xqu \Xsp \Xdu \\
\end{array}
\\ \\
r\equiv 2\pmod 6,\ r\ge 26
\end{array}$
\vskip 0.5cm
$\begin{array}{c}
\begin{array}{c}
\Xut \Xsp \Xud \Xsp \Xuq \Xsp \Xut \Xsp \Xus \Xsp \Xud \Xsp \Xut \Xsp \Xud \Xsp \Xuq \Xsp \Xut \Xsp \Xud \Xsp \Xqu \Xsp \Xst \Xsp \Xud \Xsp \Xut \Xsp \Xud \Xsp \Xus \\
\Xdu \Xsp \Xcu \Xsp \Xtu \Xsp \Xdu \Xsp \Xdu \Xsp \Xtu \Xsp \Xqu \Xsp \Xcu \Xsp \Xtu \Xsp \Xdu \Xsp \Xcu \Xsp \Xtd \Xsp \Xdu \Xsp \Xcu \Xsp \Xqu \Xsp \Xcu \Xsp \Xtu \\
\Xuq \Xsp \Xut \Xsp \Xud \Xsp \Xuc \Xsp \Xut \Xsp \Xuc \Xsp \Xud \Xsp \Xdt \Xsp \Xuc \Xsp \Xus \Xsp \Xut \Xsp \Xuc \Xsp \Xuq \Xsp \Xut \Xsp \Xud \Xsp \Xuq \Xsp \Xuc \\
\Xtu \Xsp \Xdu \Xsp \Xsu \Xsp \Xtu \Xsp \Xqu \Xsp \Xdu \Xsp \Xts \Xsp \Xuq \Xsp \Xdu \Xsp \Xtu \Xsp \Xqu \Xsp \Xdu \Xsp \Xtu \Xsp \Xdu \Xsp \Xsu \Xsp \Xtu \Xsp \Xdu \\
\end{array}
\\ \\
r=17
\end{array}$
\vskip 0.5cm
$\begin{array}{c}
\begin{array}{c||c||}
\Xut \Xsp \Xud \Xsp \Xqu \Xsp \Xst \Xsp \Xud \Xsp \Xuq \Xsp \Xut \Xsp \Xus \Xsp \Xud \Xsp \Xut \Xsp \Xud \Xsp \Xut \Xsp \Xus \Xsp \Xud \Xsp \Xut \Xsp \Xud \Xsp \Xuq & \Xut \Xsp \Xud \Xsp \Xuq \Xsp \Xut \Xsp \Xud \Xsp \Xuq \\
\Xdu \Xsp \Xcu \Xsp \Xtd \Xsp \Xdu \Xsp \Xcu \Xsp \Xtu \Xsp \Xdu \Xsp \Xdu \Xsp \Xtu \Xsp \Xqu \Xsp \Xcu \Xsp \Xqu \Xsp \Xdu \Xsp \Xtu \Xsp \Xqu \Xsp \Xcu \Xsp \Xtu & \Xdu \Xsp \Xcu \Xsp \Xtu \Xsp \Xsu \Xsp \Xcu \Xsp \Xtu \\
\Xus \Xsp \Xut \Xsp \Xuc \Xsp \Xuq \Xsp \Xut \Xsp \Xud \Xsp \Xuc \Xsp \Xut \Xsp \Xuc \Xsp \Xud \Xsp \Xut \Xsp \Xuc \Xsp \Xuq \Xsp \Xuc \Xsp \Xud \Xsp \Xdt \Xsp \Xuc & \Xus \Xsp \Xut \Xsp \Xuc \Xsp \Xud \Xsp \Xut \Xsp \Xuc \\
\Xtu \Xsp \Xqu \Xsp \Xdu \Xsp \Xtu \Xsp \Xdu \Xsp \Xsu \Xsp \Xtu \Xsp \Xqu \Xsp \Xdu \Xsp \Xtu \Xsp \Xsu \Xsp \Xdu \Xsp \Xtu \Xsp \Xdu \Xsp \Xts \Xsp \Xuq \Xsp \Xdu & \Xtu \Xsp \Xqu \Xsp \Xdu \Xsp \Xtu \Xsp \Xqu \Xsp \Xdu \\
\end{array}
\\ \\
r\equiv 5\pmod 6,\ r\ge 23
\end{array}$
\end{center}
\caption{Colouring patterns for $H^2(r)$,  $r\notin\{2,4,7,8,11\}$.}
\label{fig:H2-general}
\end{figure}

The last two theorems of this section deal with the cases not covered
by Theorem~\ref{th:Hell(r)}, that is, $\ell=2$ and $\ell=5$,
respectively.

\begin{theorem}
For every integer $r\ge 2$,
$$\pcn(H^2(r))=\left\{
   \begin{array}{ll}
      7 & \hbox{if $r\in\{2,4,7,8,11\}$,} \\
      6 & \hbox{otherwise.}
   \end{array}
\right.$$
\label{th:H2(r)}
\end{theorem}

\begin{proof}
The fact that $H^2(r)$ does not admit a packing 6-colouring for every $r\in\{2,4,7,8,11\}$
has been checked by a computer program, using brute-force search.
Packing 7-colourings for each of these graphs are depicted in Figure~\ref{fig:H2-247811}.

Assume now $r\notin\{2,4,7,8,11\}$.
We checked by a computer program, again using brute-force search,
that the subgraph of such a generalised H-graph induced by three successive ladders,
that is, by the set of vertices
$\{u_i^j\ |\ 0\le i\le 5,\ 0\le j\le 3\}$, does not admit a packing 5-colouring.
Packing 6-colourings of such generalised H-graphs 
are depicted in Figure~\ref{fig:H2-general},
according to the value of $r$, $r$ modulo 3, or $r$ modulo 6
(periodic patterns, made of 6 or 12 columns, are surrounded by double lines).
\end{proof}

\begin{figure}
\begin{center}
%
%
$\begin{array}{c}
\Xut \Xsp \Xud \\
\Xqu \Xsp \Xsu \\
\Xuc \Xsp \Xut \\
\Xtu \Xsp \Xdu \\
\Xud \Xsp \Xuq \\
\Xsu \Xsp \Xcu \\
\Xut \Xsp \Xud \\
\\
r=2
\end{array}$
\hskip 1cm
%
%
%
$\begin{array}{c}
\Xut \Xsp \Xus \Xsp \Xuc \\
\Xdu \Xsp \Xdu \Xsp \Xdu \\
\Xuq \Xsp \Xut \Xsp \Xuq \\
\Xtu \Xsp \Xcu \Xsp \Xtu \\
\Xuc \Xsp \Xud \Xsp \Xud \\
\Xdu \Xsp \Xtu \Xsp \Xcu \\
\Xus \Xsp \Xuq \Xsp \Xut \\
\\
r=3
\end{array}$
\hskip 1cm
%
%
%
$\begin{array}{c}
\Xut \Xsp \Xud \Xsp \Xqu \Xsp \Xdc \Xsp \Xuq \\
\Xdu \Xsp \Xcu \Xsp \Xus \Xsp \Xtu \Xsp \Xdu \\
\Xus \Xsp \Xut \Xsp \Xdt \Xsp \Xud \Xsp \Xut \\
\Xtu \Xsp \Xdu \Xsp \Xcu \Xsp \Xqu \Xsp \Xsu \\
\Xuq \Xsp \Xus \Xsp \Xud \Xsp \Xut \Xsp \Xud \\
\Xdu \Xsp \Xtu \Xsp \Xqu \Xsp \Xcu \Xsp \Xqu \\
\Xuc \Xsp \Xud \Xsp \Xut \Xsp \Xud \Xsp \Xut \\
\\
r=5
\end{array}$
\caption{Packing 6-colourings of $H^5(r)$, $r\in\{2,3,5\}$.}
\label{fig:H5-2345}
\end{center}
\end{figure}

\begin{figure}
\begin{center}
$\begin{array}{ccc}
\begin{array}{||c||}
\Xut \Xsp \Xus \Xsp \Xut \Xsp \Xuc \\
\Xdu \Xsp \Xdu \Xsp \Xqu \Xsp \Xdu \\
\Xuq \Xsp \Xut \Xsp \Xud \Xsp \Xuq \\
\Xtu \Xsp \Xqu \Xsp \Xcu \Xsp \Xtu \\
\Xuc \Xsp \Xuc \Xsp \Xut \Xsp \Xud \\
\Xdu \Xsp \Xdu \Xsp \Xdu \Xsp \Xcu \\
\Xus \Xsp \Xut \Xsp \Xuq \Xsp \Xut \\
\end{array}
 & & \begin{array}{||c||c}
\Xut \Xsp \Xus \Xsp \Xut \Xsp \Xuc & \Xut \Xsp \Xus \Xsp \Xut \Xsp \Xuq \Xsp \Xut \Xsp \Xuc \\
\Xdu \Xsp \Xdu \Xsp \Xqu \Xsp \Xdu & \Xdu \Xsp \Xdu \Xsp \Xcu \Xsp \Xdu \Xsp \Xdu \Xsp \Xdu \\
\Xuq \Xsp \Xut \Xsp \Xud \Xsp \Xuq & \Xuq \Xsp \Xut \Xsp \Xud \Xsp \Xuc \Xsp \Xus \Xsp \Xuq \\
\Xtu \Xsp \Xqu \Xsp \Xcu \Xsp \Xtu & \Xtu \Xsp \Xcu \Xsp \Xtu \Xsp \Xtu \Xsp \Xcu \Xsp \Xtu \\
\Xuc \Xsp \Xuc \Xsp \Xut \Xsp \Xud & \Xuc \Xsp \Xuq \Xsp \Xuq \Xsp \Xqd \Xsp \Xut \Xsp \Xud \\
\Xdu \Xsp \Xdu \Xsp \Xdu \Xsp \Xcu & \Xdu \Xsp \Xdu \Xsp \Xdu \Xsp \Xus \Xsp \Xdu \Xsp \Xcu \\
\Xus \Xsp \Xut \Xsp \Xuq \Xsp \Xut & \Xus \Xsp \Xut \Xsp \Xuc \Xsp \Xdt \Xsp \Xuq \Xsp \Xut \\
\end{array}
\\
\\
r\equiv 0\pmod 4,\ r\ge 4 &  \ \ \ & r\equiv 2\pmod 4,\ r\ge 6
\end{array}$
\vskip 0.5cm
$\begin{array}{ccc}
 \begin{array}{||c||c}
\Xdu \Xsp \Xcu \Xsp \Xdu \Xsp \Xtu \Xsp \Xcu \Xsp \Xtu & \Xdu  \Xsp \Xcu \Xsp \Xdu \Xsp \Xsu \Xsp \Xdu \Xsp \Xcd \Xsp \Xus \\
\Xuq \Xsp \Xut \Xsp \Xus \Xsp \Xuq \Xsp \Xud \Xsp \Xus & \Xut  \Xsp \Xut \Xsp \Xuq \Xsp \Xut \Xsp \Xut \Xsp \Xut \Xsp \Xqt \\
\Xtu \Xsp \Xdu \Xsp \Xtu \Xsp \Xdu \Xsp \Xtu \Xsp \Xdu & \Xqu  \Xsp \Xqu \Xsp \Xtu \Xsp \Xdu \Xsp \Xqu \Xsp \Xqu \Xsp \Xdu \\
\Xud \Xsp \Xuq \Xsp \Xuc \Xsp \Xut \Xsp \Xuq \Xsp \Xuc & \Xud  \Xsp \Xud \Xsp \Xuc \Xsp \Xuq \Xsp \Xud \Xsp \Xud \Xsp \Xuc \\
\Xcu \Xsp \Xtu \Xsp \Xdu \Xsp \Xcu \Xsp \Xdu \Xsp \Xtu & \Xcu  \Xsp \Xtu \Xsp \Xdu \Xsp \Xcu \Xsp \Xcu \Xsp \Xtu \Xsp \Xtu \\
\Xus \Xsp \Xuc \Xsp \Xuq \Xsp \Xus \Xsp \Xuc \Xsp \Xuq & \Xus  \Xsp \Xuc \Xsp \Xut \Xsp \Xut \Xsp \Xut \Xsp \Xuc \Xsp \Xud \\
\Xtu \Xsp \Xdu \Xsp \Xtu \Xsp \Xdu \Xsp \Xtu \Xsp \Xdu & \Xtu  \Xsp \Xdu \Xsp \Xqu \Xsp \Xdu \Xsp \Xsu \Xsp \Xdu \Xsp \Xqu \\
\end{array}
 & & \begin{array}{||c||c}
\Xdu \Xsp \Xcu \Xsp \Xdu \Xsp \Xtu \Xsp \Xcu \Xsp \Xtu & \Xdu \Xsp \Xcd \Xsp \Xus \\
\Xuq \Xsp \Xut \Xsp \Xus \Xsp \Xuq \Xsp \Xud \Xsp \Xus & \Xut \Xsp \Xut \Xsp \Xqt \\
\Xtu \Xsp \Xdu \Xsp \Xtu \Xsp \Xdu \Xsp \Xtu \Xsp \Xdu & \Xqu \Xsp \Xqu \Xsp \Xdu \\
\Xud \Xsp \Xuq \Xsp \Xuc \Xsp \Xut \Xsp \Xuq \Xsp \Xuc & \Xud \Xsp \Xud \Xsp \Xuc \\
\Xcu \Xsp \Xtu \Xsp \Xdu \Xsp \Xcu \Xsp \Xdu \Xsp \Xtu & \Xcu \Xsp \Xtu \Xsp \Xtu \\
\Xus \Xsp \Xuc \Xsp \Xuq \Xsp \Xus \Xsp \Xuc \Xsp \Xuq & \Xus \Xsp \Xuc \Xsp \Xud \\
\Xtu \Xsp \Xdu \Xsp \Xtu \Xsp \Xdu \Xsp \Xtu \Xsp \Xdu & \Xtu \Xsp \Xdu \Xsp \Xqu \\
\end{array}
\\
\\
r\equiv 1\pmod 6,\ r\ge 7 &  \ \ \ & r\equiv 3\pmod 6,\ r\ge 9
\end{array}$
\vskip 0.5cm
$\begin{array}{c}
 \begin{array}{||c||c}
\Xdu \Xsp \Xcu \Xsp \Xdu \Xsp \Xtu \Xsp \Xcu \Xsp \Xtu & \Xdu \Xsp \Xcu \Xsp \Xdu \Xsp \Xtu \Xsp \Xcu  \Xsp \Xtu \Xsp \Xqu \Xsp \Xtu \Xsp \Xdu \Xsp \Xqu \Xsp \Xtu \\
\Xuq \Xsp \Xut \Xsp \Xus \Xsp \Xuq \Xsp \Xud \Xsp \Xus & \Xuq \Xsp \Xut \Xsp \Xus \Xsp \Xuq \Xsp \Xud  \Xsp \Xud \Xsp \Xud \Xsp \Xus \Xsp \Xuc \Xsp \Xud \Xsp \Xus \\
\Xtu \Xsp \Xdu \Xsp \Xtu \Xsp \Xdu \Xsp \Xtu \Xsp \Xdu & \Xtu \Xsp \Xdu \Xsp \Xtu \Xsp \Xdu \Xsp \Xtu  \Xsp \Xsu \Xsp \Xtu \Xsp \Xcu \Xsp \Xtu \Xsp \Xtu \Xsp \Xdu \\
\Xud \Xsp \Xuq \Xsp \Xuc \Xsp \Xut \Xsp \Xuq \Xsp \Xuc & \Xud \Xsp \Xuq \Xsp \Xuc \Xsp \Xut \Xsp \Xuq  \Xsp \Xuc \Xsp \Xuc \Xsp \Xud \Xsp \Xud \Xsp \Xuc \Xsp \Xuc \\
\Xcu \Xsp \Xtu \Xsp \Xdu \Xsp \Xcu \Xsp \Xdu \Xsp \Xtu & \Xcu \Xsp \Xtu \Xsp \Xdu \Xsp \Xcu \Xsp \Xdu  \Xsp \Xtu \Xsp \Xdu \Xsp \Xqu \Xsp \Xqu \Xsp \Xsu \Xsp \Xtu \\
\Xus \Xsp \Xuc \Xsp \Xuq \Xsp \Xus \Xsp \Xuc \Xsp \Xuq & \Xus \Xsp \Xuc \Xsp \Xuq \Xsp \Xus \Xsp \Xuc  \Xsp \Xuq \Xsp \Xus \Xsp \Xut \Xsp \Xut \Xsp \Xut \Xsp \Xud \\
\Xtu \Xsp \Xdu \Xsp \Xtu \Xsp \Xdu \Xsp \Xtu \Xsp \Xdu & \Xtu \Xsp \Xdu \Xsp \Xtu \Xsp \Xdu \Xsp \Xtu  \Xsp \Xdu \Xsp \Xtu \Xsp \Xdu \Xsp \Xcu \Xsp \Xdu \Xsp \Xqu \\
\end{array}
\\
\\
 r\equiv 5\pmod 6,\ r\ge 11
\end{array}$
\end{center}
\caption{Colouring patterns for $H^5(r)$,  $r=4$ or $r\ge 6$.}
\label{fig:H5-general}
\end{figure}

\begin{theorem}
For every integer $r\ge 2$, $\pcn(H^5(r))=6$.
\label{th:H5(r)}
\end{theorem}

\begin{proof}
Again, we checked by a computer program, using brute-force search,
that both $H^5(2)$ and the subgraph of $H^5(r)$, $r\ge 5$, induced by three successive ladders,
that is, by the set of vertices
$\{u_i^j\ |\ 0\le i\le 5,\ 0\le j\le 6\}$, do not admit a packing 5-colouring.
Packing 6-colourings of $H^5(r)$, $r\in\{2,3,5\}$, are depicted in Figure~\ref{fig:H5-2345},
while packing 6-colourings of $H^5(r)$, $r=4$ or $r\ge 6$, are depicted in Figure~\ref{fig:H5-general}
according to the value of $r$ modulo 4, or $r$ modulo 6 
(periodic patterns, made of eight or twelve columns, are surrounded by double lines
and are repeated at least once when $r\equiv 0\pmod 4$ or $r\equiv 3\pmod 6$).
\end{proof}

\section{Discussion}

In this paper, we have studied the packing chromatic number of some classes of cubic graphs,
namely circular ladders, H-graphs and generalised H-graphs.
We have determined the exact value of this parameter for every such graph, except for
the case of H-graphs $H(r)$ with $r\ge 3$, $r$~odd (see Theorem~\ref{th:H(r)}),
for which we proved $6\le\pcn(H(r))\le 7$.
Using a computer program, we have checked that $\pcn(H(r))=7$ for every odd $r$ up to $r=13$.
We thus propose the following question.

\begin{question}
Is it true that $\pcn(H(r))=7$ for every H-graph $H(r)$ with $r\ge 3$, $r$~odd?
\end{question}


In~\cite{LBS16,LBS17}, we have extended the notion of packing colouring to the case
of digraphs.
If $D$ is a digraph, the (weak) directed distance between two vertices $u$ and~$v$ in~$D$ is defined
as the length of a shortest directed path between $u$ and~$v$, in either direction.
Using this notion of distance in digraphs, the packing colouring readily
extends to digraphs.
Recall that an orientation of an undirected graph~$G$ is any antisymmetric digraph
obtained from $G$ by giving to each edge of ~$G$ one of its two possible orientations.
It then directly follows from the definition that $\pcn(D)\le\pcn(G)$ for
any orientation $D$ of~$G$.
A natural question for oriented graphs, related to this work, is then the following.

\begin{question}
Is it true that the packing chromatic number of any oriented graph with maximum degree~$3$
is bounded by some constant?
\end{question}

\vskip 1cm

\noindent
{\bf Acknowledgment.} This work has been done while the first author
was visiting LaBRI, whose hospitality is gratefully acknowledged.



\appendix

\section{Proof of Lemma~\ref{lem:graphX}}
\label{ap:graphX}

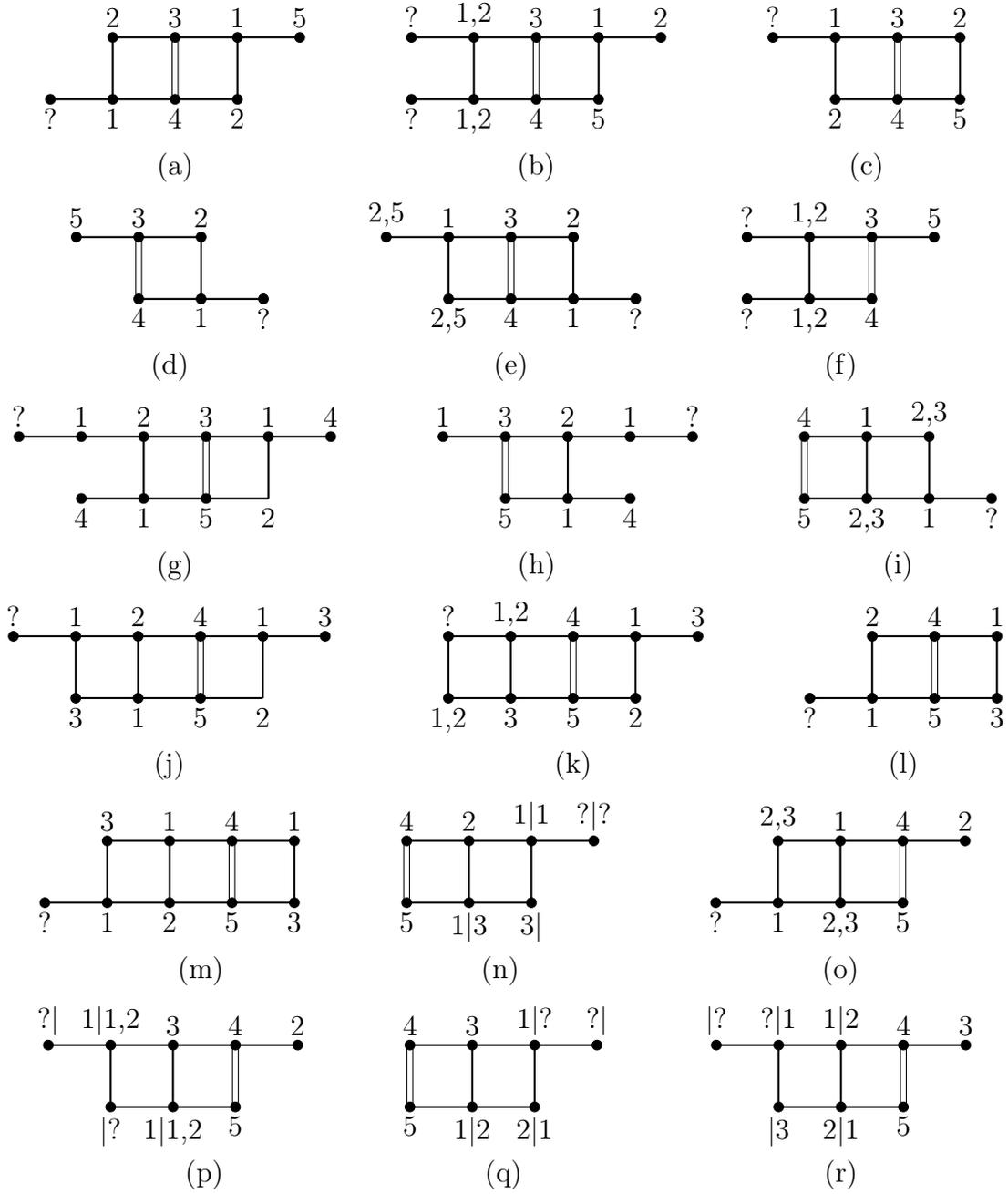
\begin{figure}
\begin{center}

\begin{tikzpicture}[domain=-0.5:4.5,x=0.9cm,y=0.9cm]
\ARETEH{0,0}{3}; \ARETEH{1,1}{3}; 
\ARETEV{1,0};
\draw[double distance = 2pt] (2,0) -- (2,1);   
\ARETEV{3,0};
\SOMMET{0,0}; \SOMMET{1,0}; \SOMMET{2,0}; \SOMMET{3,0}; 
\SOMMET{1,1}; \SOMMET{2,1}; \SOMMET{3,1}; \SOMMET{4,1}; 
\ETIQUETTEA{0,1}{};
\ETIQUETTEA{1,1}{2};
\ETIQUETTEA{2,1}{3};
\ETIQUETTEA{3,1}{1};
\ETIQUETTEA{4,1}{5};
\ETIQUETTEB{0,0}{?};
\ETIQUETTEB{1,0}{1};
\ETIQUETTEB{2,0}{4};
\ETIQUETTEB{3,0}{2};
\ETIQUETTEB{4,0}{};
\node[below] at (2,-0.7) {(a)};
\end{tikzpicture}
\hskip 1cm
\begin{tikzpicture}[domain=-0.5:4.5,x=0.9cm,y=0.9cm]
\ARETEH{0,0}{3}; \ARETEH{0,1}{4}; 
\ARETEV{1,0};
\draw[double distance = 2pt] (2,0) -- (2,1);   
\ARETEV{3,0};
\SOMMET{0,0}; \SOMMET{1,0}; \SOMMET{2,0}; \SOMMET{3,0}; 
\SOMMET{0,1}; \SOMMET{1,1}; \SOMMET{2,1}; \SOMMET{3,1}; \SOMMET{4,1}; 
\ETIQUETTEA{0,1}{?};
\ETIQUETTEA{1,1}{1,2};
\ETIQUETTEA{2,1}{3};
\ETIQUETTEA{3,1}{1};
\ETIQUETTEA{4,1}{2};
\ETIQUETTEB{0,0}{?};
\ETIQUETTEB{1,0}{1,2};
\ETIQUETTEB{2,0}{4};
\ETIQUETTEB{3,0}{5};
\ETIQUETTEB{4,0}{};
\node[below] at (2,-0.7) {(b)};
\end{tikzpicture}
\hskip 1cm
\begin{tikzpicture}[domain=-0.5:4.5,x=0.9cm,y=0.9cm]
\ARETEH{1,0}{2}; 
\ARETEH{0,1}{3}; 
\ARETEV{1,0};
\draw[double distance = 2pt] (2,0) -- (2,1);   
\ARETEV{3,0};
\SOMMET{1,0}; \SOMMET{2,0}; \SOMMET{3,0}; 
\SOMMET{0,1}; \SOMMET{1,1}; \SOMMET{2,1}; \SOMMET{3,1}; 
\ETIQUETTEA{0,1}{?};
\ETIQUETTEA{1,1}{1};
\ETIQUETTEA{2,1}{3};
\ETIQUETTEA{3,1}{2};
%
\ETIQUETTEB{0,0}{};
\ETIQUETTEB{1,0}{2};
\ETIQUETTEB{2,0}{4};
\ETIQUETTEB{3,0}{5};
\node[below] at (1.5,-0.7) {(c)};
\end{tikzpicture}

\begin{tikzpicture}[domain=-0.5:4.5,x=0.9cm,y=0.9cm]
\ARETEH{2,0}{2 }; 
\ARETEH{1,1}{2}; 
\draw[double distance = 2pt] (2,0) -- (2,1);   
\ARETEV{3,0};
\SOMMET{2,0}; \SOMMET{3,0}; \SOMMET{4,0}; 
\SOMMET{1,1}; \SOMMET{2,1}; \SOMMET{3,1}; 
\ETIQUETTEA{1,1}{5};
\ETIQUETTEA{2,1}{3};
\ETIQUETTEA{3,1}{2};
\ETIQUETTEA{4,1}{};
%
\ETIQUETTEB{1,0}{};
\ETIQUETTEB{2,0}{4};
\ETIQUETTEB{3,0}{1};
\ETIQUETTEB{4,0}{?};
\node[below] at (2.5,-0.7) {(d)};
\end{tikzpicture}
\hskip 1cm
\begin{tikzpicture}[domain=-0.5:4.5,x=0.9cm,y=0.9cm]
\ARETEH{1,0}{3}; 
\ARETEH{0,1}{3}; 
\ARETEV{1,0};
\draw[double distance = 2pt] (2,0) -- (2,1);   
\ARETEV{3,0};
\SOMMET{1,0}; 
\SOMMET{2,0}; \SOMMET{3,0}; \SOMMET{4,0}; 
\SOMMET{0,1}; 
\SOMMET{1,1}; \SOMMET{2,1}; \SOMMET{3,1}; 
\ETIQUETTEA{0,1}{2,5};
\ETIQUETTEA{1,1}{1};
\ETIQUETTEA{2,1}{3};
\ETIQUETTEA{3,1}{2};
\ETIQUETTEA{4,1}{};
\ETIQUETTEB{0,0}{};
\ETIQUETTEB{1,0}{2,5};
\ETIQUETTEB{2,0}{4};
\ETIQUETTEB{3,0}{1};
\ETIQUETTEB{4,0}{?};
\node[below] at (2,-0.7) {(e)};
\end{tikzpicture}
\hskip 1cm
\begin{tikzpicture}[domain=-0.5:4.5,x=0.9cm,y=0.9cm]
\ARETEH{0,0}{2}; \ARETEH{0,1}{3}; 
\ARETEV{1,0};
\draw[double distance = 2pt] (2,0) -- (2,1);   
\SOMMET{0,0}; \SOMMET{1,0}; \SOMMET{2,0}; 
\SOMMET{0,1}; \SOMMET{1,1}; \SOMMET{2,1}; \SOMMET{3,1}; 
\ETIQUETTEA{0,1}{?};
\ETIQUETTEA{1,1}{1,2};
\ETIQUETTEA{2,1}{3};
\ETIQUETTEA{3,1}{5};
%
\ETIQUETTEB{0,0}{?};
\ETIQUETTEB{1,0}{1,2};
\ETIQUETTEB{2,0}{4};
\ETIQUETTEB{3,0}{};
\node[below] at (1.5,-0.7) {(f)};
\end{tikzpicture}

\begin{tikzpicture}[domain=-0.5:5.5,x=0.9cm,y=0.9cm]
\ARETEH{1,0}{3}; \ARETEH{0,1}{5}; 
\ARETEV{2,0};
\draw[double distance = 2pt] (3,0) -- (3,1);   
\ARETEV{4,0};
\SOMMET{1,0}; \SOMMET{2,0}; \SOMMET{3,0}; 
\SOMMET{0,1}; \SOMMET{1,1}; \SOMMET{2,1}; \SOMMET{3,1}; \SOMMET{4,1};  \SOMMET{5,1};
\ETIQUETTEA{0,1}{?};
\ETIQUETTEA{1,1}{1};
\ETIQUETTEA{2,1}{2};
\ETIQUETTEA{3,1}{3};
\ETIQUETTEA{4,1}{1};
\ETIQUETTEA{5,1}{4};
\ETIQUETTEB{0,0}{};
\ETIQUETTEB{1,0}{4};
\ETIQUETTEB{2,0}{1};
\ETIQUETTEB{3,0}{5};
\ETIQUETTEB{4,0}{2};
\node[below] at (2.5,-0.7) {(g)};
\end{tikzpicture}
\hskip 1cm
\begin{tikzpicture}[domain=-0.5:4.5,x=0.9cm,y=0.9cm]
\ARETEH{2,0}{2}; 
\ARETEH{1,1}{4}; 
\draw[double distance = 2pt] (2,0) -- (2,1);   
\ARETEV{3,0};
\SOMMET{2,0}; \SOMMET{3,0}; \SOMMET{4,0}; 
\SOMMET{1,1}; \SOMMET{2,1}; \SOMMET{3,1}; \SOMMET{4,1}; \SOMMET{5,1}; 
\ETIQUETTEA{1,1}{1};
\ETIQUETTEA{2,1}{3};
\ETIQUETTEA{3,1}{2};
\ETIQUETTEA{4,1}{1};
\ETIQUETTEA{5,1}{?};
%
\ETIQUETTEB{1,0}{};
\ETIQUETTEB{2,0}{5};
\ETIQUETTEB{3,0}{1};
\ETIQUETTEB{4,0}{4};
\node[below] at (2.5,-0.7) {(h)};
\end{tikzpicture}
\hskip 1cm
\begin{tikzpicture}[domain=-0.5:4.5,x=0.9cm,y=0.9cm]
\ARETEH{1,0}{3}; \ARETEH{1,1}{2}; 
\draw[double distance = 2pt] (1,0) -- (1,1);   
\ARETEV{2,0};
\ARETEV{3,0};
\SOMMET{1,0}; \SOMMET{2,0}; \SOMMET{3,0}; \SOMMET{4,0}; 
\SOMMET{1,1}; \SOMMET{2,1}; \SOMMET{3,1}; 
\ETIQUETTEA{1,1}{4};
\ETIQUETTEA{2,1}{1};
\ETIQUETTEA{3,1}{2,3};
\ETIQUETTEA{4,1}{};
%
\ETIQUETTEB{1,0}{5};
\ETIQUETTEB{2,0}{2,3};
\ETIQUETTEB{3,0}{1};
\ETIQUETTEB{4,0}{?};
\node[below] at (2.5,-0.7) {(i)};
\end{tikzpicture}

\begin{tikzpicture}[domain=-0.5:5.5,x=0.9cm,y=0.9cm]
\ARETEH{1,0}{3}; \ARETEH{0,1}{5}; 
\ARETEV{1,0};
\ARETEV{2,0};
\draw[double distance = 2pt] (3,0) -- (3,1);   
\ARETEV{4,0};
\SOMMET{1,0}; \SOMMET{2,0}; \SOMMET{3,0}; 
\SOMMET{0,1}; \SOMMET{1,1}; \SOMMET{2,1}; \SOMMET{3,1}; \SOMMET{4,1};  \SOMMET{5,1};
\ETIQUETTEA{0,1}{?};
\ETIQUETTEA{1,1}{1};
\ETIQUETTEA{2,1}{2};
\ETIQUETTEA{3,1}{4};
\ETIQUETTEA{4,1}{1};
\ETIQUETTEA{5,1}{3};
\ETIQUETTEB{0,0}{};
\ETIQUETTEB{1,0}{3};
\ETIQUETTEB{2,0}{1};
\ETIQUETTEB{3,0}{5};
\ETIQUETTEB{4,0}{2};
\node[below] at (2.5,-0.7) {(j)};
\end{tikzpicture}
\hskip 1cm
\begin{tikzpicture}[domain=-0.5:4.5,x=0.9cm,y=0.9cm]
\ARETEH{0,0}{3}; \ARETEH{0,1}{4}; 
\ARETEV{0,0};
\ARETEV{1,0};
\draw[double distance = 2pt] (2,0) -- (2,1);   
\ARETEV{3,0};
\SOMMET{0,0}; \SOMMET{1,0}; \SOMMET{2,0}; \SOMMET{3,0}; 
\SOMMET{0,1}; \SOMMET{1,1}; \SOMMET{2,1}; \SOMMET{3,1}; \SOMMET{4,1}; 
\ETIQUETTEA{0,1}{?};
\ETIQUETTEA{1,1}{1,2};
\ETIQUETTEA{2,1}{4};
\ETIQUETTEA{3,1}{1};
\ETIQUETTEA{4,1}{3};
\ETIQUETTEB{0,0}{1,2};
\ETIQUETTEB{1,0}{3};
\ETIQUETTEB{2,0}{5};
\ETIQUETTEB{3,0}{2};
\ETIQUETTEB{4,0}{};
\node[below] at (2,-0.7) {(k)};
\end{tikzpicture}
\hskip 1cm
\begin{tikzpicture}[domain=-0.5:4.5,x=0.9cm,y=0.9cm]
\ARETEH{1,1}{2}; \ARETEH{0,0}{3}; 
\ARETEV{1,0};
\draw[double distance = 2pt] (2,0) -- (2,1);   
\ARETEV{3,0};
\SOMMET{0,0}; 
\SOMMET{1,0}; \SOMMET{2,0}; \SOMMET{3,0}; 
\SOMMET{1,1}; \SOMMET{2,1}; \SOMMET{3,1}; 
\ETIQUETTEA{1,1}{2};
\ETIQUETTEA{2,1}{4};
\ETIQUETTEA{3,1}{1};
%
\ETIQUETTEB{0,0}{?};
\ETIQUETTEB{1,0}{1};
\ETIQUETTEB{2,0}{5};
\ETIQUETTEB{3,0}{3};
\node[below] at (1.5,-0.7) {(l)};
\end{tikzpicture}

\begin{tikzpicture}[domain=-0.5:4.5,x=0.9cm,y=0.9cm]
\ARETEH{-1,0}{4}; \ARETEH{0,1}{3}; 
\ARETEV{0,0};
\ARETEV{1,0};
\draw[double distance = 2pt] (2,0) -- (2,1);   
\ARETEV{3,0};
\SOMMET{-1,0}; \SOMMET{0,0}; \SOMMET{1,0}; \SOMMET{2,0}; \SOMMET{3,0}; 
\SOMMET{0,1}; \SOMMET{1,1}; \SOMMET{2,1}; \SOMMET{3,1}; 
\ETIQUETTEA{0,1}{3};
\ETIQUETTEA{1,1}{1};
\ETIQUETTEA{2,1}{4};
\ETIQUETTEA{3,1}{1};
%
\ETIQUETTEB{-1,0}{?};
\ETIQUETTEB{0,0}{1};
\ETIQUETTEB{1,0}{2};
\ETIQUETTEB{2,0}{5};
\ETIQUETTEB{3,0}{3};
\node[below] at (1.5,-0.7) {(m)};
\end{tikzpicture}
\hskip 1cm
\begin{tikzpicture}[domain=-0.5:4.5,x=0.9cm,y=0.9cm]
\ARETEH{1,0}{2}; \ARETEH{1,1}{3}; 
\draw[double distance = 2pt] (1,0) -- (1,1);   
\ARETEV{2,0};
\ARETEV{3,0};
\SOMMET{1,0}; \SOMMET{2,0}; \SOMMET{3,0}; 
\SOMMET{1,1}; \SOMMET{2,1}; \SOMMET{3,1}; \SOMMET{4,1}; 
\ETIQUETTEA{1,1}{4};
\ETIQUETTEA{2,1}{2};
\ETIQUETTEA{3,1}{1|1};
\ETIQUETTEA{4,1}{?|?};
%
\ETIQUETTEB{1,0}{5};
\ETIQUETTEB{2,0}{1|3};
\ETIQUETTEB{3,0}{3|};
\ETIQUETTEB{4,0}{};
\node[below] at (2.5,-0.7) {(n)};
\end{tikzpicture}
\hskip 1cm
\begin{tikzpicture}[domain=-0.5:4.5,x=0.9cm,y=0.9cm]
\ARETEH{0,0}{3}; \ARETEH{1,1}{3}; 
\ARETEV{1,0};
\ARETEV{2,0};
\draw[double distance = 2pt] (3,0) -- (3,1);   
\SOMMET{0,0}; \SOMMET{1,0}; \SOMMET{2,0}; \SOMMET{3,0}; 
\SOMMET{1,1}; \SOMMET{2,1}; \SOMMET{3,1}; \SOMMET{4,1}; 
\ETIQUETTEA{0,1}{};
\ETIQUETTEA{1,1}{2,3};
\ETIQUETTEA{2,1}{1};
\ETIQUETTEA{3,1}{4};
\ETIQUETTEA{4,1}{2};
\ETIQUETTEB{0,0}{?};
\ETIQUETTEB{1,0}{1};
\ETIQUETTEB{2,0}{2,3};
\ETIQUETTEB{3,0}{5};
\ETIQUETTEB{4,0}{};
\node[below] at (2,-0.7) {(o)};
\end{tikzpicture}

\begin{tikzpicture}[domain=-0.5:4.5,x=0.9cm,y=0.9cm]
\ARETEH{1,0}{2}; \ARETEH{0,1}{4}; 
\ARETEV{1,0};
\ARETEV{2,0};
\draw[double distance = 2pt] (3,0) -- (3,1);   
\SOMMET{1,0}; \SOMMET{2,0}; \SOMMET{3,0}; 
\SOMMET{0,1}; 
\SOMMET{1,1}; \SOMMET{2,1}; \SOMMET{3,1}; \SOMMET{4,1}; 
\ETIQUETTEA{0,1}{?|};
\ETIQUETTEA{1,1}{1|1,2};
\ETIQUETTEA{2,1}{3};
\ETIQUETTEA{3,1}{4};
\ETIQUETTEA{4,1}{2};
%
\ETIQUETTEB{1,0}{|?};
\ETIQUETTEB{2,0}{1|1,2};
\ETIQUETTEB{3,0}{5};
\ETIQUETTEB{4,0}{};
\node[below] at (2.5,-0.7) {(p)};
\end{tikzpicture}
\hskip 1cm
\begin{tikzpicture}[domain=-0.5:4.5,x=0.9cm,y=0.9cm]
\ARETEH{1,0}{2}; \ARETEH{1,1}{3}; 
\draw[double distance = 2pt] (1,0) -- (1,1);   
\ARETEV{2,0};
\ARETEV{3,0};
\SOMMET{1,0}; \SOMMET{2,0}; \SOMMET{3,0}; 
\SOMMET{1,1}; \SOMMET{2,1}; \SOMMET{3,1}; \SOMMET{4,1}; 
\ETIQUETTEA{1,1}{4};
\ETIQUETTEA{2,1}{3};
\ETIQUETTEA{3,1}{1|?};
\ETIQUETTEA{4,1}{?|};
%
\ETIQUETTEB{1,0}{5};
\ETIQUETTEB{2,0}{1|2};
\ETIQUETTEB{3,0}{2|1};
\ETIQUETTEB{4,0}{};
\node[below] at (2.5,-0.7) {(q)};
\end{tikzpicture}
\hskip 1cm
\begin{tikzpicture}[domain=-0.5:4.5,x=0.9cm,y=0.9cm]
\ARETEH{1,0}{2}; \ARETEH{0,1}{4}; 
\ARETEV{1,0};
\ARETEV{2,0};
\draw[double distance = 2pt] (3,0) -- (3,1);   
\SOMMET{1,0}; \SOMMET{2,0}; \SOMMET{3,0}; 
\SOMMET{0,1}; 
\SOMMET{1,1}; \SOMMET{2,1}; \SOMMET{3,1}; \SOMMET{4,1}; 
\ETIQUETTEA{0,1}{|?};
\ETIQUETTEA{1,1}{?|1};
\ETIQUETTEA{2,1}{1|2};
\ETIQUETTEA{3,1}{4};
\ETIQUETTEA{4,1}{3};
%
\ETIQUETTEB{1,0}{|3};
\ETIQUETTEB{2,0}{2|1};
\ETIQUETTEB{3,0}{5};
\ETIQUETTEB{4,0}{};
\node[below] at (2,-0.7) {(r)};
\end{tikzpicture}
\caption{Configurations for the proof of Lemma~\ref{lem:graphX} (the double edge is the edge $u_iv_i$).}
\label{fig:proofX-1}
\end{center}
\end{figure}

The configurations used in the proof correspond to partial colourings of the graph $X$
and are depicted in Figures \ref{fig:proofX-1} and~\ref{fig:proofX-2}, with the following drawing convention.
If $\{a,b\}$ is the set of colours assigned to two distinct vertices, then the ``colour'' of both these
vertices is denoted ``$a,b$''.
If the same configuration describes two partial colourings of $X$ and the colours assigned to some vertex
by these two colourings are respectively $a$ and $b$, then the ``colour'' of this vertex is denoted ``$a|b$''.
Finally, if a vertex has no available colour, its ``colour'' is denoted ``$?$''. 

%
Suppose that for some $i$, $3\le i\le 5$, $\pi(u_i)\neq 1$ and $\pi(v_{i})\neq 1$.
We first prove the following claim.

\begin{claim}
$2\in\{\pi(u_i),\pi(v_i)\}$.
\label{cl:graphX}
\end{claim}

\begin{proof} 
Assume to the contrary that this is not the case, that is, $\{\pi(u_i),\pi(v_i)\}\subseteq\{3,4,5\}$.
Thanks to the symmetry exchanging $u_i$ and $v_i$, we may assume $\pi(u_i)<\pi(v_i)$, without loss of generality.
Recall that there is no edge $u_{i-2}v_{i-2}$ (resp. $u_{i+2}v_{i+2}$) in $X$ when $i=3$ (resp. $i=5$).
We consider the following cases (subscripts are taken modulo $n$).

\begin{enumerate}
\item $\pi(u_i)=3$ and $\pi(v_i)=4$.\\
In that case, we necessarily have $\pi(u_{i+1})\in\{1,2,5\}$.

If $\pi(u_{i+1})=1$, then $\{\pi(v_{i+1}),\pi(u_{i+2})\}=\{2,5\}$.
If $\pi(v_{i+1})=2$ (and $\pi(u_{i+2})=5$), then $\pi(v_{i-1})=1$, so that $\pi(u_{i-1})=2$ and no colour is available for $v_{i-2}$
(see Figure~\ref{fig:proofX-1}(a)).
If $\pi(u_{i+2})=2$ (and $\pi(v_{i+1})=5$), then $\{\pi(u_{i-1}),\pi(v_{i-1})\}=\{1,2\}$, and no colour is available
either for $u_{i-2}$ or for  $v_{i-2}$
(see Figure~\ref{fig:proofX-1}(b)).

If $\pi(u_{i+1})=2$, then $\pi(v_{i+1})\in\{1,5\}$.
If $\pi(v_{i+1})=5$, then $\pi(u_{i-1})=1$, so that $\pi(v_{i-1})=2$ and no colour is available for $u_{i-2}$
(see Figure~\ref{fig:proofX-1}(c)).
If $\pi(v_{i+1})=1$, then either $\pi(u_{i-1})=5$, so that 
no colour is available for $v_{i+2}$ (see Figure~\ref{fig:proofX-1}(d)),
or $\pi(u_{i-1})=1$,
which implies $\{\pi(u_{i-2}),\pi(v_{i-1})\}=\{2,5\}$, so that again 
no colour is available for $v_{i+2}$ (see Figure~\ref{fig:proofX-1}(e)).

Finally, if $\pi(u_{i+1})=5$, then $\{\pi(u_{i-1}),\pi(v_{i-1})\}=\{1,2\}$, and no colour is available either for
$u_{i-2}$ or for $v_{i-2}$
(see Figure~\ref{fig:proofX-1}(f)).


\item $\pi(u_i)=3$ and $\pi(v_i)=5$.\\
Observe that the proof is similar to the proof of the previous case, by switching colours 4 and~5,
in all cases illustrated in Figure~\ref{fig:proofX-1}(b), (c), (d) and (f).
Therefore, it remains only two cases to be considered,
which were illustrated in Figure~\ref{fig:proofX-1}(a) and (e), respectively.
\begin{enumerate}
\item $\pi(u_{i+1})=1$ and $\pi(u_{i+2})=4$.\\
In that case, we have $\pi(v_{i+1})=2$, and thus $\pi(v_{i-1})=1$, which implies $\pi(v_{i-2})=4$ and thus $\pi(u_{i-1})=2$, so that
$\pi(u_{i-2})=1$, $\pi(v_{i-2})=4$, and no colour is available for $u_{i-3}$ (see Figure~\ref{fig:proofX-1}(g)).

\item $\pi(u_{i+1})=2$, $\pi(u_{i-1})=1$ and $\pi(v_{i+1})=1$.\\
In that case, we necessarily have $\pi(v_{i+2})=4$, so that $\pi(u_{i+2})=1$,
and no colour is available for $u_{i+3}$ (see Figure~\ref{fig:proofX-1}(h)).

\end{enumerate}


\item $\pi(u_i)=4$ and $\pi(v_i)=5$.\\
In that case, we necessarily have $\pi(u_{i+1})\in\{1,2,3\}$.
We consider six subcases, depending on the value of $\pi(u_{i+1})$ and $i$.

\begin{enumerate}
\item $\pi(u_{i+1})=1$ and $i\in\{3,4\}$.\\
In that case, we have $\{\pi(u_{i+2}),\pi(v_{i+1})\}=\{2,3\}$, which implies $\pi(v_{i+2})=1$, and
no colour is available for $v_{i+3}$ (see Figure~\ref{fig:proofX-1}(i)).

\item $\pi(u_{i+1})=1$ and $i=5$.\\
In that case, we have $\pi(v_{6})\in\{2,3\}$.
If $\pi(v_{6})=2$, then we necessarily have $\pi(u_{7})=3$, and thus $\pi(v_{4})\in\{1,3\}$.
If $\pi(v_{4})=1$, we get successively $\pi(u_{4})=2$, $\pi(v_{3})=3$, $\pi(u_{3})=1$,
and no colour is available for $u_{2}$ (see Figure~\ref{fig:proofX-1}(j)).
If $\pi(v_{4})=3$, then $\{\pi(u_{4}),\pi(v_{3})\}=\{1,2\}$ and no colour is available for $u_{3}$ (see Figure~\ref{fig:proofX-1}(k)).

If $\pi(v_{6})=3$, then $\{\pi(u_{4}),\pi(v_{4})\}=\{1,2\}$.
If $\pi(v_{4})=1$ and $\pi(u_{4})=2$, then no colour is available for $v_3$
(see Figure~\ref{fig:proofX-1}(l)).
If $\pi(u_{4})=1$ and $\pi(v_{4})=2$, then
we necessarily have $\pi(v_3)=1$ and $\pi(u_3)=3$, and
no colour is available for $v_{2}$ (see Figure~\ref{fig:proofX-1}(m)).

\item $\pi(u_{i+1})=2$ and $i\in\{3,4\}$.\\
In that case, we necessarily have $\pi(v_{i+1})\in\{1,3\}$.
If $\pi(v_{i+1})=1$, then $\pi(v_{i+2})=3$, which implies $\pi(u_{i+2})=1$, and no colour is available for $u_{i+3}$.
If $\pi(v_{i+1})=3$, then $\pi(u_{i+2})=1$, and no colour is available for $u_{i+3}$
(see Figure~\ref{fig:proofX-1}(n)).

\item $\pi(u_{i+1})=2$ and $i=5$.\\
In that case, we necessarily have $\pi(u_{4})\in\{1,3\}$.
If $\pi(u_{4})=1$, then $\{\pi(u_3),\pi(v_4)\}=\{2,3\}$, so that $\pi(v_3)=1$,
and no colour is available for $v_2$ (see Figure~\ref{fig:proofX-1}(o)).
If $\pi(u_{4})=3$, then 
either $\pi(u_3)=\pi(v_4)=1$,
which implies $\pi(v_3)=2$ and no colour is available for $u_2$,
or
$\{\pi(u_3),\pi(v_4)\}=\{1,2\}$, and no colour
is available for $v_3$ (see Figure~\ref{fig:proofX-1}(p)).

\item $\pi(u_{i+1})=3$ and $i\in\{3,4\}$.\\
In that case, either 
$\pi(v_{i+1})=1$, so that $\pi(v_{i+2})=2$, $\pi(u_{i+2})=1$, and no colour is available for $u_{i+3}$,
or $\pi(v_{i+1})=2$, so that $\pi(v_{i+2})=1$ and no colour is available for $u_{i+2}$
(see Figure~\ref{fig:proofX-1}(q)).

\item $\pi(u_{i+1})=3$ and $i=5$.\\ 
In that case, $\pi(u_{4})\in\{1,2\}$.
If $\pi(u_4)=1$, then $\pi(v_4)=2$ and no colour is available for $u_3$.
If $\pi(u_4)=2$, then $\pi(v_4)=1$, so that $\pi(u_3)=1$ and $\pi(v_3)=3$,
and no colour is available for $u_2$ (see Figure~\ref{fig:proofX-1}(r)).
\end{enumerate}

\end{enumerate}

This completes the proof of Claim~\ref{cl:graphX}.
\end{proof}


\begin{figure}
\begin{center}
\begin{tikzpicture}[domain=-0.5:4.5,x=0.9cm,y=0.9cm]
\ARETEH{1,0}{2}; 
\ARETEH{1,1}{3}; 
\ARETEV{1,0};
\draw[double distance = 2pt] (2,0) -- (2,1);   
\ARETEV{3,0};
\SOMMET{1,0}; \SOMMET{2,0}; \SOMMET{3,0}; 
\SOMMET{1,1}; \SOMMET{2,1}; \SOMMET{3,1}; \SOMMET{4,1}; 
\ETIQUETTEA{1,1}{1};
\ETIQUETTEA{2,1}{2};
\ETIQUETTEA{3,1}{1};
\ETIQUETTEA{4,1}{4,5};
%
\ETIQUETTEB{1,0}{?};
\ETIQUETTEB{2,0}{3};
\ETIQUETTEB{3,0}{4,5};
\ETIQUETTEB{4,0}{};
\node[below] at (2.5,-0.7) {(a)};
\end{tikzpicture}
\hskip 1cm
\begin{tikzpicture}[domain=-0.5:4.5,x=0.9cm,y=0.9cm]
\ARETEH{1,0}{1}; \ARETEH{0,1}{3}; 
\ARETEV{1,0};
\draw[double distance = 2pt] (2,0) -- (2,1);   
\SOMMET{1,0}; \SOMMET{2,0}; 
\SOMMET{0,1}; \SOMMET{1,1}; \SOMMET{2,1}; \SOMMET{3,1}; 
\ETIQUETTEA{0,1}{?};
\ETIQUETTEA{1,1}{1};
\ETIQUETTEA{2,1}{2};
\ETIQUETTEA{3,1}{4};
%
\ETIQUETTEB{0,0}{};
\ETIQUETTEB{1,0}{5};
\ETIQUETTEB{2,0}{3};
\ETIQUETTEB{3,0}{};
\node[below] at (1.5,-0.7) {(b)};
\end{tikzpicture}
\hskip 1cm
\begin{tikzpicture}[domain=-0.5:4.5,x=0.9cm,y=0.9cm]
\ARETEH{1,0}{1}; 
\ARETEH{0,1}{3}; 
\ARETEV{1,0};
\draw[double distance = 2pt] (2,0) -- (2,1);   
\SOMMET{1,0}; \SOMMET{2,0}; 
\SOMMET{0,1}; \SOMMET{1,1}; \SOMMET{2,1}; \SOMMET{3,1}; 
\ETIQUETTEA{0,1}{?};
\ETIQUETTEA{1,1}{1};
\ETIQUETTEA{2,1}{2};
\ETIQUETTEA{3,1}{5};
%
\ETIQUETTEB{0,0}{};
\ETIQUETTEB{1,0}{4};
\ETIQUETTEB{2,0}{3};
\ETIQUETTEB{3,0}{};
\node[below] at (1.5,-0.7) {(c)};
\end{tikzpicture}

\vskip 0.5cm
\begin{tikzpicture}[domain=-0.5:4.5,x=0.9cm,y=0.9cm]
\ARETEH{1,0}{2}; \ARETEH{0,1}{4}; 
\ARETEV{1,0};
\draw[double distance = 2pt] (2,0) -- (2,1);   
\ARETEV{3,0};
\SOMMET{1,0}; \SOMMET{2,0}; \SOMMET{3,0}; 
\SOMMET{0,1}; \SOMMET{1,1}; \SOMMET{2,1}; \SOMMET{3,1}; \SOMMET{4,1}; 
\ETIQUETTEA{0,1}{3|};
\ETIQUETTEA{1,1}{1|?};
\ETIQUETTEA{2,1}{2};
\ETIQUETTEA{3,1}{1};
\ETIQUETTEA{4,1}{3|5};
\ETIQUETTEB{0,0}{};
\ETIQUETTEB{1,0}{?|1};
\ETIQUETTEB{2,0}{4};
\ETIQUETTEB{3,0}{5|3};
\ETIQUETTEB{4,0}{};
\node[below] at (2,-0.7) {(d)};
\end{tikzpicture}
\hskip 1cm
\begin{tikzpicture}[domain=-0.5:4.5,x=0.9cm,y=0.9cm]
\ARETEH{1,0}{1}; 
\ARETEH{0,1}{3}; 
\ARETEV{1,0};
\draw[double distance = 2pt] (2,0) -- (2,1);   
\SOMMET{1,0}; \SOMMET{2,0}; 
\SOMMET{0,1}; \SOMMET{1,1}; \SOMMET{2,1}; \SOMMET{3,1}; 
\ETIQUETTEA{0,1}{?};
\ETIQUETTEA{1,1}{1};
\ETIQUETTEA{2,1}{2};
\ETIQUETTEA{3,1}{3};
%
\ETIQUETTEB{0,0}{};
\ETIQUETTEB{1,0}{5};
\ETIQUETTEB{2,0}{4};
\node[below] at (1.5,-0.7) {(e)};
\end{tikzpicture}
\hskip 1cm
\begin{tikzpicture}[domain=-0.5:4.5,x=0.9cm,y=0.9cm]
\ARETEH{1,0}{1}; 
\ARETEH{0,1}{3}; 
\ARETEV{1,0};
\draw[double distance = 2pt] (2,0) -- (2,1);   
\SOMMET{1,0}; \SOMMET{2,0}; 
\SOMMET{0,1}; \SOMMET{1,1}; \SOMMET{2,1}; \SOMMET{3,1}; 
\ETIQUETTEA{0,1}{?};
\ETIQUETTEA{1,1}{1};
\ETIQUETTEA{2,1}{2};
\ETIQUETTEA{3,1}{5};
%
\ETIQUETTEB{0,0}{};
\ETIQUETTEB{1,0}{3};
\ETIQUETTEB{2,0}{4};
\ETIQUETTEB{3,0}{};
\node[below] at (1.5,-0.7) {(f)};
\end{tikzpicture}

\caption{Configurations for the proof of Lemma~\ref{lem:graphX} (cont.).}
\label{fig:proofX-2}
\end{center}
\end{figure}
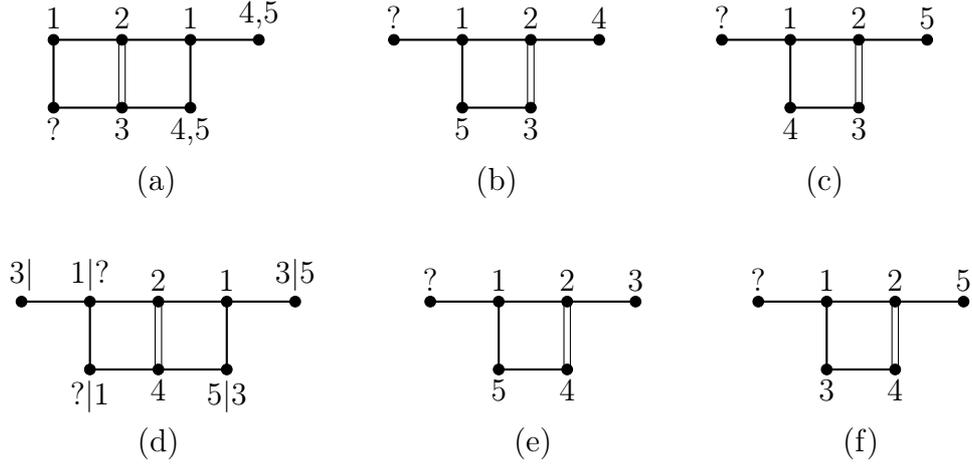

By Claim~\ref{cl:graphX}, we can thus assume $\pi(u_i)=2$, without loss of generality (again, thanks to the symmetry exchanging $u_i$ and $v_i$),
so that $\pi(v_i)\in\{3,4,5\}$.
To finish the proof of Lemma~\ref{lem:graphX}, we need to prove that $\{\pi(u_{i-1}),\pi(u_{i+1})\}=\{3,4,5\}\setminus\{\pi(v_i)\}$.
Suppose that this is not the case.
We consider the following cases, according to the value of $\pi(v_i)$.

\begin{enumerate}
\item $\pi(v_i)=3$.\\
In that case, we necessarily have $\pi(u_{i+1})\in\{1,4,5\}$.

If $\pi(u_{i+1})=1$, then $\{\pi(u_{i+2}),\pi(v_{i+1})\}=\{4,5\}$, so that $\pi(u_{i-1})=1$, and no colour is available for $v_{i-1}$
(see Figure~\ref{fig:proofX-2}(a)).

If $\pi(u_{i+1})=4$, then either 
$\pi(u_{i-1})=1$, so that $\pi(v_{i-1})=5$, and no colour is available for $u_{i-2}$
(see Figure~\ref{fig:proofX-2}(b)), 
or $\pi(u_{i-1})=5$, which contradicts our assumption since it
would imply $\{\pi(u_{i-1}),\pi(u_{i+1})\}=\{3,4,5\}\setminus\{\pi(v_i)\}$.

Similarly, if $\pi(u_{i+1})=5$, then either 
$\pi(u_{i-1})=1$, so that $\pi(v_{i-1})=4$, and no colour is available for $u_{i-2}$
(see Figure~\ref{fig:proofX-2}(c)), 
or $\pi(u_{i-1})=4$, which again contradicts our assumption.

\item $\pi(v_i)=4$ (the case $\pi(v_i)=5$ is similar, by switching colours $4$ and $5$).\\
In that case, we necessarily have $\pi(u_{i+1})\in\{1,3,5\}$.

If $\pi(u_{i+1})=1$, then $\{\pi(u_{i+2}),\pi(v_{i+1})\}=\{3,5\}$.
If $\pi(u_{i+2})=3$ and $\pi(v_{i+1})=5$, then $\pi(u_{i-1})=1$, so that $\pi(u_{i-2})=3$, and no colour is available for $v_{i-1}$.
If $\pi(u_{i+2})=5$ and $\pi(v_{i+1})=3$, then $\pi(v_{i-1})=1$, and no colour is available for $u_{i-1}$
(see Figure~\ref{fig:proofX-2}(d)).

If $\pi(u_{i+1})=3$, then either  
$\pi(u_{i-1})=1$, so that $\pi(v_{i-1})=5$, and no colour is available for $u_{i-2}$,
or $\pi(u_{i-1})=5$, which contradicts our assumption (see Figure~\ref{fig:proofX-2}(e)).

Finally, if $\pi(u_{i+1})=5$, then either 
$\pi(u_{i-1})=1$, so that $\pi(v_{i-1})=3$, and no colour is available for $u_{i-2}$,
or $\pi(u_{i-1})=3$, which contradicts our assumption (see Figure~\ref{fig:proofX-2}(f)).
\end{enumerate}

This completes the proof of Lemma~\ref{lem:graphX}.


\section{Proof of Lemma~\ref{lem:1altH}}
\label{ap:1altH}


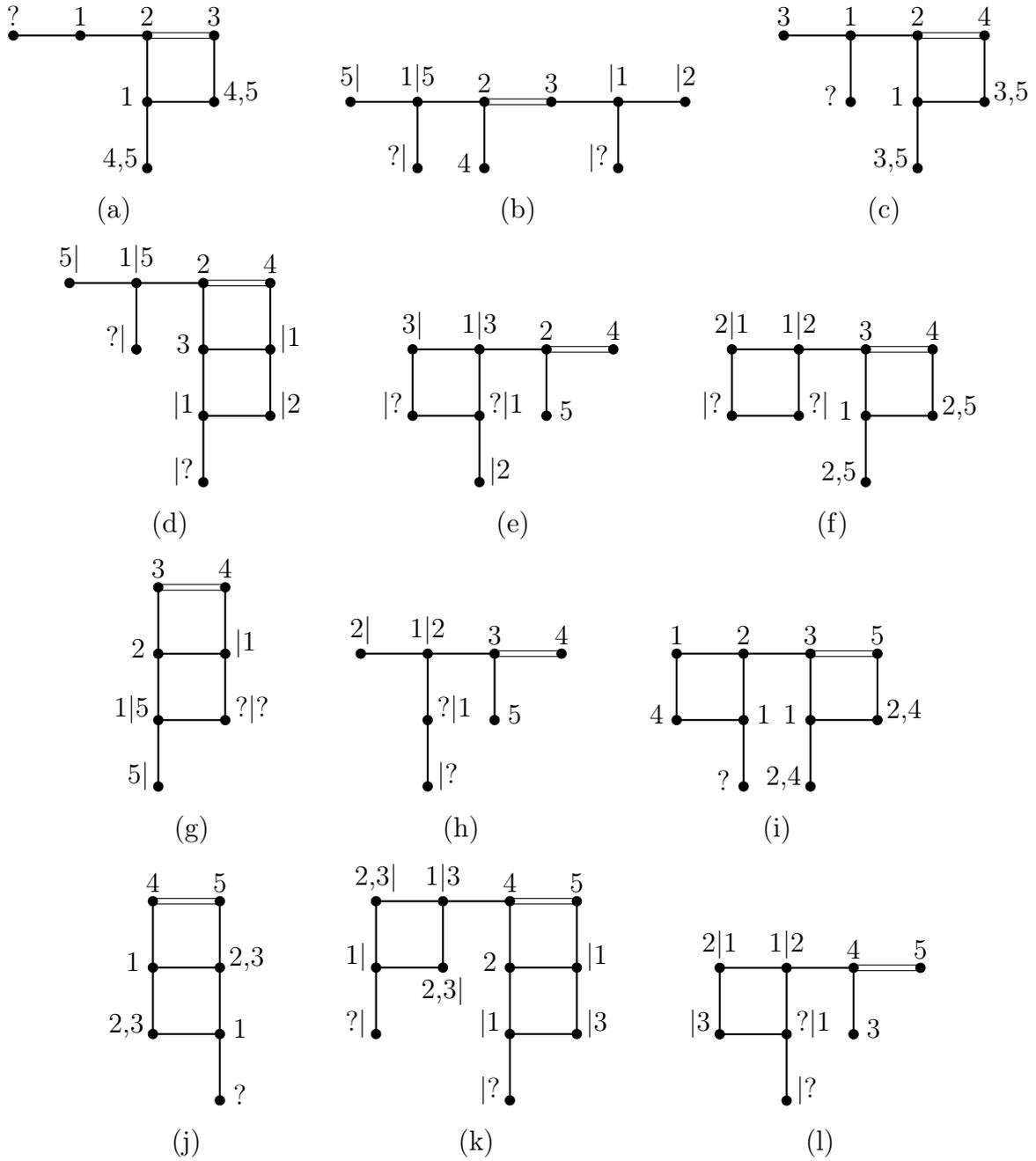
\begin{figure}
\begin{center}
\begin{tikzpicture}[domain=0:12,x=1cm,y=1cm]
\draw[double distance = 2pt] (2,0) -- (3,0);  
\ARETEH{0,0}{2}; 
\ARETEH{2,-1}{1};
\ARETEVV{2,-2}{2}; \ARETEVV{3,-1}{1};
\SOMMET{0,0}; \SOMMET{1,0}; \SOMMET{2,0}; \SOMMET{3,0};
\SOMMET{2,-1}; \SOMMET{3,-1}; 
\SOMMET{2,-2};
\ETIQUETTEA{0,0}{?};
\ETIQUETTEA{1,0}{1};
\ETIQUETTEA{2,0}{2};
\ETIQUETTEA{3,0}{3};
\ETIQUETTEA{1.7,-1.2}{1}; \ETIQUETTEA{3.4,-1.2}{4,5};
\ETIQUETTEA{1.6,-2.2}{4,5};
\ETIQUETTEA{1.5,-3}{(a)};
\end{tikzpicture}
\hskip 0.8cm
\begin{tikzpicture}[domain=0:12,x=1cm,y=1cm]
\draw[double distance = 2pt] (2,0) -- (3,0);  
\ARETEH{0,0}{2}; \ARETEH{3,0}{2}; 
\ARETEVV{1,-1}{1}; 
\ARETEVV{2,-1}{1};
\ARETEVV{4,-1}{1};
\SOMMET{0,0}; \SOMMET{1,0}; \SOMMET{2,0}; \SOMMET{3,0};  \SOMMET{4,0}; \SOMMET{5,0};
\SOMMET{1,-1}; \SOMMET{2,-1}; \SOMMET{4,-1};
\ETIQUETTEA{0,0}{5|};
\ETIQUETTEA{1,0}{1|5};
\ETIQUETTEA{2,0}{2};
\ETIQUETTEA{3,0}{3};
\ETIQUETTEA{4,0}{|1};
\ETIQUETTEA{5,0}{|2};
\ETIQUETTEA{1.7,-1.2}{4}; 
\ETIQUETTEA{0.7,-1.2}{?|};
\ETIQUETTEA{3.7,-1.2}{|?};
\ETIQUETTEA{2.5,-2}{(b)};
\end{tikzpicture}
\hskip 0.8cm
\begin{tikzpicture}[domain=0:12,x=1cm,y=1cm]
\draw[double distance = 2pt] (2,0) -- (3,0);  
\ARETEH{0,0}{2}; 
\ARETEH{2,-1}{1};
\ARETEVV{2,-2}{2}; \ARETEVV{1,-1}{1}; \ARETEVV{3,-1}{1};
\SOMMET{0,0}; \SOMMET{1,0}; \SOMMET{2,0}; \SOMMET{3,0};
\SOMMET{1,-1}; \SOMMET{2,-1}; \SOMMET{3,-1}; 
\SOMMET{2,-2};
\ETIQUETTEA{0,0}{3};
\ETIQUETTEA{1,0}{1};
\ETIQUETTEA{2,0}{2};
\ETIQUETTEA{3,0}{4};
\ETIQUETTEA{0.7,-1.2}{?}; 
\ETIQUETTEA{1.7,-1.2}{1}; \ETIQUETTEA{3.4,-1.2}{3,5};
\ETIQUETTEA{1.6,-2.2}{3,5};
\ETIQUETTEA{1.5,-3}{(c)};
\end{tikzpicture}

\begin{tikzpicture}[domain=0:12,x=1cm,y=1cm]
\draw[double distance = 2pt] (2,0) -- (3,0);  
\ARETEH{0,0}{2}; 
\ARETEH{2,-1}{1}; \ARETEH{2,-2}{1}; 
\ARETEVV{1,-1}{1}; 
\ARETEVV{2,-3}{3};
\ARETEVV{3,-2}{2};
\SOMMET{0,0}; \SOMMET{1,0}; \SOMMET{2,0}; \SOMMET{3,0};  
\SOMMET{1,-1}; \SOMMET{2,-1}; \SOMMET{3,-1};
\SOMMET{2,-2}; \SOMMET{3,-2};
\SOMMET{2,-3};
\ETIQUETTEA{0,0}{5|};
\ETIQUETTEA{1,0}{1|5};
\ETIQUETTEA{2,0}{2};
\ETIQUETTEA{3,0}{4};
\ETIQUETTEA{1.7,-1.2}{3}; 
\ETIQUETTEA{0.7,-1.2}{?|};
\ETIQUETTEA{3.3,-1.2}{|1};
\ETIQUETTEA{1.7,-2.2}{|1}; 
\ETIQUETTEA{3.3,-2.2}{|2};
\ETIQUETTEA{1.7,-3.2}{|?}; 
\ETIQUETTEA{1.5,-4}{(d)};
\end{tikzpicture}
\hskip 0.8cm
\begin{tikzpicture}[domain=0:12,x=1cm,y=1cm]
\draw[double distance = 2pt] (2,0) -- (3,0);  
\ARETEH{0,0}{2}; 
\ARETEH{0,-1}{1}; 
\ARETEVV{0,-1}{1}; 
\ARETEVV{1,-2}{2}; 
\ARETEVV{2,-1}{1};
\SOMMET{0,0}; \SOMMET{1,0}; \SOMMET{2,0}; \SOMMET{3,0};  
\SOMMET{0,-1}; \SOMMET{1,-1}; \SOMMET{2,-1}; 
\SOMMET{1,-2}; 
\ETIQUETTEA{0,0}{3|};
\ETIQUETTEA{1,0}{1|3};
\ETIQUETTEA{2,0}{2};
\ETIQUETTEA{3,0}{4};
\ETIQUETTEA{-0.3,-1.2}{|?}; 
\ETIQUETTEA{1.4,-1.2}{?|1};
\ETIQUETTEA{2.3,-1.2}{5};
\ETIQUETTEA{1.3,-2.2}{|2}; 
\ETIQUETTEA{1.5,-3}{(e)};
\end{tikzpicture}
\hskip 0.8cm
\begin{tikzpicture}[domain=0:12,x=1cm,y=1cm]
\draw[double distance = 2pt] (2,0) -- (3,0);  
\ARETEH{0,0}{2}; 
\ARETEH{0,-1}{1};
\ARETEH{2,-1}{1};
\ARETEVV{2,-2}{2}; \ARETEVV{0,-1}{1}; \ARETEVV{1,-1}{1}; \ARETEVV{3,-1}{1};
\SOMMET{0,0}; \SOMMET{1,0}; \SOMMET{2,0}; \SOMMET{3,0};
\SOMMET{0,-1}; \SOMMET{1,-1}; \SOMMET{2,-1}; \SOMMET{3,-1}; 
\SOMMET{2,-2};
\ETIQUETTEA{0,0}{2|1};
\ETIQUETTEA{1,0}{1|2};
\ETIQUETTEA{2,0}{3};
\ETIQUETTEA{3,0}{4};
\ETIQUETTEA{-0.3,-1.2}{|?}; 
\ETIQUETTEA{1.3,-1.2}{?|}; 
\ETIQUETTEA{1.7,-1.2}{1}; 
\ETIQUETTEA{3.4,-1.2}{2,5};
\ETIQUETTEA{1.6,-2.2}{2,5};
\ETIQUETTEA{1.5,-3}{(f)};
\end{tikzpicture}

\begin{tikzpicture}[domain=0:12,x=1cm,y=1cm]
\draw[double distance = 2pt] (2,0) -- (3,0);  
\ARETEH{2,-1}{1}; \ARETEH{2,-2}{1}; 
\ARETEVV{2,-3}{3};
\ARETEVV{3,-2}{2};
\SOMMET{2,0}; \SOMMET{3,0};  
\SOMMET{2,-1}; \SOMMET{3,-1};
\SOMMET{2,-2}; \SOMMET{3,-2};
\SOMMET{2,-3};
\ETIQUETTEA{2,0}{3};
\ETIQUETTEA{3,0}{4};
\ETIQUETTEA{1.7,-1.2}{2}; 
\ETIQUETTEA{3.3,-1.2}{|1};
\ETIQUETTEA{1.6,-2.2}{1|5}; 
\ETIQUETTEA{3.4,-2.2}{?|?};
\ETIQUETTEA{1.7,-3.2}{5|}; 
\ETIQUETTEA{2.5,-4}{(g)};
\end{tikzpicture}
\hskip 0.8cm
\begin{tikzpicture}[domain=0:12,x=1cm,y=1cm]
\draw[double distance = 2pt] (2,0) -- (3,0);  
\ARETEH{0,0}{2}; 
\ARETEVV{1,-2}{2}; 
\ARETEVV{2,-1}{1};
\SOMMET{0,0}; \SOMMET{1,0}; \SOMMET{2,0}; \SOMMET{3,0};  
\SOMMET{1,-1}; \SOMMET{2,-1}; 
\SOMMET{1,-2}; 
\ETIQUETTEA{0,0}{2|};
\ETIQUETTEA{1,0}{1|2};
\ETIQUETTEA{2,0}{3};
\ETIQUETTEA{3,0}{4};
\ETIQUETTEA{1.4,-1.2}{?|1};
\ETIQUETTEA{2.3,-1.2}{5};
\ETIQUETTEA{1.3,-2.2}{|?}; 
\ETIQUETTEA{1.5,-3}{(h)};
\end{tikzpicture}
\hskip 0.8cm
\begin{tikzpicture}[domain=0:12,x=1cm,y=1cm]
\draw[double distance = 2pt] (2,0) -- (3,0);  
\ARETEH{0,0}{2}; 
\ARETEH{0,-1}{1};
\ARETEH{2,-1}{1};
\ARETEVV{2,-2}{2}; \ARETEVV{0,-1}{1}; \ARETEVV{1,-2}{2}; \ARETEVV{3,-1}{1};
\SOMMET{0,0}; \SOMMET{1,0}; \SOMMET{2,0}; \SOMMET{3,0};
\SOMMET{0,-1}; \SOMMET{1,-1}; \SOMMET{2,-1}; \SOMMET{3,-1}; 
\SOMMET{1,-2}; \SOMMET{2,-2};
\ETIQUETTEA{0,0}{1};
\ETIQUETTEA{1,0}{2};
\ETIQUETTEA{2,0}{3};
\ETIQUETTEA{3,0}{5};
\ETIQUETTEA{-0.3,-1.2}{4}; 
\ETIQUETTEA{1.3,-1.2}{1}; 
\ETIQUETTEA{1.7,-1.2}{1}; 
\ETIQUETTEA{3.4,-1.2}{2,4};
\ETIQUETTEA{0.7,-2.2}{?}; 
\ETIQUETTEA{1.6,-2.2}{2,4};
\ETIQUETTEA{1.5,-3}{(i)};
\end{tikzpicture}

\begin{tikzpicture}[domain=0:12,x=1cm,y=1cm]
\draw[double distance = 2pt] (2,0) -- (3,0);  
\ARETEH{2,-1}{1}; \ARETEH{2,-2}{1}; 
\ARETEVV{2,-2}{2};
\ARETEVV{3,-3}{3};
\SOMMET{2,0}; \SOMMET{3,0};  
\SOMMET{2,-1}; \SOMMET{3,-1};
\SOMMET{2,-2}; \SOMMET{3,-2};
\SOMMET{3,-3};
\ETIQUETTEA{2,0}{4};
\ETIQUETTEA{3,0}{5};
\ETIQUETTEA{1.7,-1.2}{1}; 
\ETIQUETTEA{3.4,-1.2}{2,3};
\ETIQUETTEA{1.6,-2.2}{2,3}; 
\ETIQUETTEA{3.3,-2.2}{1};
\ETIQUETTEA{3.3,-3.2}{?}; 
\ETIQUETTEA{2.5,-4}{(j)};
\end{tikzpicture}
\hskip 0.8cm
\begin{tikzpicture}[domain=0:12,x=1cm,y=1cm]
\draw[double distance = 2pt] (2,0) -- (3,0);  
\ARETEH{0,0}{2};  
\ARETEH{0,-1}{1}; \ARETEH{2,-1}{1}; \ARETEH{2,-2}{1};
\ARETEVV{0,-2}{2}; \ARETEVV{1,-1}{1}; \ARETEVV{2,-3}{3}; \ARETEVV{3,-2}{2};
\SOMMET{0,0}; \SOMMET{1,0}; \SOMMET{2,0}; \SOMMET{3,0};
\SOMMET{0,-1}; \SOMMET{1,-1}; \SOMMET{2,-1}; \SOMMET{3,-1}; 
\SOMMET{0,-2}; \SOMMET{2,-2}; \SOMMET{3,-2};
\SOMMET{2,-3};
\ETIQUETTEA{0,0}{2,3|};
\ETIQUETTEA{1,0}{1|3};
\ETIQUETTEA{2,0}{4};
\ETIQUETTEA{3,0}{5};
\ETIQUETTEA{-0.3,-1.2}{1|}; 
\ETIQUETTEB{1,-1}{2,3|}; 
\ETIQUETTEA{1.7,-1.2}{2}; 
\ETIQUETTEA{3.3,-1.2}{|1};
\ETIQUETTEA{-0.3,-2.2}{?|};
\ETIQUETTEA{1.7,-2.2}{|1};
\ETIQUETTEA{3.3,-2.2}{|3};
\ETIQUETTEA{1.7,-3.2}{|?};
\ETIQUETTEA{1.5,-4}{(k)};
\end{tikzpicture}
\hskip 0.8cm
\begin{tikzpicture}[domain=0:12,x=1cm,y=1cm]
\draw[double distance = 2pt] (2,0) -- (3,0);  
\ARETEH{0,0}{2}; 
\ARETEH{0,-1}{1}; 
\ARETEVV{0,-1}{1}; 
\ARETEVV{1,-2}{2}; 
\ARETEVV{2,-1}{1};
\SOMMET{0,0}; \SOMMET{1,0}; \SOMMET{2,0}; \SOMMET{3,0};  
\SOMMET{0,-1}; \SOMMET{1,-1}; \SOMMET{2,-1}; 
\SOMMET{1,-2}; 
\ETIQUETTEA{0,0}{2|1};
\ETIQUETTEA{1,0}{1|2};
\ETIQUETTEA{2,0}{4};
\ETIQUETTEA{3,0}{5};
\ETIQUETTEA{-0.3,-1.2}{|3}; 
\ETIQUETTEA{1.4,-1.2}{?|1};
\ETIQUETTEA{2.3,-1.2}{3};
\ETIQUETTEA{1.3,-2.2}{|?}; 
\ETIQUETTEA{1.5,-3}{(l)};
\end{tikzpicture}
\caption{Configurations for the proof of Claim~\ref{cl:1altH-1} (the double edge is the edge $u^0_2u^0_3$).}
\label{fig:1altH-1}
\end{center}
\end{figure}


We first prove the following claim.

\begin{claim}
For every integer $j$, $0\le j< r$, either $\pi(u^0_{2j})=1$ or $\pi(u^0_{2j+1})=1$.
\label{cl:1altH-1}
\end{claim}

\begin{proof}
Thanks to the symmetries of $H^\ell(r)$,
it is enough to prove the claim for the edge $u^0_2u^0_3$.
Suppose to the contrary that $\pi(u^0_2)\neq 1$ and $\pi(u^0_3)\neq 1$.
Thanks to the symmetries of $H^\ell(r)$, we can assume $\pi(u^0_2)<\pi(u^0_3)$, without loss of generality.

We consider four cases. The corresponding configurations are depicted in Figure~\ref{fig:1altH-1}, using the 
same drawing convention as for the proof of Lemma~\ref{lem:graphX} (see Appendix~\ref{ap:graphX}).

\begin{enumerate}
\item $\pi(u^0_2)=2$ and $\pi(u^0_3)=3$.\\
In that case, $\pi(u^1_2)\in\{1,4,5\}$.
If $\pi(u^1_2)=1$, then $\{\pi(u^2_2),\pi(u^1_3)\}=\{4,5\}$, which implies $\pi(u^0_1)=1$,
and no colour is available for $u^0_0$ (see Figure~\ref{fig:1altH-1}(a)).
If $\pi(u^1_2)=4$, then either $\pi(u^0_1)=1$, which implies $\pi(u^0_0)=5$, and no colour is available for $u^1_1$,
or $\pi(u^0_1)=5$, which implies $\pi(u^0_4)=1$, $\pi(u^0_5)=2$, 
and no colour is available for $u^1_4$ (see Figure~\ref{fig:1altH-1}(b)).
The case $\pi(u^1_2)=5$ is similar, by switching colours 4 and~5.

\item $\pi(u^0_2)=2$ and $\pi(u^0_3)=4$ (the case $\pi(u^0_2)=2$ and $\pi(u^0_3)=5$ is
similar, by switching colours $4$ and~$5$).\\
In that case, $\pi(u^1_2)\in\{1,3,5\}$.
If $\pi(u^1_2)=1$, then $\{\pi(u^2_2),\pi(u^1_3)\}=\{3,5\}$, which implies $\pi(u^0_1)=1$,
$\pi(u^0_0)=3$, and no colour is available for $u^1_1$ (see Figure~\ref{fig:1altH-1}(c)).
If $\pi(u^1_2)=3$, then either $\pi(u^0_1)=1$, which implies $\pi(u^0_0)=5$,
and no colour is available for $u^1_1$,
or $\pi(u^0_1)=5$, which implies $\pi(u^2_2)=\pi(u^1_3)=1$, so that $\pi(u^2_3)=2$,
and no colour is available for $u^3_2$ (see Figure~\ref{fig:1altH-1}(d)).
Finally, if $\pi(u^1_2)=5$, then either $\pi(u^0_1)=1$, which implies $\pi(u^0_0)=3$,
and no colour is available for $u^1_1$,
or $\pi(u^0_1)=3$, which implies $\pi(u^1_1)=1$, $\pi(u^2_1)=2$,
and no colour is available for $u^1_0$  (see Figure~\ref{fig:1altH-1}(e)).
  
\item $\pi(u^0_2)=3$ and $\pi(u^0_3)=4$.\\
In that case, $\pi(u^1_2)\in\{1,2,5\}$.
If $\pi(u^1_2)=1$, then $\{\pi(u^2_2),\pi(u^1_3)\}=\{2,5\}$, and thus
either $\pi(u^0_1)=1$, so that $\pi(u^0_0)=2$, and no colour is available for $u^1_1$,
or $\pi(u^0_1)=2$, so that $\pi(u^0_0)=1$, and no colour is available for $u^1_0$ (see Figure~\ref{fig:1altH-1}(f)).
If $\pi(u^1_2)=2$, then either
$\pi(u^2_2)=1$, which implies $\pi(u^3_2)=5$, and no colour is available for $u^2_3$,
or $\pi(u^2_2)=5$, which implies $\pi(u^1_3)=1$, and no colour is available for $u^2_3$ (see Figure~\ref{fig:1altH-1}(g)).
Finally, if $\pi(u^1_2)=5$, then either
$\pi(u^0_1)=1$, which implies $\pi(u^0_0)=2$, and no colour is available for $u^1_1$,
or $\pi(u^0_1)=2$, which implies $\pi(u^1_1)=1$, and no colour is available for $u^2_1$ (see Figure~\ref{fig:1altH-1}(h)).

\item $\pi(u^0_2)=3$ and $\pi(u^0_3)=5$.\\
This case is similar to the previous one, by switching colours 4 and~5, except
when $\pi(u^1_2)=1$ (which implies $\{\pi(u^2_2),\pi(u^1_3)\}=\{2,4\}$) and $\pi(u^0_1)=2$.
In that case, we necessarily have $\pi(u^0_0)=\pi(u^1_1)=1$, which implies $\pi(u^1_0)=4$,
and no colour is available for $u^2_1$ (see Figure~\ref{fig:1altH-1}(i)).

\item $\pi(u^0_2)=4$ and $\pi(u^0_3)=5$.\\
In that case, $\pi(u^1_2)\in\{1,2,3\}$.
If $\pi(u^1_2)=1$, then $\{\pi(u^2_2),\pi(u^1_3)\}=\{2,3\}$, which implies $\pi(u^2_3)=1$,
and no colour is available for $u^3_3$ (see Figure~\ref{fig:1altH-1}(j)).
If $\pi(u^1_2)=2$, then either
$\pi(u^0_1)=1$, which implies $\{\pi(u^0_0),\pi(u^1_1)\}=\{2,3\}$,
so that $\pi(u^1_0)=1$, and no colour is available for $u^2_0$,
or $\pi(u^0_1)=3$, which implies $\pi(u^2_2)=\pi(u^1_3)=1$, 
so that $\pi(u^2_3)=3$, and no colour is available for $u^3_2$ (see Figure~\ref{fig:1altH-1}(k)).
Finally, if $\pi(u^1_2)=3$, then either
$\pi(u^0_1)=1$, which implies $\pi(u^0_0)=2$, and no colour is available for $u^1_1$,
or $\pi(u^0_1)=2$, which implies $\pi(u^0_0)=\pi(u^1_1)=1$, so that $\pi(u^1_0)=3$,
and no colour is available for $u^2_1$ (see Figure~\ref{fig:1altH-1}(l)).
\end{enumerate}
This completes the proof of Claim~\ref{cl:1altH-1}.
\end{proof}

Since the cycle induced by the set of vertices $\{u^0_0,u^0_1,\dots,u^0_{2r-1}\}$
has even length, and adjacent vertices cannot be assigned the same colour,
it follows from Claim~\ref{cl:1altH-1} that colour~1 must be used
on each edge $u^0_ju^0_{j+1}$, $0\le j\le 2r-1$
(subscripts are taken modulo $2r$).
By symmetry, colour~1 must also be used 
on each edge $u^{\ell+1}_ju^{\ell+1}_{j+1}$, $0\le j\le 2r-1$.
This concludes the proof of Lemma~\ref{lem:1altH}.

\end{document}